\documentclass[11pt]{amsart}
\usepackage[
  left=2.5cm,
  right=2.5cm,
  top=3cm,
  bottom=3cm
]{geometry}
\usepackage[utf8]{inputenc}
\usepackage[T1]{fontenc}
\usepackage[british]{babel}
\usepackage[numbers,sort&compress]{natbib}

\usepackage{amsmath,amssymb,amsfonts,amsthm}

\def\smallsection#1{\smallskip\noindent\textbf{#1}.}
\usepackage{mathtools}
\usepackage{mathrsfs}
\usepackage{bbm}
\usepackage{dsfont}
\usepackage{stmaryrd}
\usepackage{bm}
\usepackage{braket}

\usepackage{graphicx}
\usepackage{pgf,tikz}
\usetikzlibrary{positioning,fit,calc,arrows.meta}

\usepackage{paralist}
\usepackage{comment}
\usepackage{soul}
\usepackage{xcolor}
\usepackage{framed}

\usepackage{natbib}

\usepackage{hyperref}
\usepackage{cleveref}

\definecolor{navyblue}{rgb}{0.0,0.0,0.5}
\hypersetup{
  colorlinks=true,
  linkcolor=black,
  citecolor=blue,
  urlcolor=navyblue,
  anchorcolor=blue
}

\newcommand{\nn}{\nonumber}
\newcommand{\1}{\mathbbm{1}}
\newcommand{\Tr}{\operatorname{Tr}}

\newcommand{\ceil}[1]{\left\lceil {#1} \right\rceil}

\newcommand{\kb}[1]{|#1\rangle\!\langle#1|}

\newtheorem{conditionA}{Condition}
\renewcommand{\theconditionA}{A}

\newtheorem{conditionB}{Condition}

\newcommand{\R}{\mathbb{R}}
\newcommand{\C}{\mathbb{C}}
\newcommand{\N}{\mathbb{N}}

\newcommand{\cH}{\mathcal{H}}

\newcommand{\cL}{\mathcal{L}}

\newcommand{\E}{\mathbb{E}}

\newcommand{\spec}{\operatorname{Sp}}

\newcommand{\eps}{\varepsilon}
\newcommand{\NA}{N_{\mathcal{A}}}
\renewcommand{\epsilon}{\varepsilon}
\renewcommand{\phi}{\varphi}

\newcommand{\cA}{\mathcal{A}}
\newcommand{\hatfM}{\widehat{f}_{\!\scalebox{0.6}{$\mathscr{M}$}}
}

\newcommand{\fM}{f_{\!\scalebox{0.6}{$\mathscr{M}$}}
}

\newcommand{\Ntot}{N_{\operatorname{tot}}}
\newcommand{\bin}[1]{}

\theoremstyle{plain}
\newtheorem{theorem}{Theorem}[section]
\newtheorem{lemma}[theorem]{Lemma}
\newtheorem{proposition}[theorem]{Proposition}
\newtheorem{corollary}[theorem]{Corollary}

\newtheorem{remark}[theorem]{Remark}

\theoremstyle{definition}
\newtheorem{definition}[theorem]{Definition}

\newtheorem{condition}[theorem]{Condition}

\theoremstyle{remark}

\newtheorem{claim}{Claim}[theorem]

\numberwithin{equation}{section}

\AtBeginDocument{\setlength{\emergencystretch}{4em}}

\begin{document}

\title[Infinite-dimensional Quantum Gibbs sampling]{Quantum Gibbs Sampling in infinite dimensions:\\
\tiny{Generation, mixing times and circuit implementation}}

\author{Simon Becker}
\address{Bocconi University, Milan, Italy}
\email{simon.becker@unibocconi.it}

\author{Cambyse Rouz\'e}
\address{Inria, T\'el\'ecom Paris -- LTCI,  
Institut Polytechnique de Paris,  
91120 Palaiseau, France}
\email{cambyse.rouze@inria.fr}

\author{Robert Salzmann}
\address{Institute for Quantum Information, RWTH Aachen University, Germany \newline
Inria, ENS de Lyon, Université Claude Bernard Lyon 1, LIP, 69342, Lyon cedex 07, France}
\email{robert.salzmann@rwth-aachen.de}

\maketitle
\vspace{-0.3cm}

\begin{abstract}
\noindent 
We develop a rigorous and implementable framework for Gibbs sampling of infinite-dimensional quantum systems governed by unbounded Hamiltonians. Extending dissipative Gibbs samplers beyond finite dimensions raises fundamental obstacles, including ill-defined generators, the absence of spectral gaps on natural Banach spaces, and tensions between implementability and convergence guarantees.
We overcome these issues by constructing KMS-symmetric quantum Markov semigroups on separable Hilbert spaces that are both well-posed and efficiently implementable on qubit hardware. Our generation theory is based on the abstract framework of Dirichlet forms, adapted here to the case of algebras of bounded operators over separable Hilbert spaces.
Leveraging the spectral properties of our self-adjoint generators, we establish quantitative convergence results in trace distance, including regimes of fast thermalization.  In contrast, we also identify Hamiltonians for which a naive choice of generators guaranteeing implementability generally comes at the cost of losing   uniform exponential convergence guarantees for  the associated evolutions, thereby establishing a strong trade-off between implementability and convergence.
 Our framework applies to a wide class of models—including  Schrödinger operators, Gaussian systems, and Bose–Hubbard  Hamiltonians—and provides a unified approach linking rigorous  infinite-dimensional analysis with algorithmic Gibbs state  preparation.

\end{abstract}
\vspace{-0.3cm}
\tableofcontents

\section{Introduction}

Davies semigroups \cite{Davies1974} are among the most prominently used models for the thermalization of quantum systems into their Gibbs states. Formally, for a system at inverse temperature $\beta$ with an associated finite-dimensional Hilbert space $\mathcal{H}$ and Hamiltonian $H$, and a finite set of so-called \emph{bare jump} operators $\{A^\alpha\}_{\alpha\in\mathcal{A}}$ on $\mathcal{H}$ that is closed under taking adjoints, i.e.~$\{A^\alpha\}_{\alpha\in\mathcal{A}}=\{(A^\alpha)^\dagger\}_{\alpha\in\mathcal{A}}$, with trivial commutant, the generator in the Schr\"{o}dinger picture is defined as
\begin{align}\label{eq:Davies}
\mathcal{L}_{\operatorname{D}}(\rho)
&\!= \!\!\!\!\!\sum_{\substack{\alpha\in\mathcal{A}\\\omega\in B(H)}} \!\!\!\!\Upsilon(\omega)
   \Big( A^\alpha_\omega \rho (A^\alpha_\omega)^\dagger
   - \tfrac{1}{2}\{(A^\alpha_\omega)^\dagger A^\alpha_\omega, \rho\}    \Big).
\end{align}
Here $B(H)=\operatorname{Sp}(H)-\operatorname{Sp}(H)$ denotes the set of Bohr frequencies of $H$. Given the spectral decomposition $H=\sum_{E\in\operatorname{Sp}(H)}EP_E$, the jump operators take the form
\begin{align*}
A^\alpha_\omega
&:= \sum_{\substack{E,E'\in\operatorname{Sp}(H)\\ E-E'=\omega}}
    P_{E} \, A^\alpha \, P_{E'}.
\end{align*}
The function $\Upsilon:\mathbb{R}\to\mathbb{R}_+$ encodes the rate at which each jump occurs, and in order for the evolution to fix the Gibbs state $\sigma_\beta := \frac{e^{-\beta H}}{\Tr(e^{-\beta H})}$, it satisfies the KMS symmetry condition $\Upsilon(-\omega) = e^{\beta \omega}\,\Upsilon(\omega)$. A well-known difficulty with Davies generators, already present in finite dimensions, is that the jump operators depend on the generally unknown spectral decomposition of $H$. This dependence renders both their circuit implementation and general proofs of convergence for the resulting dynamics particularly delicate. Recently, several alternatives to Davies generators have been proposed whose definitions avoid the spectral decomposition of $H$, making them appealing for quantum simulation tasks \cite{wocjan2023szegedy,rall2023thermal,chen2023quantum}. Previous methods generally achieved only approximate preparation of the target Gibbs state, and only under restrictive assumptions. By contrast, \cite{chen2023efficient} provided the first exact sampler based on Lindbladian dynamics. This result was later complemented by simpler Lindbladian constructions \cite{ding2025efficient,gilyen2024quantum} that require only finitely many jumps.
In \cite{ding2025efficient}, the jump operators are defined as matrix-valued integrals of the form
\begin{align}
L^\alpha
&:= \int e^{itH} A^\alpha e^{-itH}\, {f}(t)\, dt = \sum_{E,E'\in\operatorname{Sp}(H)} \,\widehat{f}(E{-}E')\, P_E A^\alpha P_{E'} = \sum_{\nu\in B(H)}\widehat{f}(\nu)\, A^\alpha_\nu,\label{def:Lalpha}
\end{align}
with a smooth and sufficiently fast decaying \textit{filter function} $f\in L^1(\mathbb{R})$. Note that we take the slightly unconventional definition $\widehat{f}(\nu) =\int_{\mathbb R} f(t) e^{i\nu t} \ dt$ from \cite{ding2025efficient} above. The integral formulation \eqref{def:Lalpha} permits implementation via oracle access to block encodings of the Hamiltonian evolution and bare jumps $A^\alpha$, after time-discretization. The associated Lindblad generator takes the GKLS form

\vspace{- 0.1 in}

\begin{align}\label{eq:GKLS}
\!\!\!\mathcal{L}_{\widehat{f},H}(\rho)
\!&=\! \!-i[B,\rho]
 + \!\!\sum_{\alpha\in\mathcal{A}}\!\! \Big( L^\alpha  \rho(L^\alpha)^\dagger
 \!-\! \tfrac{1}{2}\{(L^\alpha)^\dagger L^\alpha, \rho\} \Big),
\end{align}
where the Hermitian operator $B$ is carefully chosen so that the Gibbs state remains a fixed point of the evolution, $\mathcal{L}_{\widehat{f},H}(\sigma_\beta)=0$.

It was recently shown that the convergence properties of the evolution generated by~$\mathcal{L}_{\widehat{f},H}$ are strictly weaker than those of~$\mathcal{L}_{\mathrm D}$ \cite{slezak2026polynomial}, revealing a fundamental tension between \emph{implementability} and \emph{efficiency}. This tension becomes even more pronounced for very large systems, for which the choice of the filter function $f$ becomes key \cite{rouze2024efficient}. This leads to the central question addressed in this paper:
\begin{center}
\textit{Is there a viable way to reconcile \textbf{efficiency} and \textbf{implementability} \\for the dissipative preparation of Gibbs states of infinite-dimensional systems?}
\end{center}

\noindent
To tackle this question, we develop a rigorous framework for Gibbs samplers of infinite-dimensional quantum systems that simultaneously ensures
\begin{itemize}
\item well-posedness of the dynamics, as in infinite dimensions the Lindblad generator may fail to generate a trace-preserving semigroup,
\item spectral convergence guarantees, as convergence properties change dramatically due to the absence of a spectral gap on the trace class operators in this setting, and
\item efficient implementation on finite-dimensional, qubit-base hardware.
\end{itemize}

\smallskip

Complementary hypercontractive criteria for relative-entropy decay within this framework are developed in~%
\cite{BeckerRouzeSalzmannHypercontractivity2026}\hskip0pt%
, with applications to Bose--Hubbard samplers with higher-order jumps and fixed-cutoff Coulomb models.

 \subsection{Gibbs-preserving Markovian  dynamics}\label{sec:GibbsMarkovintro}

 \noindent 
\noindent   It is well-known that unbounded operators that formally satisfy a GKLS-type equation on a natural domain may fail to generate legitimate quantum Markovian dynamics; for instance, the two-photon pure birth process defined with a vanishing Hamiltonian and jump operator $L=(a^\dagger)^2$, where $a^\dagger$ denotes the creation operator over $L^2(\mathbb{R})$, does not preserve the trace \cite[Example~3.3]{chebotarev2000sufficient}. Related pathologies were studied by Fagnola and collaborators, who established sufficient conditions for conservativity by imposing additional structural and domain assumptions on the generators %
\cite{chebotarev2000sufficient,fagnola1999quantum}\hskip0pt%
.
Other approaches to the generation problem include the seminal works of Davies \cite{Davies1977,davies1976quantum,davies1977generator} and Holevo \cite{holevo1993generators,holevo1996quantum}, who established abstract sufficient conditions for unbounded generators of QMSs, albeit those are often difficult to verify in concrete many-body models. More recently, simpler and more explicit sufficient conditions for generation were obtained in \cite{gondolf2024energy} for classes of generators whose jump operators are polynomials in creation and annihilation operators. While well suited to a broad class of continuous-variable models, these results do not apply to the type of generators \eqref{eq:GKLS} considered in the present work. 

\medskip

\noindent \textbf{Main result I (Well-defined Gibbs dynamics).}
Under explicit compatibility and domain assumptions on the Hamiltonian,
bare jumps and filter, we construct trace-preserving quantum Markov
semigroups on a separable Hilbert space that are KMS-symmetric and
preserve the exact Gibbs state
$\sigma_\beta=e^{-\beta H}/\Tr(e^{-\beta H})$
(Section~\mbox{%
\ref{sec:generalframe} }\hskip0pt%
and
Proposition~\mbox{%
\ref{propDirichlettoSchro}}\hskip0pt%
). Under additional hypotheses,
\hskip0pt%
  every initial density
operator converges to $\sigma_\beta$ in trace norm (Theorem~\mbox{%
\ref{thm:generalUniqueness}}).

\medskip

We make use of the abstract theory of KMS-symmetric quantum Markov semigroups developed in \cite{Albeverio1977,Davies1992,Goldstein1995} in order to derive our generation theorem. Although the generator we consider is formally identical to that introduced in \cite{ding2025efficient}, additional compatibility conditions relating the Hamiltonian $H$, the jump operators, and the filter function are required in order to establish the well-posedness of the master equation: 
For the moment, we pick an arbitrary bounded filter function $\widehat{f}:\mathbb{R}\to\mathbb{C}$ which simply needs to satisfy the symmetry condition
\begin{align}\label{eq:symmetryintro}
\overline{\widehat{f}(\nu)}=\widehat{f}(-\nu)\,e^{-\beta\nu/2}.
\end{align}
In Proposition \ref{propDirichlettoSchro}, we show that these choices lead to a well-defined Lindbladian $\mathcal{L}$ generating a semigroup of quantum channels over a separable Hilbert space $\mathcal{H}$ formally defined as on the right-hand side of \Cref{eq:GKLS},
and with a coherent term $B$ satisfying the following:
for any $E,E'\in \operatorname{Sp}(H)$ with corresponding eigenstates $|E\rangle$, $|E'\rangle$,
\begin{equation}
\begin{split}\label{eq:Bintro} \langle E'&|B|E\rangle=\frac{i}{2}\operatorname{tanh}\Big(\tfrac{\beta(E'-E)}{4}\Big)\sum_{\alpha\in\mathcal{A}}\sum_{\substack{\nu_1,\nu_2\in B(H)\\\nu_2-\nu_1=E'-E}}\overline{\widehat{f}(\nu_1)}\widehat{f}(\nu_2)\,\langle E'|(A^\alpha)^\dagger P_{E+\nu_2}A^\alpha |E\rangle. 
\end{split}
\end{equation}
\smallskip

\noindent At this stage, the generator $\mathcal{L}_{\widehat{f},H}$ still seems to depend on the spectral decomposition of $H$. To circumvent this, we need representations of the jumps $L^\alpha$ and of the coherent term $B$ that are independent of the spectral decomposition of $H$. For this, we first 
 assume that    $\widehat f$ is the Fourier transform of an
$L^1(\mathbb R)$ function $f$ , giving the jumps   the integral
formula  in the first identity of \Cref{def:Lalpha} . The weighted
moment assumptions used to control truncation are stated in
Section~\ref{sec:FiniteTruncBareJumps} . Next, we find a spectrally agnostic  representation of the coherent part $B$. In finite dimensions,  \cite{ding2025efficient} considered the drift part $G\rho+\rho  \,G^\dagger$, with 
\begin{align*}
G&:=-iB-\frac{1}{2}\sum_\alpha (L^\alpha)^\dagger L^\alpha  =\sum_{\alpha\in\mathcal{A}}\sum_{\nu\in B(H)}\widehat{g}(\nu)\big((L^\alpha)^\dagger L^\alpha\big)_\nu,
\end{align*}
with $\widehat{g}(\nu):=-\frac{1}{2}\big(\operatorname{tanh}(-\beta \nu/4)+1\big) $. As such, since the function $\widehat{g}$ does not possess a smooth Fourier transform, $G$ does not seem to admit an integral representation.
In \cite{ding2025efficient}, the issue was resolved by introducing a compactly supported cut-off function $\kappa(\nu)=1$ for $\nu\in [-2\|H\|,2\|H\|]$, and denoting 
$
\widehat{g}_\kappa(\nu):=\widehat{g}(\nu)\,\kappa(\nu).
$
This way, $G$ coincides with the operator defined by replacing $\widehat{g}$ with $\widehat{g}_\kappa$. Moreover, since $\widehat{g}_\kappa$ admits a Fourier transform $g_\kappa$, we can write
\begin{align*}
G=\sum_{\alpha\in \mathcal{A}}\int_{\mathbb{R}} g_\kappa(t) \,e^{iHt}\big((L^\alpha)^\dagger L^\alpha\big) e^{-itH}\,dt.
\end{align*}
When $H$ is unbounded, a sharp cutoff function $\kappa$ regularizing the function $\widehat{g}_\kappa$, without modifying the generator $G$, no longer exists. Moreover, imposing a cut-off would eventually lead to a loss of the Gibbs-preservation property of the evolution generated by $\mathcal{L}_{\widehat{f},H}$. 

\smallskip
To resolve these conflicting constraints, we consider the following Gaussian weighted version of the generator $\mathcal{L}_{\widehat{f},H}$, which has already been considered in finite dimensions in \cite{gilyen2024quantum,slezak2026polynomial}: for a given width $\sigma_E\ge 0$,
\begin{align}
&\mathcal{L}_{\sigma_E,\widehat{f},H}(\rho)\!:=\!\sum_\alpha\sum_{\nu_1,\nu_2}\!e^{-\frac{(\nu_1-\nu_2)^2}{8\sigma_E^2}}\,\overline{\widehat{f}(\nu_1)}\widehat{f}(\nu_2) \Big(\!\!-i[B^\alpha_{\nu_1,\nu_2},\rho] +A^\alpha_{\nu_2}\rho (A^\alpha_{\nu_1})^\dagger\!-\!\frac{1}{2}\{(A^\alpha_{\nu_1})^\dagger A^\alpha_{\nu_2},\rho\}\Big),\nonumber
\end{align}
where 
  \begin{align*}
B^\alpha_{\nu_1\nu_2}
&=\frac{i}{2}\tanh(\beta(\nu_2-\nu_1)/4)\,(A^\alpha_{\nu_1})^\dagger A^\alpha_{\nu_2}.
\end{align*} 
This formal definition is shown to generate a semigroup of quantum channels   fixing  $\sigma_{\beta}$ in \Cref{monotonicitygap}. The semigroup is, in fact, KMS-symmetric with respect to $\sigma_\beta$, which means that for any $t$, the Heisenberg-dual, weak$^*$-continuous semigroup generated by $\mathcal{L}_{\smash{\sigma_E,\widehat{f},H}}^\dagger$, is self-adjoint with respect to the so-called KMS scalar product on the bounded operators on $\mathcal{H}$
\begin{align}\label{eq:KMSsymmetry}
\langle \bullet,\bullet \rangle_{\sigma_\beta}:(X,Y)\mapsto \Tr\Big(\sigma_\beta^{\frac{1}{2}}X^\dagger \sigma_\beta^{\frac{1}{2}}Y\Big).
\end{align}
Moreover, $\mathcal{L}_{\smash{\sigma_E,\widehat{f},H}}$ interpolates between the generator $\mathcal{L}_{\smash{\widehat{f},H}}$ at  $\sigma_E=\infty$ and the Davies generator $\mathcal{L}_{\operatorname{D}}$ associated with $H$ at a rate of   $\Upsilon=|\widehat{f}|^2$  at $\sigma_E=0$. Most importantly, the added Gaussian envelope helps in regularizing the drift: in Proposition \ref{integralrep}, we show that
\begin{align}\label{eq:LsigmaEintro}
\mathcal{L}_{\sigma_E,\widehat{f},H}(\rho)=G_{\sigma_E}\,\rho+\rho\, G_{\sigma_E}^\dagger+\Phi_{\sigma_E,\widehat{f},H}(\rho),
\end{align}
with 
  \begin{align*}
 G_{\sigma_E}:=-\sum_{\alpha\in\mathcal{A}}\int_{-\infty}^ \infty g(t)\,e^{itH}((L^{\alpha})^\dagger L^\alpha)e^{-itH}dt,
\end{align*} 
with 
$g(t)=\frac{1}{2\pi}\int_{-\infty}^\infty \frac{e^{-\nu^2/8\sigma_E^2}}{1+e^{\beta\nu/2}} e^{-i\nu t}d\nu$, $X^{\alpha}_s:=e^{isH}L^\alpha e^{-isH}$, and a CP map 
\begin{align*}
\Phi_{\sigma_E,\widehat{f},H}(\rho):=\sigma_E\sqrt{\frac{2}{\pi}}\sum_{\alpha}\int_{\mathbb{R}}e^{-2\sigma_E^2 s^2}\, \, X_s^{\alpha} \rho (X_s^\alpha)^\dagger \,ds.
\end{align*}
  Under the standing assumptions on $H$ and the bare jumps, an
integrable time-domain filter with KMS-symmetric Fourier transform
therefore gives  a spectral-agnostic  integral representation of a
 generator that exactly fixes the Gibbs state of $H$ at  inverse
temperature $\beta$. 

\smallskip

 \subsection{Convergence guarantees via spectral  analysis}\label{sec.fastconv}

 Assuming $\ker \mathcal{L}_{\sigma_E,\widehat f,H}=
\mathbb C{\sigma_\beta}$, we argue that the evolution generated by $\mathcal{L}_{\sigma_E,\widehat{f},H}$ converges to its fixed point $\sigma_\beta$.
Given a subset $\mathscr{S}$ of input quantum states, we define the mixing time by
\[
t_{\operatorname{mix}}(\epsilon,\mathscr{S})
\!:=\!\inf\Bigl\{ t \ge 0 \Big| \| \rho_t - \sigma_\beta \|_1 \le \epsilon \,\forall\rho \in \mathscr{S}\! \Bigr\}
\]
where $\rho_t$ denotes the state $e^{t\mathcal{L}_{\sigma_E,\widehat{f},H}}(\rho)$ evolved according to the semigroup generated by $\mathcal{L}_{\sigma_E,\widehat{f},H}$, and $\sigma_\beta$ is its unique stationary state. 
Establishing polynomial-time convergence for Gibbs sampling dynamics is notoriously challenging, as the mixing time depends delicately on both the Hamiltonian $H$ and the inverse temperature $\beta$. In \cite{rouze2024efficient}, it was first proved that high-temperature Gibbs states of geometrically local Hamiltonians with locally finite-dimensional constituents can be prepared in time $t_{\mathrm{mix}}=\mathcal{O}(n)$; this was later improved in \cite{rouze2411optimal} to the optimal bound $t_{\mathrm{mix}}=\mathcal{O}(\log n)$. These mixing results were subsequently extended to several low-temperature regimes, including spin chains \cite{https://doi.org/10.48550/arxiv.2510.08533}, certain CSS codes above critical thermodynamic temperatures \cite{stengele2025modified}, and perturbations of Gaussian fermionic models \cite{vsmid2025rapid,vsmid2025polynomial,tong2025fast}. In contrast, exponential lower bounds below critical temperatures were also obtained through quantum extensions of the bottleneck lemma \cite{gamarnik2024slow}. Nevertheless, all such existing approaches currently rely fundamentally on the boundedness of the generator.

\smallskip
First, we contemplate the possibility of getting uniform convergence of the evolution over the entire set $\mathscr{S}\equiv \mathscr{S}(\mathcal{H})$ of states: given a strongly continuous semigroup over a Banach space $\mathscr{B}$, a necessary condition for it to converge in norm is that the spectrum of its generator $A$ is gapped, meaning that there exists a positive constant $\lambda_0>0$ such that for all $\lambda\in\operatorname{Sp}(A)\backslash \{0\}$ we have $\operatorname{Re}(\lambda)\le -\lambda_0$ \cite[Corollary 4.1.2]{salzmann2024robustness}.
 A natural Banach space on which one would want to study the  dynamics generated by $\mathcal{L}_{\sigma_E,\widehat{f},H}$ is the  space of trace-class operators over $\mathcal{H}$  ; there a  spectral gap is necessary, but not sufficient, for uniform  exponential convergence  to $\sigma_\beta$ .  Unfortunately, this condition often fails, even in the simplest  settings (see Proposition \ref{thm.qOUnogap}). This first  observation, in sharp contrast with the finite-dimensional setting,  strongly suggests the need to restrict the set $\mathscr{S}$ of  input states.

\smallskip

Here again, the imposed KMS-symmetry condition \eqref{eq:KMSsymmetry} can be used to our advantage. Indeed, this condition implies that we can associate with the generator $\mathcal{L}_{\sigma_E,\widehat{f},H}$ a self-adjoint operator $L_{\sigma_E,\widehat{f},H}$ on the Hilbert space $\mathscr{T}_2$ of Hilbert-Schmidt operators on $\mathcal{H}$, via the defining property that for any $x\in \mathscr{T}_2$, 
\begin{align*}
\sigma_\beta^{\frac{1}{4}}e^{tL_{\sigma_E,\widehat{f},H}}(x)\sigma_\beta^{\frac{1}{4}}=e^{t\mathcal{L}_{\sigma_E,\widehat{f},H}}(\sigma_\beta^{\frac{1}{4}}x\sigma_\beta^{\frac{1}{4}}).
\end{align*}
By a standard use of H\"{o}lder's inequality, we conclude that for input states of the form $\rho=\sigma_{\beta}^{\frac{1}{4}}x\sigma_\beta^{\frac{1}{4}}$,
\begin{align*}
\big\|e^{t\mathcal{L}_{\sigma_E,\widehat{f},H
}}(\rho)-\sigma_{\beta}\big\|_1\le \|e^{tL_{\sigma_E,\widehat{f},H}}(x)-\sigma_{\beta}^{\frac{1}{2}}\|_2.
\end{align*}
Moreover, assuming $\ker L_{\sigma_E,\widehat f,H}=\mathbb C\sqrt{\sigma_\beta}$ and a positive spectral gap, the above norm decays exponentially in $t$ with a rate $\lambda_2\equiv \operatorname{gap}(L_{\sigma_E,\widehat{f},H})$ given by the largest value for which $L_{\sigma_E,\widehat{f},H}$ satisfies the condition that $\operatorname{Re}(\lambda)\le -\lambda_2$ for all $\lambda\in \operatorname{Sp}(L_{\sigma_E,\widehat{f},H})\backslash \{0\}$, so we obtain the convergence

\begin{align*}
\big\|e^{t\mathcal{L}_{\sigma_E,\widehat{f},H}}(\rho)-\sigma_{\beta}\big\|_1\!\le\! e^{-\lambda_2t}\|x-\sigma_\beta^{\frac{1}{2}}\|_2\le e^{-\lambda_2 t}\|x\|_2,
\end{align*}
where in the last bound we also used that $x-\sqrt{\sigma_\beta}$ is orthogonal to $\sqrt{\sigma_\beta}$ due to the normalization of the state $\rho$. Thus, on the set $\mathscr{S}\equiv \mathscr{S}_2$ of states of the form $\rho={\sigma_\beta}^{1/4}x{\sigma_\beta}^{1/4}$ with $x\in \mathscr{T}_2$, the dynamics converges exponentially fast to $\sigma_\beta$. Interestingly, in \Cref{lem.gapmonoton} we also find that a positive spectral gap cannot decrease as $\sigma_E$ decreases, provided the larger-width generator has kernel $\mathbb C\sqrt{\sigma_\beta}$ (see also \cite{slezak2026polynomial}):
\begin{align*}
\operatorname{gap}(L_{\sigma_E,\widehat{f},H})\nearrow\qquad 
\text{ as } \qquad \sigma_E\searrow,
\end{align*}
which justifies the intuition that the Davies dynamics converges faster than that of \cite{ding2025efficient}.

In \Cref{sec:spectralgaps}, we show that the generator   $L_{\widehat{f},H}$  is gapped for various models: first in \Cref{sec:gaussexamples}, we consider Gaussian models with quadratic Hamiltonians. All these results are valid for a Schwartz filter function $f$.   For one-mode number-preserving models,
Theorem~\ref{thm:DownwardRateCriterion} gives a sharp criterion in terms of
the downward rates and, in particular, a gap obstruction
whenever the downward transition rates approach zero along a subsequence.
This applies, in particular, to Gaussian frequency filters and
$H=N^p$, $p>1$, and persists under Gaussian convolution.

\medskip

\begingroup
\noindent\textbf{Main result II (A sharp criterion for thermalization rates).}
For one-mode Hamiltonians $H=h(N)$, where $N$ is the number operator,
with jumps $a,a^\dagger$, suppose that $h(n)$ is eventually
non-decreasing and $h(n+s)-h(n)\ge\delta>0$ for all sufficiently
large $n$, for some fixed integer $s\ge1$.
Under the standing generation assumptions, every KMS-symmetric filter
that is real and positive on the active frequencies satisfies
\[
\operatorname{gap}(L_{\sigma_E,\widehat f,h(N)})>0
\quad\Longleftrightarrow\quad
\inf_{n\ge1}n|\widehat f(h(n-1)-h(n))|^2>0,
\qquad 0<\sigma_E\le\infty.
\]
Thus a positive gap is equivalent to a uniform positive lower bound
on the downward transition rates (Theorem~\mbox{%
\ref{thm:DownwardRateCriterion} }\hskip0pt%
and
Corollary~\mbox{%
\ref{cor:SharpPolynomialFilterGaussian}}\hskip0pt%
).

\noindent For the polynomially damped Metropolis filters
$\widehat f_\eta(\nu)=\hatfM(\nu)/(1+\beta^2\nu^2)^{\eta/2}$
defined in \mbox{%
\eqref{eq:PolynomialFilter}}\hskip0pt%
, the criterion becomes
\[
\sup_{n\ge0}\frac{(1+\beta^2[h(n+1)-h(n)]^2)^\eta}{n+1}<\infty.
\]
In particular, for $H=N^p$, $p>1$, the positive-gap range is exactly
$0\le\eta\le[2(p-1)]^{-1}$
(Corollary~\mbox{%
\ref{thm:SharpPolynomialFilter}}\hskip0pt%
). Every fixed $\eta>0$
also satisfies the filter regularity assumptions of our implementation
theory (see Section \ref{secintroimplementation} below), thus reconciling positive gaps with implementable filters.

\par\endgroup

\subsection{Efficient implementation}\label{secintroimplementation}
 In order to implement the sampler generated by the unbounded Lindbladian  $\cL_{\sigma_E,\widehat f,H}$ we consider in Section~\ref{sec:FiniteTruncBareJumps} a finite-dimensional approximation scheme which we then use to obtain an efficient circuit implementation on a qubit-based quantum computer in Section~\ref{sec:FullFiniteDimPipeline}. For that we consider for each truncation level $M\in\N$ a finite rank projection  $P_M$ and for each $\alpha\in\cA$   a projection $\pi^\alpha_M$ of rank $M+1$ on the factor where the jump acts, tensored with the identity on any other factors.  Here, the truncated subspace $\operatorname{im}(P_M)$ serves as the \emph{system register} of the quantum device on which we aim to implement the Gibbs sampler approximately. For instance, in many-body or multi-mode systems, the number of bare jumps $|\mathcal{A}|$ is typically proportional to the number of particles or modes. Accordingly, one usually considers truncations for which the local register space is associated with the image of $\pi^\alpha_M$ and the dimension of the full system register satisfies
$
\log\left( \dim\bigl(\operatorname{im}(P_M)\bigr)\right)
= \mathcal{O}\bigl(|\mathcal{A}|\log(M)\bigr),
$ although, at this stage, we are in principle free to leave the dimension of the system register unspecified.

Using the projections $\pi^\alpha_M$ and $P_M$ we can define finite-dimensional truncations of the bare jumps and the Hamiltonian as 
\begin{align}
\left(A^\alpha\right)^{\le M} := \pi^\alpha_M A^\alpha\pi^\alpha_M\qquad\text{and}\qquad  H_{\le M} := P_M H P_M.
\end{align}
This allows us to define the finite-dimensional Lindblad generator
$\mathcal{L}^{\smash{\leq M}}_{\smash{\sigma_E,\widehat f,H_{\leq M}}}$
by replacing the bare jump operators and the Hamiltonian in the unbounded generator
$\mathcal{L}_{\smash{\sigma_E,\widehat f,H}}$
with $\bigl(A^\alpha\bigr)^{\smash{\leq M}}$ and $H_{\smash{\leq M}}$, respectively. If, in addition, the compatibility condition
$
\bigl(A^\alpha\bigr)^{\smash{\leq M}}\operatorname{im}(P_M)\subseteq \operatorname{im}(P_M)$
is satisfied, then
$\mathcal{L}^{\smash{\leq M}}_{\smash{\sigma_E,\widehat f,H_{\leq M}}}$
generates a quantum Markov semigroup on the finite-dimensional system register.

In Section~\ref{sec:FiniteTruncBareJumps}, we show that this finite-dimensional generator provides a good approximation to the target generator $\mathcal{L}_{\sigma_E,\widehat f,H}$ for large truncation parameter $M$. Since $\mathcal{L}_{\smash{\sigma_E,\widehat f,H}}$ is typically unbounded, the approximation has to be understood pointwise on suitably energy-constrained states. To make this precise, we introduce a self-adjoint, positive semidefinite energy observable $\NA$ and consider states whose expectation values with respect to suitable exponential weights in $\NA$ are finite. We then assume that both the truncated bare jumps and the truncated Hamiltonian approximate their original counterparts well on such inputs, with errors that are exponentially small in $M$ up to polynomial prefactors, see~\eqref{eq:AbstractFiniteTruncBareJumps} and~\eqref{eq:HamiltonianLeakTrunc}. We further require that the Hamiltonian evolution is compatible with the same energy constraint, in the sense that it drives low-energy states into higher-energy sectors only at a controlled exponential rate, c.f.~\eqref{eq:ExpBound}. As discussed in Sections~\ref{sec:TruncateJumps} and \ref{sec:FinitDimGenerator} all of these assumptions are naturally satisfied for multi-mode bosonic systems with bare jumps being the creation and annihilation operators, certain Hamiltonians constructed as bounded degree polynomials in $a_i$ and $a^\dagger_i$ and for the choice $\NA \equiv \sum_{i} a^\dagger_ia_i.$   

Under these assumptions, we show in Theorem~\ref{thm:IntegrableGibbsDynamics} for 
filters with $f\in L^1(\mathbb R)$  decaying sufficiently fast in time representation and input states $\rho$ satisfying
\begin{align}
\label{eq:rhocIntro}
    \rho \le \mathfrak{c} \,\sigma_\beta
\end{align}
for some $\mathfrak{c}\ge 1$ that for evolution time $t\ge 0$ and accuracy $\eps>0$ we can achieve\footnote{Here, $\mathcal{O}$ and $\widetilde{\mathcal{O}}$  treat all non-displayed parameters as constant and $\widetilde{\mathcal{O}}$ additionally supresses subdominant polylogarithmic factors in $t$ or $1/\lambda_2$ respectively.}    \begin{align}\label{eq:MetalIntro}
&\left\|\left(e^{t\cL_{\sigma_E,\widehat f,H}}-e^{t\cL^{\le M}_{\sigma_E,\widehat f,H_{\le M}}}\right)(\rho)\right\|_1\le \eps,\qquad\text{with}\qquad M = \mathcal{O}\left(\operatorname{poly}\left(\log\left(2+\frac{t\,\mathfrak{c}\,E_{\operatorname{Gibbs}}|\cA|}{ \eps}\right)\right)\right).
    \end{align} 
Here, $E_{\operatorname{Gibbs}}$ denotes the expectation value of the mentioned exponential energy observable involving $\NA$ with respect to the Gibbs state. In Section~\ref{sec:FullFiniteDimPipeline}, under the additional Gevrey assumption Condition~\ref{condfgevrey},  we then combine this
result with the work of \cite{chen2023quantum,gilyen2024quantum,chen2023efficient,ding2025efficient}, in particular \cite[Theorem 18]{ding2025efficient}, which provide efficient circuit implementations of such finite dimensional Lindbladian dynamics by approximating the involved time integrals via linear combinations of unitaries \cite{LiWang23} given oracle access to the Hamiltonian evolution $e^{\smash{-isH_{\le M}}}$  and block encoding of the bare jumps of the truncated bare jumps $\left(A^\alpha\right)^{\smash{\le M}}$.

\medskip 

\noindent 
\textbf{Main result III (Finite-dimensional implementation, Theorem~\ref{thm:CircuitDynamicsIntegrable}, Corollary~\ref{cor:GibbsCircuitIntegrable}).}
Provided a state preparation circuit for  a finite-register approximation $\rho_M$ of the  input state $\rho$ in \eqref{eq:rhocIntro},  as specified in \eqref{eq:registerInputSurgical},  
$e^{\smash{t\cL_{\sigma_E,\widehat f, H}}}(\rho)$ can be prepared on the finite dimensional system register within $\eps$-trace distance with total Hamiltonian simulation time corresponding to $H_{\le M}$ of order
  \begin{align*}
\widetilde{\mathcal{O}}\left( \,(1+t)\,\operatorname{poly}\left(|\mathcal{A}|\,,\, \log\left(2+\frac{\mathfrak{c}\,E_{\operatorname{Gibbs}}}{\epsilon}\right)\right)\right).
\end{align*} 
Therefore, assuming $\ker L_{\sigma_E,\widehat f,H}=\mathbb C\sqrt{\sigma_\beta}$ and positivity of spectral gap, $\lambda_2\equiv \operatorname{gap}(L_{\sigma_E,\widehat f, H})>0,$ the Gibbs state of the Hamiltonian $H$ can be prepared via a finite dimensional circuit with Hamiltonian simulation time corresponding to $H_{\le M}$ of order
  \begin{align*}
\widetilde{\mathcal{O}}\left(\left(1+\frac{1}{\lambda_2}\right)\operatorname{poly}\left(|\mathcal{A}|\,,\,\log\left(2+\frac{\mathfrak{c}E_{\operatorname{Gibbs}}}{\epsilon}\right)\right)\right).
\end{align*} 

\noindent For the polynomially damped filters $\widehat f_\eta$,
$\eta>0$, Lemma~\ref{lem:PolynomialFilterTime} provides the
weighted time integrability required for truncation, and
Lemma~\ref{lem:PolynomialFilterGevrey} verifies the Gevrey
assumption used for circuit implementation. Thus the same
implementation results apply directly to these filters. In the
one-mode setting,   Corollary~\ref{thm:SharpPolynomialFilter}  and its
extension in Corollary~\ref{cor:SharpPolynomialFilterGaussian}
characterize exactly which exponents preserve a positive gap.

  For many-body continuous-variable quantum systems of  interest, both the quantity $\mathfrak{c}$ in~\eqref{eq:rhocIntro} and the Gibbs energy $E_{\operatorname{Gibbs}}$ typically scale exponentially with the number of particles or modes, much like the partition function of the Gibbs state. Owing to the logarithmic dependence of the above complexity bounds on these quantities, which stems from the exponential energy constraint underlying the finite-dimensional approximation scheme, the resulting implementations of the Lindblad dynamics and Gibbs state preparation remain efficient, with resource requirements that scale polynomially in the number of particles or modes.
 This statement assumes polynomial bounds on the remaining constants and oracle costs.

\bigskip

\smallsection{Acknowledgement}
SB would like to thank Lin Lin for fruitful discussions on the Gibbs sampling of powers of the number operator. He would also like to thank Jeff Galkowski and Maciej Zworski for bringing the issue of defining Gibbs dynamics for Schr\"odinger operators to his attention. This led to Theorem \ref{thm:Schroedinger_operators}.
SB would also like to acknowledge support from the SNF Grant PZ00P2\_216019. 
CR is supported by France 2030 under the
French National Research Agency award number ''ANR-22-EXES-0013''. RS acknowledges support by the European Research Council (ERC Grant Agreement No.~948139 and ERC Grant AlgoQIP, Agreement No. 851716), from the Excellence Cluster Matter and Light for Quantum Computing (ML4Q-2), from the QuantERA II Programme of the
European Union’s Horizon 2020 research and innovation programme under Grant Agreement No
101017733 (VERIqTAS) as well as the government grant managed by the Agence Nationale de la
Recherche under the Plan France 2030 with the reference ANR-22-PETQ-0007.

\smallsection{AI statement}
AI-assisted tools were used only for light language editing in the first version of this article posted on arXiv. In preparing the revised version, we used these tools for a detailed critical review of the original manuscript, which helped us strengthen several statements and streamline the presentation. The authors are responsible for the mathematical content and the final text.

\section{KMS-symmetric generators in infinite dimensions}\label{sec:generalframe}

\noindent In this section, we construct a family of quantum Gibbs samplers for infinite-dimensional quantum systems. Given a densely defined self-adjoint operator $(H,D(H))$ on a separable Hilbert space $\mathcal{H}$ with $H\ge -h_0 I$, $h_0\ge 0$, and inverse temperature $\beta>0$, we aim to prepare the corresponding Gibbs state of the form
\begin{align*}
\sigma_\beta:=\frac{e^{-\beta H}}{\mathcal{Z}(\beta)},\qquad \text{ with }\qquad \mathcal{Z}(\beta):=\Tr(e^{-\beta H})<\infty,
\end{align*}
for Hamiltonians   satisfying the Gibbs hypothesis $\mathcal Z(b)<\infty$ for every $b>0$ .
We recall that   this hypothesis directly implies, in infinite dimension,  that the spectrum is unbounded above, that the spectrum is discrete, and that the energy levels cannot become \emph{too dense} with increasing energy values. In contrast to the finite-dimensional setting, where generators of quantum dynamical semigroups have long been fully characterized \cite{Lindblad1976,Gorini1976}, extensions to unbounded jumps are more intricate. Here, we leverage a certain detailed balance condition that will allow us to define our evolutions following the abstract theory of KMS-symmetric quantum Markov semigroups developed in \cite{Albeverio1977,Davies1992,Goldstein1995,Cipriani1997}.

We consider a finite set $\{A^\alpha\}_{\alpha\in\mathcal{A}}$ of closed, densely defined jump operators $A^\alpha$, each on its respective domain $D(A^\alpha)$, whose family is invariant under taking adjoints, i.e. $\{A^\alpha\}_{\alpha\in \mathcal{A}}=\{(A^\alpha)^\dagger\}_{\alpha\in \mathcal{A}}$. We assume that a common test domain $D\subseteq\bigcap_\alpha D(A^\alpha)$ includes all the eigenstates of $H$; equality of the closed operators' domains is not required, and we require the following condition throughout the paper

\begin{framed}
\begin{conditionA}\label{eq:condAalphas}
There exist some moments $0\le\gamma,  \mu$, with $\gamma\le \mu$, as well as a constant $C>0$ such that, denoting $\widetilde{H}:=H+(h_0+1)I$,
\begin{align}
\label{eq:BoundBareJumpsWithHam}
\|A^\alpha \widetilde{H}^{-\gamma}\|,\quad \|\widetilde{H}^\gamma\,A^\alpha \widetilde{H}^{-\mu}\|\le C.
\end{align}
Next, we consider a bounded function $\widehat{f}:\mathbb{R}\to\mathbb{C}$ with 
\begin{align}
\label{eq:KMSFilterFunction}
\overline{\widehat{f}(\nu)}=\widehat{f}(-\nu)\,e^{-\beta\nu/2}\qquad \forall\nu\in\mathbb{R}.
\end{align}
\end{conditionA}
\end{framed}
 By \eqref{eq:KMSFilterFunction} and Condition~\ref{eq:condAalphas}, $\sup_\nu e^{s\beta\nu}|\widehat f(\nu)|\le \|\widehat{f}\|$ for every $0\le s\le1/2$.

\noindent Writing the spectral and eigenvalue decompositions of $H$ as $H=\sum_{E\in\operatorname{Sp}(H)}EP_E=\sum_{i}E_i|E_i\rangle\langle E_i|$, with $E_0\le E_1\le \cdots $, where $\{\ket{E_i}\}_{i\in\mathbb{N}_0}$ denotes the energy eigenbasis of $H$ by slight abuse of notation ($\ket{E_i}\ne \ket{E_{i+1}}$, even if $E_i=E_{i+1}$), we will extensively make use of the subspaces
\begin{align*}
\mathcal{F}:=\operatorname{span}\{\ket{E_i}\}_{E_i\in\operatorname{Sp}(H)}\,,\qquad \mathscr{F}:=\operatorname{span}\{|E_i\rangle\langle E_j|\}_{E_i,E_j\in\operatorname{Sp}(H)} 
\end{align*}
of the Hilbert space $\mathcal{H}$, resp.~of the space $\mathscr{T}_1(\mathcal{H})$ of trace-class operators over $\mathcal{H}$.
Both spans are algebraic. Thus $\mathcal F$
consists of finite linear combinations of energy eigenvectors,
whereas $\mathscr F$ consists of operators with finitely many
nonzero matrix entries in this eigenbasis. 
The next claim is standard.
\begin{claim}\label{densityFFs}
The space $\mathcal{F}$ is dense in $\mathcal{H}$, while $\mathscr{F}$ is dense in $\mathscr{T}_1(\mathcal{H})$.
\end{claim} 

For Schr\"odinger operators, we then have the following theorem that shows that, for the set of bare jumps given by the creation and annihilation operators, i.e., $\{A^\alpha\}_{\alpha\in\cA}\equiv \{a_j,a^\dagger_j\}_{j=1}^m,$ Condition \theconditionA \ is satisfied under very general assumptions. 
A different perspective from our Dirichlet form approach, by directly verifying Davies' conditions to obtain a semigroup in the space of trace-class operators, has been pursued in \cite{GalkowskiZworskiDavies}.    
\begin{theorem}
\label{thm:Schroedinger_operators}
    We consider the Schr\"odinger operator $H=-\Delta+V.$ 
    We write $\langle x\rangle :=     (1+|x|^2)^{1/2}$ for $x\in\mathbb R^d$. 
    Let $V$ be real-valued and satisfy 
\begin{itemize}
\item $V \in C^{\infty}(\mathbb R^d)$ with $
V(x)\ge c\langle x\rangle^r-c_0,$ with $ c>0,c_0 \ge 0, r\ge 1,$
as well as
\[
|\partial_x^\alpha V(x)|\le C_\alpha \langle x\rangle^{r-|\alpha|},
\qquad \alpha\in\mathbb N^d.
\]
Choose   $\lambda>c_0+1$  and set 
\[
\widetilde H:=H+\lambda.
\]
Let $a_j,a_j^{\dagger}$ be the annihilation and creation operators.
Then for every   $n\in\mathbb N_0$ ,
\[
\widetilde H^n a_j \widetilde H^{-n-1},\widetilde H^n a_j^{\dagger} \widetilde H^{-n-1}\in \mathcal B(L^2(\mathbb R^d)),
\]
i.e. $\gamma=1,\mu=2$ are admissible in Condition \theconditionA \  for the set of bare jumps being $\{A^\alpha\}_{\alpha\in\cA}\equiv \{a_j,a^\dagger_j\}_{j=1}^d.$
\item $V=\nu |x|^2 + W$ for $\nu>0$ with   $W \in  L^{p_d}(\mathbb R^d)+ \langle x \rangle^{\alpha} L^{\infty}(\mathbb R^d)\text{ with } \alpha<2$, where $p_d=\max\{2,d/2\}$ for $d\ne4$ and $p_4>2$ is fixed,  then $\gamma=1/2$ and $\mu=1$ are admissible in Condition \theconditionA \ for the set of bare jumps being $\{A^\alpha\}_{\alpha\in\cA}\equiv \{a_j,a^\dagger_j\}_{j=1}^d.$
\end{itemize}
\end{theorem}
We stress that the first case includes trapping potentials to very high order and the second case includes singular potentials such as   three-dimensional Coulomb potentials . 
The proof of this theorem is given in Appendix \ref{appendixC}.

\subsection{Dirichlet forms and Gibbs generators on Hilbert-Schmidt operators}

\noindent 
Next, for any Bohr frequency $\nu\in B(H)=\operatorname{Sp}(H)-\operatorname{Sp}(H)$ and label $\alpha\in\mathcal{A}$, we introduce the energy jump operators $(A^\alpha_\nu,\mathcal{F})$ by
\begin{align*}
A^\alpha_\nu\ket{E_i}:=P_{E_i+\nu}A^\alpha\ket{E_i}.
\end{align*}
Above, we also write $P_E=0$ whenever $E$ is not an eigenvalue of $H$.
With slight abuse of notation, we also denote $(A^\alpha_\nu)^\dagger=\big((A^\alpha)^\dagger\big)_{-\nu}$, where we recall that $\{A^\alpha\}_{\alpha\in\mathcal{A}}$ is invariant under adjoints. Similarly, we formally define the operators $(A^\alpha_{\nu_1})^\dagger A^\alpha_{\nu_2}$ on $\mathcal{F}$, whose actions on an energy eigenstate $\ket{E_i}$ are
\begin{align*}
(A^\alpha_{\nu_1})^\dagger A^\alpha_{\nu_2}\ket{E_i}=P_{E_i+\nu_2-\nu_1}(A^\alpha)^\dagger P_{E_i+\nu_2}A^\alpha\ket{E_i}.
\end{align*}
Next, we formally introduce the operator
\begin{align}\label{def:L}
L_{\widehat{f},H}(\lambda x+\mu\sqrt{\sigma_\beta})\!:=-\lambda \sum_{\substack{\alpha\in\mathcal{A}\\\nu_1,\nu_2\in B(H)}}\frac{\overline{\widehat{f}(\nu_1)}\widehat{f}(\nu_2)\,e^{\beta(\nu_1+\nu_2)/4}}{2\cosh((\nu_1-\nu_2)\beta/4)}\,(\delta^\alpha_{\nu_1})^{\dagger}\delta^\alpha_{\nu_2}(x)
\end{align}
where 
\begin{equation}
\label{eq:delta}\delta^\alpha_\nu(x):=e^{-\frac{\beta \nu}{4}}A^\alpha_\nu x-e^{\frac{\beta\nu}{4}}xA^\alpha_\nu
\end{equation}
and $(\delta_\nu^\alpha)^{\dagger}$ denotes the formal adjoint of $\delta^\alpha_\nu$ in the Hilbert space $\mathscr{T}_2(\mathcal{H})$ of Schatten-2 operators over $\cH$ endowed with the Hilbert-Schmidt inner product $\langle A,B\rangle:=\Tr(A^\dagger B)$, i.e.
\begin{align*}
(\delta^\alpha_{\nu})^\dagger(x)=e^{-\frac{\beta\nu}{4}}(A^\alpha_\nu)^\dagger x-e^{\frac{\beta\nu}{4}}x(A^\alpha_\nu)^\dagger.
\end{align*}

 \begin{lemma}\label{lem:definL}
 Under \Cref{eq:condAalphas}, the expression \eqref{def:L} defines a  densely defined, negative semidefinite symmetric operator  $(L_{\widehat{f},H},\operatorname{span}(\mathscr{F}\cup  \{\sqrt{\sigma_\beta}\}))$ over $\mathscr{T}_2(\mathcal{H})$, with  $\sqrt{\sigma_\beta}\in \operatorname{Ker}(L_{\widehat{f},H})$.
 \end{lemma}

 \begin{proof}
We consider $x=|E_i\rangle\langle E_j|$,
\begin{align}
\label{eq:deltaTriangle}
\nn\|(\delta^\alpha_{\nu_1})^\dagger\delta^\alpha_{\nu_2}x\|_2&\le \!e^{-\frac{\beta(\nu_1+\nu_2)}{4}}\|(A^\alpha_{\nu_1} )^\dagger A^\alpha_{\nu_2}x\|_2+e^{\frac{\beta(\nu_1+\nu_2)}{4}}\|xA^\alpha_{\nu_2}(A^\alpha_{\nu_1})^\dagger\|_2\\
&\quad +e^{\frac{\beta(\nu_2-\nu_1)}{4}}\|(A^\alpha_{\nu_1})^\dagger x A^\alpha_{\nu_2}\|_2+e^{\frac{\beta(\nu_1-\nu_2)}{4}}\|A^\alpha_{\nu_2}x(A^\alpha_{\nu_1})^\dagger\|_2.
\end{align}

We treat each term in \eqref{eq:deltaTriangle} one by one: For the first one we see 
  \begin{align*}
&\sum_{\nu_1,\nu_2} \frac{|\widehat{f}(\nu_1)\widehat{f}(\nu_2)|\,}{2\cosh((\nu_1-\nu_2)\beta/4)}\|(A^\alpha_{\nu_1} )^\dagger A^\alpha_{\nu_2}|E_i\rangle\langle E_j|\|_2 \\
&\le \sum_{E',E\in\spec(H)} \frac{|\widehat{f}(E'-E)\widehat{f}(E'-E_i)|\,}{2\cosh((E_i-E)\beta/4)}\|P_E(A^\alpha)^\dagger P_{E'}A^\alpha|E_i\rangle\langle E_j|\|_2  \\
&\le \|\widehat f\|_\infty \sum_{E,E'\in\spec(H)} \frac{|\widehat{f}(E'-E)|}{2\cosh((E_i-E)\beta/4)}\|P_E(A^\alpha)^\dagger P_{E'}A^\alpha|E_i\rangle\langle E_j|\|_2 \\
&\le e^{\beta E_i/4}\|\widehat f\|_\infty\|\widetilde H^{-\gamma}A^\alpha\| \|\widetilde H^{\gamma}A^\alpha\widetilde H^{-\mu}\| (E_i+1+h_0)^\mu\sum_{E,E'\in\spec(H)}|\widehat{f}(E'-E)| e^{-\beta E/4}.
\end{align*} 
It remains to argue that the last sum above converges. Indeed, by the  boundedness and KMS symmetry \eqref{eq:KMSFilterFunction} of $\widehat{f}$ in \Cref{eq:condAalphas}: 
  \begin{align*}
\sum_{E,E'\in\spec(H)}|\widehat{f}(E'-E)| e^{-\beta E/4}&=\sum_{E,E'\in\spec(H)}|\widehat{f}(E'-E)|^{3/4} |\widehat{f}(E'-E)|^{1/4}e^{\beta(E'-E)/8}e^{\beta(E-E')/8}e^{-\beta E/4}\\
&\le \|\widehat{f}\| \sum_{E,E'\in\spec(H)}e^{-\beta(E+E')/8}<\infty
\end{align*} 
where the convergence is ensured by the Gibbs hypothesis.
For the second term in \eqref{eq:deltaTriangle} we see
  \begin{align*}
    &\sum_{\nu_1,\nu_2}     \frac{|\widehat{f}(\nu_1)\widehat{f}(\nu_2)|\,}{2\cosh((\nu_1-    \nu_2)\beta/4)}e^{\frac{\beta(\nu_1+\nu_2)}{2}}\||E_i\rangle    \langle E_j|A^\alpha_{\nu_2}(A^\alpha_{\nu_1})^\dagger\|_2\\&    \le\|\widehat{f}\|^2\sum_{E',E\in\spec(H)}     \frac{e^{\frac{\beta(E_j-E)}{2}}}{2\cosh((E'-E_j)\beta/4)}\||E_i    \rangle\langle E_j|A^\alpha P_E(A^\alpha)^\dagger P_{E'}\|_2     \\&\le e^{\beta E_j/4}\|\widehat{f}\|^2 \| A^\alpha     \widetilde H^{-\gamma}\| \|\widetilde H^{-\mu}A^\alpha     \widetilde H^{\gamma}\| e^{\beta E_j/2}(E_j+h_0+1)^\mu \times     \\&\qquad\times \left(\sum_{E\in\spec(H)}e^{-\beta     E/2}\right)\left(\sum_{E'\in\spec(H)}e^{-\beta     E'/4}\right)<\infty\,,
\end{align*}
Here we again used the KMS symmetry
\eqref{eq:KMSFilterFunction}, which gives
$e^{\beta\nu/2}|\widehat{f}(\nu)|
=|\widehat{f}(-\nu)|\le\|\widehat{f}\|$.
 
 For the third term in \eqref{eq:deltaTriangle} we argue as
  \begin{align*}
    &\sum_{\nu_1,\nu_2}     \frac{|\widehat{f}(\nu_1)\widehat{f}(\nu_2)|\,}{2\cosh((\nu_1-    \nu_2)\beta/4)}
e^{\frac{\beta\nu_2}{2}}\|(A^\alpha_{\nu_1})^\dagger \ket{E_i}\!\bra{E_j}A^\alpha_{\nu_2}\|_2\\&\le e^{\frac{\beta E_j}{2}} \|\widehat f\|^2_\infty \sum_{E,E''\in\spec(H)} \frac{e^{\frac{-\beta E''}{2}}}{2\cosh((E''+E_i-E-E_j)\beta/4)}
 \|P_E(A^\alpha)^\dagger \ket{E_i}\!\bra{E_j}A^\alpha P_{E''}\|_2  \\&\le e^{\frac{\beta (E_j+E_i)}{4}} \|\widehat  f\|^2_\infty\|\widetilde H^{-\gamma}A^\alpha\|^2(E_i+h_0+1)^{ \gamma}(E_j+h_0+1)^{\gamma}  \left(\sum_{E\in\spec(H)} e^{-\beta  E/4}\right)^2<\infty.
\end{align*} 
Finally, for the fourth term in \eqref{eq:deltaTriangle} we argue as
  \begin{align*}
&\sum_{\nu_1,\nu_2}
\frac{|\widehat{f}(\nu_1)\widehat{f}(\nu_2)|\,}{2\cosh((\nu_1-\nu_2)\beta/4)}e^{\frac{\beta\nu_1}{2}}
\|A^\alpha_{\nu_2}\ket{E_i}\!\bra{E_j}
(A^\alpha_{\nu_1})^\dagger\|_2\\
&\le\|\widehat{f}\|\sum_{E',E'''\in\spec(H)}
\frac{|\widehat{f}(E'''-E_i)|\,}{2\cosh((E'+E_i-E'''-E_j)\beta/4)}
\|P_{E'''}A^\alpha\ket{E_i}\!\bra{E_j}
(A^\alpha)^\dagger P_{E'}\|_2\\
&\le\|\widehat{f}\|e^{\beta E_j/4}
\|A^\alpha\widetilde H^{-\gamma}\|^2
(E_i+h_0+1)^\gamma(E_j+h_0+1)^\gamma\\
&\qquad\times\sum_{E',E'''\in\spec(H)}
|\widehat{f}(E'''-E_i)|e^{\beta(E'''-E_i)/4}
e^{-\beta E'/4}\\
&\le\|\widehat{f}\|^2e^{\beta E_j/4}
\|A^\alpha\widetilde H^{-\gamma}\|^2
(E_i+h_0+1)^\gamma(E_j+h_0+1)^\gamma\\
&\qquad\times\sum_{E',E'''\in\spec(H)}
e^{-\beta(E'''-E_i)/8}e^{-\beta E'/4}<\infty.
\end{align*} 
For the second inequality we used
$(2\cosh u)^{-1}\le e^{-u}$ and
$\|P_E A^\alpha\ket{E_k}\|
\le\|A^\alpha\widetilde H^{-\gamma}\|(E_k+h_0+1)^\gamma$
for $k=i,j$, together with the rank-one Hilbert--Schmidt norm
identity. The first and last inequalities use the KMS symmetry
\eqref{eq:KMSFilterFunction} and boundedness in the form
$e^{s\beta\nu}|\widehat{f}(\nu)|\le\|\widehat{f}\|$
for $0\le s\le1/2$, with $s=1/2$ and $s=3/8$, respectively.

The finiteness of the last sum is once again a consequence of the
Gibbs hypothesis. From this we can conclude that
\begin{align*}
    \|L_{\widehat f,H}\ket{E_i}\!\bra{E_j}\|_2<\infty.
\end{align*}
Moreover, for any $x\in \mathscr{F}$, denoting $x_\nu:=\widehat{f}(\nu)e^{\smash{\frac{\beta\nu}{4}}}\delta_\nu^\alpha x$ and $x_t:=\sum_{\smash{\nu\in B(H)}}e^{\smash{it\nu}}x_\nu$, 
\begin{align}
\mathcal{E}_{\widehat{f},H}(x)&:=-\langle x,L_{\widehat{f},H}x\rangle=\sum_{\substack{\alpha\in\mathcal{A}\\%
\nu_1,\nu_2\in B(H)}}\frac{1}{2\cosh((\nu_1-\nu_2)\beta/4)}\,\langle x_{\nu_1},x_{\nu_2}\rangle\nonumber \\
& =\sum_{\substack{\alpha\in\mathcal{A}\\\nu_1,\nu_2\in B(H)}}\!\!\!\int_{-\infty}^\infty \frac{1}{\beta \cosh(2\pi t/\beta)}\,\langle e^{it\nu_1}x_{\nu_1},e^{it\nu_2}x_{\nu_2}\rangle dt =\sum_{\alpha\in\mathcal{A}}\int_{-\infty}^\infty \frac{1}{\beta  \cosh(2\pi t/\beta)}\,\|x_t\|_2^2 \,dt,\label{eq:Dirichletformdef} 
\end{align}
where the above manipulations are justified by Fubini's theorem, together with
\begin{align*}
\|x_t\|_2\le \sum_{\nu}|\widehat{f}(\nu)|e^{\frac{\beta\nu}{4}}\|\delta^\alpha_\nu(x)\|_2<\infty,
\end{align*}
where we used $\beta^{-1}\int_{-\infty}^\infty\cosh(2\pi t/\beta)^{-1}dt=\frac{1}{2}<\infty$ together with \Cref{eq:condAalphas} so that
  \begin{align}
&\sum_{\nu}|\widehat{f}(\nu)|\,e^{\frac{\beta\nu}{4}}\,\|\delta^\alpha_\nu(x)\|_2\\
&\qquad\le \sum_{\nu}|\widehat{f}(\nu)|\|A^\alpha_\nu\ket{E_i}\bra{E_j}\|_2+\sum_{\nu}|\widehat{f}(\nu)|e^{\frac{\beta\nu}{2}}\|\ket{E_i}\bra{E_j}A^\alpha_\nu\|_2\nonumber\\
&\qquad= \sum_{E}|\widehat{f}(E-E_i)|\|P_EA^\alpha \ket{E_i}\bra{E_j}\|_2+\sum_{E'}|\widehat{f}(E_j-E')|e^{\frac{\beta(E_j-E')}{2}}\|\ket{E_i}\bra{E_j}A^\alpha P_{E'}\|_2\nonumber\nonumber\\
&\qquad= \sum_{E}|\widehat{f}(E-E_i)|\|P_EA^\alpha \ket{E_i}\bra{E_j}\|_2+\sum_{E'}|\widehat{f}(E'-E_j)|\|\ket{E_i}\bra{E_j}A^\alpha P_{E'}\|_2\nonumber\nonumber\\
&\qquad<\infty.\label{eq.deltaEEbraket}
\end{align} 
Thus, $L_{\widehat{f},H}$ is negative semidefinite on $\mathscr{F}$. It remains to verify that for any $x\in \mathscr{F}$,   $\langle \sqrt{\sigma_\beta},L_{\widehat f,H}x\rangle=0$ . For this, it suffices to compute, for any two Bohr frequencies $\nu_1,\nu_2\in B(H)$ and $x=\ket{E_i}\bra{E_j}$, the inner products
\begin{align*}
\sqrt{\mathcal{Z}(\beta)}\langle \sqrt{\sigma_\beta},(\delta^\alpha_{\nu_1})^\dagger\delta^\alpha_{\nu_2}(x)\rangle&=e^{-\frac{\beta(\nu_1+\nu_2)}{4}}e^{-\frac{\beta E_j}{2}}\bra{ E_j}(A^\alpha_{\nu_1})^\dagger A^\alpha_{\nu_2}\ket{E_i}\\
&+e^{\frac{\beta(\nu_1+\nu_2)}{4}}e^{-\frac{\beta E_i}{2}}\bra{E_j}A^\alpha_{\nu_2}(A^\alpha_{\nu_1})^\dagger \ket{E_i}\\
&-\sqrt{\mathcal{Z}(\beta)}\,e^{\frac{\beta(\nu_1-\nu_2)}{4}}\bra{E_j}(A^\alpha_{\nu_1})^\dagger \sqrt{\sigma_\beta }A^\alpha_{\nu_2}\ket{E_i}\\
&-\sqrt{\mathcal{Z}(\beta)}\,e^{\frac{\beta(\nu_2-\nu_1)}{4}}\bra{E_j}A^\alpha_{\nu_2}\sqrt{\sigma_\beta}(A^\alpha_{\nu_1})^\dagger\ket{E_i} \\
&=\delta_{E_i+\nu_2,E_j+\nu_1}e^{-\frac{\beta(\nu_1+\nu_2)}{4}}e^{-\frac{\beta E_j}{2}}\bra{E_j}(A^\alpha_{\nu_1})^\dagger P_{E_i+\nu_2}A^\alpha_{\nu_2}\ket{E_i}\\
&+\delta_{E_j-\nu_2,E_i-\nu_1}e^{\frac{\beta(\nu_1+\nu_2)}{4}}e^{-\frac{\beta E_i}{2}}\bra{E_j}A^\alpha_{\nu_2}P_{E_i-\nu_1}(A^\alpha_{\nu_1})^\dagger\ket{E_i}\\
&-\delta_{E_i+\nu_2,E_j+\nu_1}e^{\frac{\beta(\nu_1-\nu_2)}{4}}e^{-\frac{\beta(E_i+\nu_2)}{2}}\bra{E_j}(A^\alpha_{\nu_1})^\dagger P_{E_i+\nu_2}A^\alpha_{\nu_2}\ket{E_i}\\
&-\delta_{E_i-\nu_1,E_j-\nu_2}e^{\frac{\beta(\nu_2-\nu_1)}{4}}e^{\frac{-\beta(E_i-\nu_1)}{2}}\bra{E_j}A^\alpha_{\nu_2}P_{E_i-\nu_1}(A^\alpha_{\nu_1})^\dagger\ket{E_i}=0.
\end{align*}
Therefore, $\langle\sqrt{\sigma_\beta},L_{\widehat{f},H}(x)\rangle=0$ for all $x\in\mathscr{F}$ by the definition of $L_{\widehat{f},H}(x)$. Thus, $L_{\widehat{f},H}$ remains symmetric over $\operatorname{span}(\mathscr{F}\cup\{\sqrt{\sigma_{\beta}}\})$, and negativity follows directly from the observation that $\mathcal{E}_{\widehat{f},H}(\lambda x+\mu\sqrt{\sigma_\beta})=|\lambda|^2\mathcal{E}_{\widehat{f},H}(x)$ for all $x\in\mathscr{F}$ and all $\lambda,\mu\in\mathbb{C}$.
\end{proof}

\noindent In what follows, we denote by $\mathscr{F}_{\sigma_\beta}:=\operatorname{span}(\mathscr{F}\cup\{\sqrt{\sigma_\beta}\})$.
Equivalently,
$\mathscr F_{\sigma_\beta}=\mathscr F+
\mathbb C\sqrt{\sigma_\beta}$ is a subspace of
$\mathscr T_2(\mathcal H)$, whereas $\mathcal F$ is a subspace of
$\mathcal H$. In infinite dimension,
$\mathscr F\subsetneq\mathscr F_{\sigma_\beta}$: every element of
$\mathscr F$ has finite rank, while $\sqrt{\sigma_\beta}$ has
infinite rank. The enlarged initial domain explicitly includes
the Gibbs square root, on which the generator vanishes. 

  By the Friedrichs extension recalled in Lemma~\ref{lem:Friedrichs}, the  quadratic form $\mathcal{E}_{\widehat{f},H}$ over $\mathscr{F}_{\sigma_\beta}\times \mathscr{F}_{\sigma_\beta}$ defined as $\mathcal{E}_{\widehat{f},H}(x,y):=-\langle x,L_{\widehat{f},H}y\rangle$ induces the generator of a strongly continuous semigroup
\cite[Lemma 10.16, Theorem 10.7]{Schmdgen2012}: let $\mathcal{E}$ be a quadratic form on $\mathscr{T}_2(\mathcal{H})_{\mathbb{R}}$, namely a non-negative definite, symmetric bilinear form $D(\mathcal{E})\times D(\mathcal{E})\to\mathbb{R}$, where $D(\mathcal{E})$ is a dense subspace of $\mathscr{T}_2(\mathcal{H})_{\mathbb{R}}$. The form $\mathcal{E}$ is said to be closed if $D(\mathcal{E})$ equipped with the norm $\|x\|_{D(\mathcal{E})}:=\sqrt{\|x\|_2^2+\mathcal{E}(x)}$ is a Hilbert space. $\mathcal{E}$ is said to be closable if it admits a closed extension, i.e., there exists a closed quadratic form $\mathcal{E}'$ such that $D(\mathcal{E})\subset D(\mathcal{E}')$ and $\mathcal{E}'$ coincide with $\mathcal{E}$ on $D(\mathcal{E})\times D(\mathcal{E})$.
In the case of a closed form $(\mathcal{E},D(\mathcal{E}))$, we extend it to a form $ \mathcal{E}':\mathscr{T}_2(\mathcal{H})\to \mathbb{R}_+\cup \{+\infty\} $ by setting $\mathcal{E}'(x)=+\infty$ whenever $x\notin D(\mathcal{E})$. In that case, $\mathcal{E}'$ is lower semicontinuous on $\mathscr{T}_2(\mathcal{H})$ \cite[Proposition 10.1]{Schmdgen2012}: for any sequence $x_n\to x$ in $\mathscr{T}_2(\mathcal{H})$,
\begin{align}\label{eq.lsc}
\mathcal{E}'\big(\lim_{n\to\infty} x_n\big)\le \liminf_{n\to\infty}\mathcal{E}'\big(x_n\big).
\end{align}
In what follows, we do not distinguish between $\mathcal{E}$ and $\mathcal{E}'$. Finally, given a densely defined, closed, non-negative symmetric linear form $\mathcal{E}$ over $\mathscr{T}_2(\mathcal{H})$, a subspace $\mathscr{D}\subseteq D(\mathcal{E})$ is called a form core for $\mathcal{E}$ if for every $x\in D(\mathcal{E})$, there exists a sequence $(x_n)\subseteq \mathscr{D}$ such that $\|x-x_n\|_{D(\mathcal{E})}\to 0$ as $n\to\infty$.

 \begingroup
 \begin{lemma}[Friedrichs extension]\label{lem:Friedrichs}
Let $L$ be a densely defined negative semidefinite  symmetric operator and define $\mathcal{E}(x,y):=-\langle L x,y \rangle$ with $x,y \in D(L)$. Then $\mathcal{E}$ is closable, and $D(L)$ is a form core for its closure.
Conversely, any closed quadratic form $\mathcal{E}$ admits a unique densely defined non-positive self-adjoint operator $L$ on  the real Hilbert space $\mathscr{T}_2(\mathcal{H})_{\mathbb{R}}$ of Hermitian elements in $\mathscr{T}_2$, defined on the set of elements $x\in D(\mathcal{E})$ for which there exists $y\in \mathscr{T}_2(\mathcal{H})_{\mathbb{R}}$ such that, for all $ z\in D(\mathcal{E})$, $\mathcal{E}(x,z)=-\Tr(zy)$. In that case, $L(x)=y$. 
\end{lemma} 
\noindent The above construction readily implies that $\mathcal{E}_{\widehat{f},H}$ is closable. For simplicity of notation, we continue to denote its closure by $\mathcal{E}_{\widehat{f},H}$ over the domain $D(\mathcal{E}_{\widehat{f},H})\times D(\mathcal{E}_{\widehat{f},H})$, so that $\mathscr{F}_{\sigma_\beta}$ is a form core for $\mathcal{E}_{\widehat{f},H}$. By Friedrichs extension, it also admits a unique densely defined non-positive self-adjoint operator on $\mathscr{T}_2(\mathcal{H})_{\mathbb{R}}$, called the generator of $\mathcal{E}_{\widehat{f},H}$. Since the latter extends $L_{\widehat{f},H}$, we also identify the two by abuse of notation. In other words, $L_{\widehat{f},H}$ generates a strongly continuous, symmetric semigroup of contractions over $\mathscr{T}_2(\cH)$ \cite[Proposition 6.14]{Schmdgen2012}, which we denote by $\{e^{\smash{t L_{\widehat{f},H}}}\}_{t\ge 0}$.
 \par\endgroup

\smallskip

In the next section, we argue that the semigroup $\{e^{t L_{\widehat{f},H}}\}_{t\ge 0}$ induces a semigroup of quantum channels over the trace-class operators $\mathscr{T}_1(\mathcal{H})$. This is achieved by leveraging the connection between semigroups of completely positive, trace preserving maps over $\mathscr{T}_1(\mathcal{H})$ and completely Markov semigroups over $\mathscr{T}_2(\mathcal{H})$: given a faithful quantum state $\omega$ on $\mathcal{H}$, we say that a densely defined, closed, non-negative definite, symmetric bilinear form $\mathcal{E}$ is a Dirichlet form (with respect to $\omega$) if it satisfies the following condition: for any self-adjoint $x\in D(\mathcal{E})$, denoting by $x_+$ the positive part of $x$ and $x_\wedge:=x-(x-\sqrt{\omega})_+$, it holds that $x_+,x_\wedge\in D(\mathcal{E})$ and $
\mathcal{E}(x_+),\mathcal{E}(x_\wedge)\le \mathcal{E}(x)$.
A strongly continuous, symmetric semigroup $\{T_t\}_{t\ge 0}$ contraction is said to be Markov (with respect to $\omega$) if and only if $0\le x\le \sqrt{\omega}$ implies that $0\le T_t(x)\le \sqrt{\omega}$ for all $t\ge0$. More generally, let ${\mathscr{M}}_n$ denote the algebra of $n\times n$ matrices acting on $\mathbb{C}^n$. We let $\mathcal{E}^{(n)}$ denote the quadratic form on $\mathscr{T}_2(\mathbb{C}^n\otimes \mathcal{H})_{\mathbb{R}}$, where the state on ${\mathscr{M}}_n$ is taken as the usual trace, given by $D(\mathcal{E}^{(n)}):={\mathscr{M}}_n\otimes D(\mathcal{E})$, with $\mathcal{E}^{(n)}(\{x_{ij}\})=\sum_{i,j=1}^n\mathcal{E}(x_{ij})$.
Then the form $\mathcal{E}$ is said to be $n$-Dirichlet if the form $\mathcal{E}^{(n)}$ is Dirichlet for the underlying state $\omega^{(n)}:=\tau_n\otimes \omega$, where $\tau_n$ denotes the maximally mixed state in ${\mathscr{M}}_n$, and completely Dirichlet if it is $n$-Dirichlet for all $n$. Similarly, the semigroup $\{T_t\}_{t\ge 0}$ is $n$-Markov if for all $0\le x\le \sqrt{\tau_n\otimes \omega}$, $0\le \operatorname{id}_n\otimes T_t(x)\le \sqrt{\tau_n\otimes \omega}$, and completely Markov if it is $n$-Markov for all $n$. Next, we make use of the equivalence between completely Dirichlet forms and completely Markov semigroups \cite[Theorem 5.7]{Goldstein1995} to prove the complete Markov property of the semigroup generated by $L_{\widehat{f},H}$:

\begin{proposition}\label{GibbsDirichlet}
The closed quadratic form $\mathcal E_{\widehat f,H}$ is completely Dirichlet.
Moreover, the strongly continuous, symmetric semigroup
  $\{e^{\smash{t L_{\widehat f,H}}}\}_{\smash{t\geq 0}}$  is completely Markov with
respect to the state $\sigma_\beta$.
\end{proposition}

\begin{proof}
From the previous discussion, it suffices to prove that the form
$\mathcal E_{\smash{\widehat f,H}}$ is completely Dirichlet. We first prove the
claim on the form core
$\mathscr F_{\sigma_\beta}$ 
and then pass to the closure by lower
semicontinuity (cf. Equation \eqref{eq.lsc}). Recall that, on the core, the generator is given by $L_{\widehat f,H}=K+T$, where
  \[
\begin{aligned}
K(x)
:={}&
-\sum_{\alpha,\nu_1,\nu_2}
\frac{\overline{\widehat f(\nu_1)}\widehat f(\nu_2)}{2\cosh(\beta(\nu_1-\nu_2)/4)}
\Big[
(A^\alpha_{\nu_1})^\dagger A^\alpha_{\nu_2}x
+
e^{\beta(\nu_1+\nu_2)/2}
xA^\alpha_{\nu_2}(A^\alpha_{\nu_1})^\dagger
\Big],
\end{aligned}
\] 
and
  \[
\begin{aligned}
T(x)
:={}&
\sum_{\alpha,\nu_1,\nu_2}
\frac{\overline{\widehat f(\nu_1)}\widehat f(\nu_2)}{2\cosh(\beta(\nu_1-\nu_2)/4)}
\Big[
e^{\beta\nu_2/2}
(A^\alpha_{\nu_1})^\dagger xA^\alpha_{\nu_2}
+
e^{\beta\nu_1/2}
A^\alpha_{\nu_2}x(A^\alpha_{\nu_1})^\dagger
\Big].
\end{aligned}
\] 
We now show that $T$ is completely positive. Since the set of bare jumps is
closed under adjoints, for every $\alpha$ there is a label $\bar\alpha$ such
that $A^{\bar\alpha}=(A^\alpha)^\dagger$, $(A^\alpha_\nu)^\dagger=A^{\bar\alpha}_{-\nu}$.
Using this identity in the second term of $T$, and reindexing $\alpha\mapsto \bar\alpha$,
$\nu_1\mapsto -\nu_2$, $\nu_2\mapsto -\nu_1$, we obtain
  \[
T(x)
=
\sum_{\alpha,\nu_1,\nu_2}
\frac{\overline{\widehat f(\nu_1)}\widehat f(\nu_2)}{2\cosh(\beta(\nu_1-\nu_2)/4)}
\left(e^{\beta\nu_1/2}+e^{\beta\nu_2/2}\right)
(A^\alpha_{\nu_1})^\dagger xA^\alpha_{\nu_2}.
\] 
Define $g(\nu):=e^{\beta\nu/4}\widehat f(\nu)$. The KMS symmetry   $\overline{\widehat f(\nu)}=\widehat f(-\nu)e^{-\beta\nu/2}$ implies $g(-\nu)=\overline{g(\nu)}$ . Therefore,
  \[
\frac{\overline{\widehat f(\nu_1)}\widehat f(\nu_2)}{2\cosh(\beta(\nu_1-\nu_2)/4)}
\left(e^{\beta\nu_1/2}+e^{\beta\nu_2/2}\right)
=
\overline{g(\nu_1)}g(\nu_2).
\] 
Indeed, writing $s=\frac{\beta(\nu_1-\nu_2)}{4}$, one has
  \[
\begin{aligned}
\overline{\widehat f(\nu_1)}\widehat f(\nu_2)
\left(e^{\beta\nu_1/2}+e^{\beta\nu_2/2}\right)& =
\overline{g(\nu_1)}g(\nu_2)
e^{-\beta(\nu_1+\nu_2)/4}
\left(e^{\beta\nu_1/2}+e^{\beta\nu_2/2}\right)
\\
& =
\overline{g(\nu_1)}g(\nu_2)
\left(e^s+e^{-s}\right)
=
2\overline{g(\nu_1)}g(\nu_2)\cosh(s).
\end{aligned}
\] 
Therefore
  \[
T(x)
=
\sum_{\alpha,\nu_1,\nu_2}
\overline{g(\nu_1)}g(\nu_2)
(A^\alpha_{\nu_1})^\dagger xA^\alpha_{\nu_2}.
\] 
Equivalently, if $R^\alpha:=\sum_{\nu\in  B(H)}g(\nu)A^\alpha_\nu$, then, on the core,
$T(x)=\sum_{\alpha\in\mathcal A}(R^\alpha)^\dagger xR^\alpha$. Thus $T$ is completely positive.  Finite spectral/frequency cutoffs and Lemma~\ref{lem:definL} justify the same identities on the amplified core.  We now prove the Markov property. Let $n\in\mathbb N$ and let
$\mathbf{x}\in \mathscr{M}_n\otimes \mathscr{F}_{\sigma_\beta}$ be self-adjoint. Write $\mathbf{x}=\mathbf{x}_+-\mathbf{x}_-$
for its positive and negative parts. Since
$L_{\widehat f,H}$ is negative semidefinite and symmetric, we have
  \[
\begin{aligned}
\mathcal E_{\widehat f,H}^{(n)}(\mathbf{x})
-
\mathcal E_{\widehat f,H}^{(n)}(\mathbf{x}_+)
={}&
-\langle \mathbf{x}_-,(\mathrm{id}_n\otimes L_{\widehat f,H})(\mathbf{x}_-)\rangle
\\
&+
\langle \mathbf{x}_+,(\mathrm{id}_n\otimes L_{\widehat f,H})(\mathbf{x}_-)\rangle
\\
&+
\langle \mathbf{x}_-,(\mathrm{id}_n\otimes L_{\widehat f,H})(\mathbf{x}_+)\rangle
\\
\geq{}&
\langle \mathbf{x}_+,(\mathrm{id}_n\otimes L_{\widehat f,H})(\mathbf{x}_-)\rangle
\\
&+
\langle \mathbf{x}_-,(\mathrm{id}_n\otimes L_{\widehat f,H})(\mathbf{x}_+)\rangle .
\end{aligned}
\] 
The  following trace identities are justified first under finite cutoffs and then by Lemma~\ref{lem:definL}.
The  contribution of $K$ to the last two terms vanishes. Indeed, $K$ is a
sum of left and right multiplication operators. Hence, using cyclicity of the
trace and $\mathbf{x}_+\mathbf{x}_-=\mathbf{x}_-\mathbf{x}_+=0$, we get
\[
\langle \mathbf{x}_+,(\mathrm{id}_n\otimes K)(\mathbf{x}_-)\rangle
=
\langle \mathbf{x}_-,(\mathrm{id}_n\otimes K)(\mathbf{x}_+)\rangle
=
0.
\]
Consequently,
\[
\begin{aligned}
\mathcal E_{\widehat f,H}^{(n)}(\mathbf{x})
-
\mathcal E_{\widehat f,H}^{(n)}(\mathbf{x}_+)
\geq{}&
\langle \mathbf{x}_+,(\mathrm{id}_n\otimes T)(\mathbf{x}_-)\rangle+
\langle \mathbf{x}_-,(\mathrm{id}_n\otimes T)(\mathbf{x}_+)\rangle .
\end{aligned}
\]
Since $T$ is completely positive and $\mathbf{x}_+,\mathbf{x}_-\geq 0$, both terms on the
right-hand side are nonnegative. Hence $\mathcal E_{\widehat f,H}^{(n)}(\mathbf{x}_+)
\leq
\mathcal E_{\widehat f,H}^{(n)}(\mathbf{x})$.   For the upper cut-off, with $\mathbf a=\sigma_\beta^{(n)/2}$ and $\mathbf y=\mathbf a-\mathbf x$, one has $\mathbf x_\wedge=\mathbf a-\mathbf y_+$ and, since $\mathbf a$ lies in the kernel of the form, $\mathcal E^{(n)}(\mathbf x_\wedge)=\mathcal E^{(n)}(\mathbf y_+)\le\mathcal E^{(n)}(\mathbf y)=\mathcal E^{(n)}(\mathbf x)$. By lower semicontinuity and the form-core property both inequalities pass to the closure.  Hence $\mathcal E_{\widehat f,H}$ is completely Dirichlet   and the  associated symmetric semigroup is  completely Markov with respect to $\sigma_\beta$.
\end{proof}

\subsection{From Hilbert to Schr\"{o}dinger}\label{appendixHilberttoSchro}

 \begingroup
 Next, we show how to define a semigroup of quantum channels on the space $\mathscr{T}_1(\mathcal{H})$ of trace-class operators from the completely Markov semigroup constructed in Proposition \ref{GibbsDirichlet}. We start by deriving a simple density argument: in what follows, given a faithful quantum state $\omega$ on $\mathcal{H}$ and two real numbers $\lambda\le  \mu$, we denote by $B_{[\lambda,\mu]}(\omega)$ the set of self-adjoint, trace-class operators $x$ with $\lambda \,\omega \le x\le \mu\,\omega$.

\begin{lemma}\label{densityT1T2}
Let $\omega$ be a faithful quantum state on $\mathcal{H}$. Then the set of positive semidefinite, trace-class operators $\mathscr{T}_1(\mathcal{H})_+$ satisfies
\begin{align*}
\mathscr{T}_1(\mathcal{H})_+=\overline{\bigcup_{\lambda\in\mathbb{R}_+}B_{[0,\lambda]}(\omega)}^{\|\cdot \|_1}\,.
\end{align*}
\end{lemma}

\begin{proof}

 Denote by $P_K$ the projection onto the subspace corresponding to  the first $K$ largest eigenvalues of $\omega$. Then, for any  positive semidefinite, trace-class operator $a$,        $0\le P_Ka P_K\le \|a\|P_K\le  \frac{\|a\|}{\lambda_K}\omega$  , where $\lambda_K$  denotes the $K$-th largest eigenvalue of $\omega$. Moreover, we can  clearly see that $\|P_Ka P_K-a\|_1\to 0$ as $K\to\infty$. The  result follows.
\end{proof}

\noindent  Next, we consider the completely Dirichlet form $\mathcal{E}_{\widehat{f},H}$ over $\mathscr{T}_2(\mathcal{H})_{\mathbb{R}}$ with the associated generator $L_{\widehat{f},H}$.
Using the non-commutative Radon-Nikodym theorem (see \cite[Lemma 2.2c]{schmitt1986radon} or \cite[Lemma 1.5]{Goldstein1995}), for any $0\le y\le \sigma_\beta$, there exists $X\in \mathscr{B}(\mathcal{H})_+$  , the set of
positive semidefinite bounded operators on $\cH$, 
 with $\|X\|\le 1$ and $y=\sigma_\beta^{\smash{\scriptstyle 1/2}}X\sigma_\beta^{\smash{\scriptstyle 1/2}}$. We denote $x=\sigma_\beta^{\smash{\scriptstyle 1/4}}X\sigma_\beta^{\smash{\scriptstyle 1/4}}\in\mathscr{T}_2(\mathcal{H})_+$, so that $y=\sigma_{\beta}^{\smash{1/4}}x\sigma_\beta^{\smash{1/4}}$. Uniqueness of $X$ (and thus $x$) follows from the fact that $\sigma_\beta^{\smash{1/2}}X\sigma_\beta^{\smash{1/2}}=\sigma_\beta^{\smash{1/2}}Y\sigma_\beta^{\smash{1/2}}\Rightarrow X=Y$ by the faithfulness of $\sigma_\beta$. Next, given the continuous embedding 
\begin{align*}
\iota_2:\mathscr{T}_2(\mathcal{H})\to\mathscr{T}_1(\mathcal{H})\,,\qquad \iota_2(x)=\sigma_\beta^{\frac{1}{4}}x\sigma_\beta^{\frac{1}{4}},
\end{align*}
we can define induced maps on the trace-class operators as
\begin{align}\label{eqcLtoL}
\Phi_t(y)=\iota_2\circ e^{tL_{\widehat{f},H}}(x).
\end{align}
By the complete Markov property, 
\begin{align*}
0\le x\le \sqrt{ \sigma_\beta}\quad \Rightarrow\quad  0\le e^{tL_{\widehat{f},H}}(x)\le \sqrt{ \sigma_\beta}\quad \Rightarrow \quad 0\le  \Phi_t(y)\le \sigma_\beta.
\end{align*}

\begin{lemma}\label{T1semigroupgeneration}
The maps $\Phi_t$ defined above can be extended by linearity to a strongly continuous semigroup of uniformly bounded, completely positive, trace preserving linear maps on      $\mathscr{T}_1(\mathcal{H})$, preserving its self-adjoint real subspace,   with $\sup_{t\ge 0}\|\Phi_t\|_{1\to 1}\le 1$ and $\Phi_t(\sigma_\beta)=\sigma_\beta$ for all $t\ge 0$.
\end{lemma}
\begin{proof}
  For $y=\iota_2(x)$, Hilbert--Schmidt symmetry and
$e^{tL_{\widehat f,H}}\sqrt{\sigma_\beta}=\sqrt{\sigma_\beta}$ give
\[
\Tr\bigl(\iota_2(e^{tL_{\widehat f,H}}x)\bigr)
=\langle\sqrt{\sigma_\beta},e^{tL_{\widehat f,H}}x\rangle_2
=\langle\sqrt{\sigma_\beta},x\rangle_2=\Tr(y).
\] 
  In particular, for $0\le y\le\lambda\sigma_\beta$, positivity and
trace preservation imply $\|\Phi_t(y)\|_1=\|y\|_1$.
Lemma~\mbox{%
\ref{densityT1T2} }\hskip0pt%
supplies density of this order ideal in
the positive trace class.

To obtain contractivity on the full complex trace class, let $P_K$
be a finite spectral projection of $\sigma_\beta$. Every positive
operator supported on $P_K\mathcal H$ is dominated by a multiple
of $\sigma_\beta$, so \mbox{%
\eqref{eqcLtoL} }\hskip0pt%
defines a positive,
trace-preserving map on this finite matrix algebra. Applying the
complete Markov property at every matrix level shows that the map
is completely positive. Its adjoint is unital and completely
positive, hence contractive in operator norm; by trace duality,
\[
\|\Phi_t(y)\|_1\le\|y\|_1,
\qquad y=P_KyP_K.
\] 
  The maps on these nested matrix algebras agree by their defining
formula. Since $P_KyP_K\to y$ in trace norm for every trace-class
$y$, they extend uniquely to a completely positive, trace-preserving
contraction on $\mathscr T_1(\mathcal H)$. Positivity also ensures
preservation of the self-adjoint real subspace.

The intertwining identity holds for every $x\in\mathscr T_2(\mathcal H)$:
\[
\Phi_t(\iota_2(x))=\iota_2(e^{tL_{\widehat f,H}}x).
\] 
  Indeed, approximate $x$ by $P_KxP_K$ in Hilbert--Schmidt norm and
use the continuity of $\iota_2$, $e^{tL_{\widehat f,H}}$ and $\Phi_t$.
The semigroup property and strong continuity now follow on the
dense range of $\iota_2$ from the Hilbert--Schmidt semigroup, and
extend to all of $\mathscr T_1$ by contractivity. Finally,
\[
\Phi_t(\sigma_\beta)
=\iota_2(e^{tL_{\widehat f,H}}\sqrt{\sigma_\beta})
=\sigma_\beta.
\] 
 \end{proof}

\noindent
Note that Lemma~\ref{T1semigroupgeneration} alone imply that there is a completely positive,
trace-preserving projection $\Pi$ on $\mathscr T_1(\mathcal H)$ such that
\[
\lim_{t\to\infty}\|\Phi_t(y)-\Pi(y)\|_1=0,
\qquad y\in\mathscr T_1(\mathcal H),
\]
whose range is exactly the stationary trace-class operators. Indeed, the
spectral theorem gives $e^{tL_{\widehat f,H}}x\to P_0x$ in
$\mathscr T_2$, where $P_0$ projects onto $\ker L_{\widehat f,H}$.
On the dense range of $\iota_2$ set
$\Pi(\iota_2x)=\iota_2(P_0x)$. The intertwining identity and
$\|\iota_2x\|_1\le\|x\|_2$ give convergence there; contractivity
extends both the limit and convergence to all of $\mathscr T_1$.
Complete positivity and trace preservation pass to the limit, and the
semigroup property identifies its range and shows $\Pi^2=\Pi$.
If the nonzero spectrum has gap $\lambda>0$, then additionally
\[
\|\Phi_t(\iota_2x)-\Pi(\iota_2x)\|_1
\le e^{-\lambda t}\|(I-P_0)x\|_2.
\]
The same conclusions apply to the Gaussian-convoluted generators below.

 \smallskip
 \par\endgroup

\noindent Back to the generator $L_{\smash{\widehat{f},H}}$ introduced in \Cref{def:L}, we aim to obtain a Lindblad form for the generator of the associated semigroup $\{\Phi_t\}_{t\ge 0}$. For this, we formally introduce the jumps and drift operators:
\begin{align}
\label{eq:LalphainGenerationTheory}
L^\alpha:&=\sum_{\nu\in B(H)}\widehat{f}(\nu)A^\alpha_\nu=\sum_{E,E'\in\spec(H) }\widehat{f}(E'-E)P_{E'}A^\alpha P_{E},\qquad \text{for }\alpha\in\mathcal{A},\\
\nn G:&=-\sum_{\substack{\alpha\in\mathcal{A}\\\nu_1,\nu_2\in B(H)}}\, \frac{\overline{\widehat{f}(\nu_1)}\widehat{f}(\nu_2)\,e^{\frac{\beta(\nu_1-\nu_2)}{4}}}{2\cosh((\nu_1-\nu_2)\beta/4)}\,(A^\alpha_{\nu_1})^\dagger A^\alpha_{\nu_2} \\\label{eq:GinGenerationTheore}&=-\sum_{\substack{\alpha\in \mathcal{A}\\E,E',E''\in\spec(H)}}  \frac{\overline{\widehat{f}(E'-E'')}\widehat{f}(E'-E)\,e^{\frac{\beta(E-E'')}{4}}}{2\cosh((E-E'')\beta/4)}P_{E''}\left(A^\alpha\right)^\dagger P_{E'} A^\alpha P_E .
\end{align}
It is clear that these operators are well-defined on the domain $\mathcal{F}.$ In practice, they can, in fact, however, directly be defined on the larger domain $D(\widetilde H^\kappa)$ for certain $\kappa \ge 0$, as we show in Lemma~\ref{lem.LalphaGwelldefinedF} below.
To see this, we introduce the functions\footnote{From \Cref{eq:condAalphas} we know that $\widehat f$ is uniformly bounded. In this case, the function $F_2$ can be bounded by the expression $F_2(E)\lesssim F_1(E)\sum_{E''\in\spec(H)}\frac{1}{1+e^{\beta(E''-E)/2}},$, which is often easier to analyze in practice.}
\begin{align}
\label{eq:defF_1F_2}
    F_1(E):= \sum_{E'\in\spec(H)} |\widehat     f(E'-E)|,\quad\text{and}\quad F_2(E):=\sum_{E',E''\in\spec(H)}      \frac{|\widehat{f}(E'-E'')\widehat{f}(E'-E)|}{1+e^{\beta(E''    -E)/2}}.
\end{align}
for $E\in\spec(H).$ Note that from \Cref{eq:condAalphas} we easily see that $F_1(E)<\infty$ for all $E\in\spec(H).$ In practice, these are well-behaved and growing polynomially in $E$.
\begin{framed}
\begin{conditionB}\label{eq:F_1Kill}
There exists $\kappa_1>\gamma$ such that the following sums converge:
\begin{align}
\label{eq:kappa_1}
\sum_{E\in\spec(H)} F_1(E)(1+h_0+E)^{-(\kappa_1-\gamma)}&<\infty\, \text{ and }\,  \sum_{E\in\spec(H)} F_2(E) \left(1+h_0+E\right)^{-(\kappa_1-\gamma)}<\infty.
\end{align}
\end{conditionB}
\end{framed}
\begin{lemma}
\label{lem.LalphaGwelldefinedF}
For $\mu\ge \gamma\ge 0$ being defined in \Cref{eq:condAalphas},
let $\kappa_1> \gamma$ be such that \Cref{eq:F_1Kill} holds. Then the operators $L^\alpha$ are well-defined on the dense domain $D(\widetilde H^{\kappa_1})$ for all $\alpha\in\mathcal{A}$ and relatively $\widetilde H^{\kappa_1}$-bounded. Furthermore,  $G$ is well-defined on the dense domain $D(\widetilde H^{\kappa_1+\mu-\gamma})$ and relatively $\widetilde H^{\kappa_1+\mu-\gamma}$-bounded.
\end{lemma}
\begin{proof}
As by \Cref{eq:condAalphas} we have that $P_E'A^\alpha P_E$ is bounded and, in particular, well-defined on $D(\widetilde H^{\kappa_1})$, it remains to show that the sum on the right hand side of the definition $L^\alpha$ in \eqref{eq:LalphainGenerationTheory} when applied on $\ket{\psi}\in D(\widetilde H^{\kappa_1})$ has finite Hilbert space norm\footnote{More precisely, we show absolute convergence of $\sum_{E,E'\in\spec{H}}|\widehat f(E'-E)|\|P_{E'}A^\alpha P_{E}\ket{\psi}\|.$ Hence, the sequence $\ket{\phi_{\tilde E}} := \sum_{\substack{E,E'\in\spec(H)\\E,E'\le \tilde E}} \widehat{f}(E'-E)P_{E'}A^\alpha P_{E}\ket{\psi}$ converges in $\cH$ and we define $L^\alpha\ket{\psi}:=\lim_{\tilde E}\ket{\phi_{\tilde E}}.$}.
For $\ket{\psi} \in D(\widetilde H^{\kappa_1})$ we see
\begin{align*}
    \left\|L^\alpha \ket{\psi}\right\|&\le     \sum_{E,E'\in\spec{H}}|\widehat f(E'-E)|\|P_{E'}A^\alpha     P_{E}\ket{\psi}\|  \\&\le \|A^{\alpha}\widetilde H^{-\gamma}\|     \left\|\widetilde H^{\kappa_1}\ket{\psi}\right\|\sum_{E\in    \spec(H)}F_1(E) \left(1+h_0+E\right)^{-(\kappa_1-\gamma)}     \lesssim \left\|\widetilde H^{\kappa_1}\ket{\psi}\right\|
\end{align*}
where we used \Cref{eq:condAalphas} on the jumps $A^\alpha$ and \Cref{eq:F_1Kill}.
We can argue similarly for the $G$ operator as for $\ket{\psi}\in D(\widetilde H^{\kappa_1+\mu-\gamma})$ we have
\begin{align*}
\|G\ket{\psi}\|&\le \sum_{\substack{\alpha\in \mathcal{A}\\E,E',E''\in\spec(H)}}  \frac{|\widehat{f}(E'-E'')\widehat{f}(E'-E)|\,e^{\frac{\beta(E-E'')}{4}}}{2\cosh((E-E'')\beta/4)}\left\|P_{E''}\left(A^\alpha\right)^\dagger P_{E'} A^\alpha P_E\ket{\psi}\right\| \\&\le |\mathcal{A}|
\max_{\alpha\in\mathcal{A}}\Bigl(
\left\|\widetilde H^{-\gamma}A^\alpha\right\|\left\|\widetilde H^{\gamma}A^\alpha \widetilde H^{-\mu}\right\|
\Bigr) \left\|\widetilde H^{\kappa_1+\mu-\gamma}\ket{\psi}\right\|\sum_{E\in\spec(H)} F_2(E) \left(1+h_0+E\right)^{-(\kappa_1-\gamma)}\\&\lesssim \left\|\widetilde H^{\kappa_1+\mu-\gamma}\ket{\psi}\right\|.
\end{align*}
\end{proof}

\noindent Next, we provide an expression for the generator $\cL_{\widehat{f},H}$ of the semigroup $\{\Phi_t\}_{t\ge 0}$ on $\mathscr{F}$.
\begin{proposition}\label{propDirichlettoSchro}
The generator $\mathcal{L}_{\widehat{f},H}$ of the semigroup $\{\Phi_t\}_{t\ge 0}$ on $\mathscr{T}_1(\mathcal{H})$ satisfies  $\mathscr{F}\subset D(\mathcal{L}_{\widehat{f},H})$. Moreover, given any two energy eigenstates $\ket{E_i},\ket{E_j}$, the generator $\mathcal{L}_{\smash{\widehat{f},H}}$ associated with the semigroup $\{\Phi_t\}_{t\ge 0}$ evaluated at $\ket{E_i}\bra{E_j}$ takes the form
\begin{align}\label{eq:cLonF}
\mathcal{L}_{\widehat{f},H}(\ket{E_i}\bra{E_j})\!=\sum_{\alpha\in\mathcal{A}}L^\alpha \ket{E_i}(L^\alpha\ket{E_j})^\dagger+G\ket{E_i}\bra{E_j}+\ket{E_i}(G\ket{E_j})^\dagger. 
\end{align}
Following standard notations, we denote $\Phi_t:=e^{t\mathcal{L}_{\widehat{f},H}}$. 
\end{proposition}

\begin{proof}
In order to derive \eqref{eq:cLonF}, it suffices to consider $y=\ket{E_j}\bra{E_i}$, for any two energies $E_i,E_j\in\operatorname{Sp}(H)$. It is clear that there exists a unique $x=\lambda_{i,j}y\in\mathscr{T}_2(\mathcal{H})$, with   $\lambda_{i,j}=\mathcal{Z}(\beta)^{1/2}e^{\smash{\beta(E_j+E_i)/4}}$ , such that $y=\iota_2(x)$. Therefore, since $x\in D(L_{\widehat{f},H})$, we have that $y\in D(\mathcal{L}_{\widehat{f},H})$ and, by the continuity of the embedding $\iota_2$,
\begin{align*}
\mathcal{L}_{\widehat{f},H}(y)=\lim_{t\to 0}^{\mathscr{T}_1(\mathcal{H})}\frac{\Phi_t(y)-y}{t}=\iota_2\Big(\lim^{\mathscr{T}_2(\mathcal{H})}_{t\to 0}\frac{e^{tL_{\widehat{f},H}}(x)-x}{t}\Big)=\iota_2\circ L_{\widehat{f},H}(x).
\end{align*}
Equation \eqref{eq:cLonF} follows by the definition \eqref{def:L} of $L_{\widehat{f},H}$ and direct computation: first, we directly get that
\begin{align*}
\iota_2\circ L_{\widehat{f},H}(x)=-\sum_{\substack{\nu_1,\nu_2\in B(H)\\\alpha\in\mathcal{A}}}\frac{\overline{\widehat{f}(\nu_1)}\widehat{f}(\nu_2)e^{\frac{\beta(\nu_1+\nu_2)}{4}}}{2\cosh((\nu_1-\nu_2)\beta/4)}\,x^\alpha_{\nu_1,\nu_2}\equiv L_1+L_2+L_3+L_4,
\end{align*}
where each $L_j$ is a sum corresponding to one of the four elements in the decomposition of $x^\alpha_{\nu_1,\nu_2}$ below:
\begin{align}\label{eq:4termstomassage}
x^\alpha_{\nu_1,\nu_2}:=e^{-\frac{\beta\nu_2}{2}}(A^\alpha_{\nu_1})^\dagger A^\alpha_{\nu_2}y+e^{\frac{\beta\nu_2}{2}}yA^\alpha_{\nu_2}(A^\alpha_{\nu_1})^\dagger-e^{-\frac{\beta\nu_2}{2}}A^\alpha_{\nu_2}y(A^\alpha_{\nu_1})^\dagger-e^{\frac{\beta\nu_2}{2}}(A^\alpha_{\nu_1})^\dagger yA^\alpha_{\nu_2}.
\end{align}
Although this is obvious by construction, by an almost identical analysis to the one done in the proof \Cref{lem:definL}, we can also verify that each of the four sums above defines an element in $\mathscr{T}_1(\mathcal{H})$ by hand. Next, we consider the sums associated with the second and fourth terms in \eqref{eq:4termstomassage}: 
\begin{align*}
L_2=-\sum_{\substack{\nu_1,\nu_2\in B(H)\\\alpha\in\mathcal{A}}}\,\frac{\overline{\widehat{f}(\nu_1)}\widehat{f}(\nu_2)e^{\frac{\beta(\nu_1+\nu_2)}{4}}}{2\cosh((\nu_1-\nu_2)\beta/4)}\,e^{\frac{\beta\nu_2}{2}}yA^\alpha_{\nu_2}(A^\alpha_{\nu_1})^\dagger.
\end{align*}
\Cref{eq:condAalphas} implies that $\overline{\widehat{f}(-\nu_2)}\widehat{f}(-\nu_1)e^{\frac{-\beta(\nu_1+\nu_2)}{4}}=\widehat{f}(\nu_2)\overline{\widehat{f}(\nu_1)}e^{\frac{\beta(\nu_1+\nu_2)}{4}}$. Thus, by change of variables $\nu_1\leftarrow -\nu_2$, $\nu_2\leftarrow -\nu_1$, the invariance of $\{A^\alpha\}_{\alpha\in\mathcal{A}}$ under adjoints further implies that 
  \begin{align*}
L_2=-\sum_{\substack{\nu_1,\nu_2\in B(H)\\\alpha\in\mathcal{A}}}\,\frac{\overline{\widehat{f}(\nu_1)}\widehat{f}(\nu_2)e^{\frac{\beta(\nu_1+\nu_2)}{4}}}{2\cosh((\nu_1-\nu_2)\beta/4)}\,e^{-\frac{\beta\nu_1}{2}}y(A^\alpha_{\nu_1})^\dagger A^\alpha_{\nu_2}.
\end{align*} 
Similarly, we get that 
\begin{align*}
L_4=\sum_{\substack{\nu_1,\nu_2\in B(H)\\\alpha\in\mathcal{A}}}\frac{\overline{\widehat{f}(\nu_1)}\widehat{f}(\nu_2)e^{\frac{\beta(\nu_1+\nu_2)}{4}}}{2\cosh((\nu_1-\nu_2)\beta/4)}\,e^{-\frac{\beta\nu_1}{2}}A^\alpha_{\nu_2} y(A^\alpha_{\nu_1})^\dagger.
\end{align*}
With these observations, we can rewrite the generator in the Schr\"{o}dinger picture as
\begin{align*}
\mathcal{L}_{\widehat{f},H}(y)=\sum_{\substack{\nu_1,\nu_2\\\alpha\in\mathcal{A}}}\overline{\widehat{f}(\nu_1)}\widehat{f}(\nu_2)\,\left(A^\alpha_{\nu_2} y(A^\alpha_{\nu_1})^\dagger-\frac{e^{\frac{\beta(\nu_1-\nu_2)}{4}}(A^\alpha_{\nu_1})^\dagger A^\alpha_{\nu_2}y+e^{\frac{\beta(\nu_2-\nu_1)}{4}}y(A^\alpha_{\nu_1})^\dagger A^\alpha_{\nu_2}}{2\cosh((\nu_1-\nu_2)\beta/4)}\right),
\end{align*}
as claimed in \eqref{eq:cLonF}. 

\end{proof}

\noindent Our next goal is to show that the domain of $\mathcal{L}_{\widehat{f},H}$ contains a larger set of energy constrained quantum states. For this, as well as in order to streamline the analysis of the next sections, we make use of a simple tool defined in \cite{gondolf2024energy}: given $\delta_1,\delta_2\ge 0$, we introduce the quantum Sobolev spaces defined on 
\begin{align*}
D(\mathcal{W}_H^{\delta_1,\delta_2}):=\Big\{\widetilde{H}^{-\delta_1}a \widetilde{H}^{-\delta_2}\,\Big|\,a\in\mathscr{T}_1(\mathcal{H})\Big\}\text{ with norms }
\|x\|_{\mathcal{W}_H^{\delta_1,\delta_2}}:=\|\mathcal{W}_H^{\delta_1,\delta_2}(x)\|_1,
\end{align*}
where $\mathcal{W}^{\delta_1,\delta_2}_H(x):=\widetilde{H}^{\delta_1} x \widetilde{H}^{\delta_2}$. Since the inverse $(\mathcal{W}_H^{\delta_1,\delta_2})^{-1}$ is bounded, $\mathcal{W}^{\delta_1,\delta_2}_H$ is closed and $(D(\mathcal{W}_H^{\delta_1,\delta_2}),\|\cdot \|_{\smash{\mathcal{W}_H^{\delta_1,\delta_2}}})$ is a Banach space equivalent to the domain $D(\mathcal{W}^{\delta_1,\delta_2}_{\smash{H}})$ endowed with the graph norm of $\mathcal{W}^{\delta_1,\delta_2}_{\smash{H}}$
\cite[Theorem 2.2]{gondolf2024energy}. Moreover, for any $E\ge 0$
\[
\mathscr{S}_E(\widetilde{H}^{2\delta})
:=\bigl\{\, \rho \in \mathscr{S}(\mathcal{H}) \;\big|\; \Tr(\widetilde{H}^{2\delta}\rho) \le E \,\bigr\}\subset \mathscr{S}(\mathcal{H})\cap D(\mathcal{W}_H^{\delta,\delta}).
\]

\begingroup
 \begin{lemma}\label{mathFcoreW}
For any      $\delta_1,\delta_2\ge0$  , $\mathscr{F}$ is a
  graph-norm 
 core for $\mathcal{W}_H^{\delta_1,\delta_2}$
   on $\mathscr T_1(\mathcal H)$ .
\end{lemma}
\begin{proof}
For any $a\in D(\mathcal{W}_H^{\delta_1,\delta_2})$, there is $x\in\mathscr{T}_1(\mathcal{H})$ such that $a:=(\mathcal{W}_H^{\delta_1,\delta_2})^{-1}(x)$. We truncate both $x$ and $a$ by the projection $P_n$ onto the subspace spanned by the eigenstates $\ket{E_0},\cdots,\ket{E_{n}}$: $x_n=P_nxP_n$ and $a_n=P_naP_n$. We have by construction that $a_n=(\mathcal{W}_H^{\delta_1,\delta_2})^{-1}(x_n)$. By the proof of \Cref{densityFFs}, we have that $a_n\to a$ and $x_n\to x$.
\end{proof}
 \par\endgroup

 This core property concerns the graph norm of
$\mathcal W_H^{\delta_1,\delta_2}$ on $\mathscr T_1(\mathcal H)$.
The form core used in the Friedrichs construction of
$\mathcal E_{\widehat f,H}$ on $\mathscr T_2(\mathcal H)$ is
$\mathscr F_{\sigma_\beta}$. 
Next, for any two operators $(L_1,D(L_1))$, $(L_2,D(L_2))$ that are relatively $\widetilde{H}^{\delta_1}$-bounded and relatively $\widetilde{H}^{\delta_2}$-bounded, the operators $L_1\widetilde{H}^{-\delta_1}$ and $L_2\widetilde{H}^{-\delta_2}$ can be extended into bounded operators, such that for any $x\in D(\mathcal{W}_H^{\delta_1,\delta_2})$, the trace-class operator 
\begin{align*}
L_1\cdot  x\cdot L_2^\dagger:=L_1\widetilde{H}^{-\delta_1}(L_2\widetilde{H}^{-\delta_2}(\mathcal{W}_H^{\delta_1,\delta_2}(x))^\dagger)^\dagger
\end{align*}
satisfies
\begin{align*}
\|L_1\cdot x\cdot L_2^\dagger\|_1\le \|L_1\widetilde{H}^{-\delta_1}\|\,\|L_2\widetilde{H}^{-\delta_2}\|\,\|x\|_{\mathcal{W}^{\delta_1,\delta_2}_H}. 
\end{align*}
In what follows, we also denote $I\cdot x\cdot L^\dagger_2\equiv x\cdot L^\dagger_2$ and $L_1\cdot x\cdot I\equiv  L_1\cdot x$.
\begin{proposition}\label{prop.generationtheorem} Under \Cref{eq:condAalphas} and \Cref{eq:F_1Kill}, we have $ D(\mathcal{W}_H^{\kappa_1,\kappa_1})\cap D(\mathcal{W}_H^{\kappa_1+\mu-\gamma,0})\cap D(\mathcal{W}_H^{0,\kappa_1+\mu-\gamma}) \subseteq D(\mathcal{L}_{\widehat{f},H})$, and for any $x\in D(\mathcal{W}_H^{\kappa_1,\kappa_1})\cap D(\mathcal{W}_H^{\kappa_1+\mu-\gamma,0})\cap D(\mathcal{W}_H^{0,\kappa_1+\mu-\gamma})$,
\begin{align}\label{eq.LindbladformL}
\mathcal{L}_{\widehat{f},H}(x)=\sum_{\alpha\in\mathcal{A}}L^\alpha\cdot x\cdot (L^\alpha)^\dagger+G\cdot x+x\cdot G^\dagger.
\end{align}
\end{proposition}
\begin{proof}
We use that $\mathcal{L}_{\smash{\widehat{f},H}}$ is closed in $\mathscr{T}_{\smash{1}}(\mathcal{H})$ by the property of being the generator of the strongly continuous semigroup $\{\Phi_t\}_{t\ge 0}$. Therefore, it suffices to show that for any $x\in D(\mathcal{W}_H^{\kappa_1,\kappa_1})\cap D(\mathcal{W}_H^{\kappa_1+\mu-\gamma,0})\cap D(\mathcal{W}_H^{0,\kappa_1+\mu-\gamma})$ there is a sequence $x_n\to x$ in $\mathscr{T}_1(\mathcal{H})$ and $y\in \mathscr{T}_1(\mathcal{H})$ with $\mathcal{L}_{\widehat{f},H}(x_n)\to y$ in $\mathscr{T}_1(\mathcal{H})$. First, we observe that the expression \eqref{eq.LindbladformL} coincides with the definition \eqref{eq:cLonF}
on $\mathscr{F}$. Moreover, for any $x\in D(\mathcal{W}_H^{\kappa_1,\kappa_1})\cap D(\mathcal{W}_H^{\kappa_1+\mu-\gamma,0})\cap D(\mathcal{W}_H^{0,\kappa_1+\mu-\gamma})$, 
there is a sequence $\{x_n\}_{n}$ in $\mathscr{F}$ such that $x_n\to x$ and $\mathcal{W}_H^{\kappa_1,\kappa_1}(x_n)\to \mathcal{W}_H^{\kappa_1,\kappa_1}(x),\,\mathcal{W}_H^{\kappa_1+\mu-\gamma,0}(x_n)\to \mathcal{W}_H^{\kappa_1+\mu-\gamma,0}(x) $ and $\mathcal{W}_H^{0,\kappa_1+\mu-\gamma}(x_n)\to\mathcal{W}_H^{0,\kappa_1+\mu-\gamma}(x)$ by the core property of $\mathscr{F}$ (cf.~  Lemma~\ref{mathFcoreW} ). Fixing $y\equiv \sum_{\alpha\in\mathcal{A}}L^\alpha\cdot x \cdot (L^\alpha)^\dagger+G\cdot x+x\cdot G^\dagger$ , it suffices to show that for each $\alpha$, $L^\alpha\cdot x_n \cdot (L^\alpha)^\dagger\to L^\alpha\cdot x\cdot (L^\alpha)^\dagger$, $G\cdot x_n\to G\cdot x$ and   $x_n\cdot G^\dagger\to x\cdot G^\dagger$ . This we show using Lemma~\ref{lem.LalphaGwelldefinedF} and explicitly only verify the first convergence, as the others follow similarly:
\begin{align*}
\|L^\alpha\cdot x_n\cdot (L^\alpha)^\dagger- L^\alpha\cdot x\cdot (L^\alpha)^\dagger\|_1&=\|L^\alpha\widetilde{H}^{-\kappa_1}(L^\alpha\widetilde{H}^{-\kappa_1}(\mathcal{W}_H^{\kappa_1,\kappa_1}(x_n-x))^\dagger)^\dagger\|_1\\
&\le \|L^\alpha \widetilde{H}^{-\kappa_1}\|^2\,\|x_n-x\|_{\mathcal{W}_H^{\kappa_1,\kappa_1}}\to 0.
\end{align*}
\end{proof}

 \begingroup
 \subsection{Uniqueness of   the invariant  state}
 \leavevmode
In Theorem~\ref{thm:generalUniqueness}, we give a criterion for uniqueness of the invariant Gibbs state and convergence to it in trace norm.
For that, we derive relations involving the Bohr components of the bare jumps from the vanishing of the Dirichlet form and give a condition under which these imply $\ker L_{\widehat f,H}=\mathbb C\sqrt{\sigma_\beta}.$

  \begin{theorem}[General uniqueness criterion]
\label{thm:generalUniqueness}
Assume the generation hypotheses of Section~\mbox{%
\ref{sec:generalframe}}\hskip0pt%
,
including the Gibbs hypothesis and Condition~\mbox{%
\ref{eq:condAalphas}}\hskip0pt%
,
and fix $\beta>0$. Suppose that $\widehat f(\nu)\ne0$ whenever
$A_\nu^\alpha\ne0$ for some $\alpha$.
For $x\in D(\mathcal E_{\widehat f,H})$, define
\[
B_x(\varphi,\psi)
:=\langle e^{\beta H/4}\varphi,xe^{\beta H/4}\psi\rangle,
\qquad \varphi,\psi\in\mathcal F.
\] 
  If $\mathcal E_{\widehat f,H}(x)=0$, then the Bohr-component relations
\[
B_x((A^\alpha_\nu)^\dagger\varphi,\psi)
=B_x(\varphi,A^\alpha_\nu\psi),
\qquad \alpha\in\mathcal A,\quad\nu\in B(H),\quad\varphi,\psi\in\mathcal F
\]
hold. When $A^\alpha\mathcal F\subset\mathcal F$ for every $\alpha$,
they also imply
\begin{equation}\label{eq:comm_relations}
B_x((A^\alpha)^\dagger\varphi,\psi)
=B_x(\varphi,A^\alpha\psi),
\qquad\alpha\in\mathcal A,\quad\varphi,\psi\in\mathcal F.
\end{equation} 
  Suppose further that the Bohr-component relations force $B_x$ to be a
scalar multiple of $\langle\cdot,\cdot\rangle$ (in particular, it suffices
that the bare jumps preserve $\mathcal F$ and that
\eqref{eq:comm_relations} force this conclusion). Then
  $\ker L_{\widehat f,H}
=\{x\in D(\mathcal E_{\widehat f,H}):\mathcal E_{\widehat f,H}(x)=0\}
=\mathbb C\sqrt{\sigma_\beta}.$
  Moreover, $\sigma_\beta$ is the unique invariant density operator and
\[
\lim_{t\to\infty}\|\Phi_t(\rho)-\sigma_\beta\|_1=0
\] 
 for   every density operator $\rho\in\mathscr T_1(\mathcal H)$.
\end{theorem}

 \begin{proof}
 Let   $\mathcal E_{\widehat f,H}(x)=0$. Since
$\mathscr F_{\sigma_\beta}=\mathscr F+\mathbb C\sqrt{\sigma_\beta}$
is a form core, there are $x_j\in\mathscr F_{\sigma_\beta}$ with
$x_j\to x$ in Hilbert--Schmidt norm and
$\mathcal E_{\widehat f,H}(x_j)\to0$.
For fixed $\alpha$ and $u,v\in\mathcal F$, set
\[
p_j(t):=\sum_\nu\widehat f(\nu)e^{\beta\nu/4}
\langle u,\delta_\nu^\alpha(x_j)v\rangle e^{it\nu}.
\] 
For energy eigenvectors $u,v$, the weak coefficients are
\begin{equation}\label{eq:weak-frequency-coefficient}
c_\nu(x):=\widehat f(\nu)\langle(A_\nu^\alpha)^\dagger u,xv\rangle
-e^{\beta\nu/2}\widehat f(\nu)\langle u,xA_\nu^\alpha v\rangle.
\end{equation}
Orthogonality of the energy projections, Cauchy--Schwarz and KMS symmetry give
\[
\sum_\nu|c_\nu(x)|\le\|\widehat f\|_\infty
\bigl(\|(A^\alpha)^\dagger u\|\,\|xv\|
+\|x^\dagger u\|\,\|A^\alpha v\|\bigr).
\]
The same bound for $x_j-x$ shows convergence of the coefficients in
$\ell^1$, and hence uniform convergence of $p_j$ to the absolutely
convergent series $p_x(t)=\sum_\nu c_\nu(x)e^{it\nu}$.
Finite linear combinations of eigenvectors follow by linearity.
Writing $w(t)=[\beta\cosh(2\pi t/\beta)]^{-1}$,
the form representation \mbox{%
\eqref{eq:Dirichletformdef} }\hskip0pt%
gives
\[
\int_{\mathbb R}w(t)|p_j(t)|^2\,dt
\le\|u\|^2\|v\|^2\mathcal E_{\widehat f,H}(x_j)\longrightarrow0.
\] 
Since $w$ is integrable and strictly positive, uniform convergence gives
$p_x=0$. Averaging $p_x(t)e^{-it\nu_0}$ over $[-T,T]$ and passing
the limit through the absolutely summable series shows $c_{\nu_0}(x)=0$.
Nonvanishing of
$\widehat f$ on the active frequencies consequently gives
\[
e^{-\beta\nu/4}\langle(A_\nu^\alpha)^\dagger u,xv\rangle
=e^{\beta\nu/4}\langle u,xA_\nu^\alpha v\rangle
\] 
  for every $\alpha,\nu,u,v$ as above.
Taking $u=e^{\beta H/4}\varphi$, $v=e^{\beta H/4}\psi$ and using
$A_\nu^\alpha e^{\beta H/4}
=e^{-\beta\nu/4}e^{\beta H/4}A_\nu^\alpha$ on $\mathcal F$
yields the analogous relation for $B_x$ and $A_\nu^\alpha$.
If the bare jumps preserve $\mathcal F$, summing the finitely many nonzero terms on each test vector proves
\mbox{%
\eqref{eq:comm_relations}}\hskip0pt%
.

  The assumed irreducibility of these form relations gives
$B_x(\varphi,\psi)=\lambda\langle\varphi,\psi\rangle$.
Since $e^{\beta H/4}\mathcal F=\mathcal F$, this implies
$x=\lambda e^{-\beta H/2}$ by density and boundedness.
Conversely, $\sqrt{\sigma_\beta}$ has zero form energy.
The kernel assertion follows from the representation theorem for
the closed nonnegative form associated with $-L_{\widehat f,H}$.

  The stationary-projection conclusion above now gives
$\Phi_t(y)\to\Tr(y)\sigma_\beta$ for every $y\in\mathscr T_1$.
Any invariant density operator therefore equals $\sigma_\beta$.
 \end{proof}

\paragraph{Jumps with infinite energy columns.}
If $H$ has simple spectrum, each active nonzero Bohr frequency belongs
to at most one ordered pair of energies, and the graph of nonzero
bare-jump matrix elements is connected, the Bohr components are matrix
units along the edges. Their form relations make $B_x$ diagonal with
equal diagonal entries, so Theorem~\ref{thm:generalUniqueness} applies.
For example, take $H|n\rangle=2^n|n\rangle$ and the bounded self-adjoint
jump $A=|v\rangle\langle v|$ with $v_n>0$ for all $n$ and
$\|v\|=1$. The nonzero differences $2^m-2^n$ determine their ordered
pairs uniquely, and the graph is complete. Thus every bounded,
nonvanishing KMS filter gives uniqueness and trace-norm convergence,
although $A\mathcal F\not\subset\mathcal F$. The criterion and example also hold
for every finite positive Gaussian width: the weight $w*g_{\sigma_E}$
is positive and integrable, so the same proof applies. No gap is asserted.

  \begin{corollary}[One-mode uniqueness without frequency separation]
\label[corollary]{thm:uniqueGibbsSurgical}
Let $\mathcal H=L^2(\mathbb R)$,  $N=a^\dagger a$  , and
$H=h(N)+K$, where $h:\mathbb N_0\to\mathbb R$ and
$K=K^\dagger=1_{[0,n_0]}(N)K1_{[0,n_0]}(N)$ for some $n_0$.
Assume the standing generation and Gibbs hypotheses for $H$ and
the jumps $\{a,a^\dagger\}$, and suppose that $\widehat f$ does not
vanish on their active Bohr frequencies.  Then  , for every inverse
temperature covered by these hypotheses, $\sigma_\beta$ is the unique
invariant density operator and every density operator converges to it
in trace norm. In particular,
$\ker L_{\widehat f,H}=\mathbb C\sqrt{\sigma_\beta}$.
The same conclusions hold for the Gaussian-convoluted semigroups
of Section~\mbox{%
\ref{monotonicitygap}}\hskip0pt%
, for every $0<\sigma_E<\infty$.
 \end{corollary}
 \begin{proof}
  The Gibbs hypothesis implies finite-dimensional spectral subspaces.
Since $K$ acts in a finite Fock block, the finite-energy space of $H$
is exactly the algebraic span of the Hermite basis  $(e_n)_{n\ge0}$  ,
which is preserved by $a,a^\dagger$.
For a sesquilinear form $B$ satisfying \mbox{%
\eqref{eq:comm_relations}}\hskip0pt%
, put
$b_{mn}=B(e_m,e_n)$. The relations give
 \[
\sqrt{m+1}\,b_{m+1,n}=\sqrt n\,b_{m,n-1},\qquad
\sqrt m\,b_{m-1,n}=\sqrt{n+1}\,b_{m,n+1}.
\]
  The boundary cases and iteration show that  $b_{mn}=0$ for   $m\ne n$
and $b_{nn}=b_{00}$. Hence $B$ is scalar, and
Theorem~\mbox{%
\ref{thm:generalUniqueness} }\hskip0pt%
applies.
For Gaussian convolution, the Fourier weight in the form representation
is $w*g_{\sigma_E}$, where $g_{\sigma_E}$ is the Gaussian probability
density in \mbox{%
\eqref{integralrepLse}}\hskip0pt%
. This weight is strictly positive
and integrable, so the finite-exponential-sum proof of that theorem
applies unchanged. The same KMS embedding then gives trace-norm
convergence.
 \end{proof}
  \par\endgroup
 \section{Spectral gap}\label{sec:spectralgaps}

Now that we have rigorously constructed our semigroups of quantum channels that will serve as our Gibbs samplers, we need to develop a framework to study their convergence towards their fixed point $\sigma_\beta$. In this section, we utilize the spectral properties of the operator $L_{\widehat{f},H}$ in order to control the mixing time
\[
t_{\operatorname{mix}}(\epsilon,\mathscr{S}_{D_p(\sigma_\beta)})
\!:=\!\inf\Bigl\{ t \ge 0 \Big| \| \rho_t - \sigma_\beta \|_1 \le \epsilon \,\forall\rho \in \mathscr{S}_{D_p(\sigma_\beta)}\! \Bigr\}
\]
for a given subset $\mathscr{S}_{D_p(\sigma_{\beta})}$ of input quantum states $\rho$ with $\widehat{D}_p(\rho\|\sigma_\beta)\le D_p(\sigma_\beta)<\infty$, $p>1$. Assuming $\ker L_{\widehat f,H}=\mathbb C\sqrt{\sigma_\beta}$ as in Section~\ref{sec.fastconv}, and writing $\lambda_2\equiv \operatorname{gap}(L_{\widehat{f},H})>0$, 
  \begin{align*}
\big\|e^{t\mathcal{L}_{\widehat{f},H}}(\rho)-\sigma_{\beta}\big\|_1\le e^{-\lambda_2 t+{D}_2(\sigma_\beta)/2}
\quad \Longrightarrow \quad t_{\operatorname{mix}}(\epsilon,\mathscr{S}_{D_2(\sigma_\beta)})\le  \frac{1}{\lambda_2}\Big(D_2(\sigma_\beta)/2+\log\Big(\frac{1}{\epsilon}\Big)\Big).
\end{align*} 
Under this simple-kernel assumption, controlling the spectral gap controls the mixing time.

\subsection{The harmonic oscillator}\label{sec:gaussexamples}

In order to build some intuition, we first consider the simple case of a Gaussian thermal state over a single-mode quantum bosonic system. Here, $\mathcal{H} = L^2(\mathbb{R})$ and we denote by $a$ and $a^\dagger$ the annihilation and creation operators, defined on a common dense domain $\mathcal{S}(\mathbb{R})$ of Schwartz functions, where they satisfy the canonical commutation relation
\[
[a,a^\dagger] = I.
\]
The associated number operator is given by
  \begin{align}\label{def:Numberop}
N := a^\dagger a = \sum_{n \in \mathbb{N}_0} n\, |n\rangle\langle n|,
\end{align} 
where   $\{ |n\rangle \}_{n \in \mathbb{N}_0}$  denotes the Fock basis. We recall that these operators act on the Fock basis as
\[
a|n\rangle = \sqrt{n}\,|n-1\rangle,
\qquad
a^\dagger |n\rangle = \sqrt{n+1}\,|n+1\rangle.
\]
We start by choosing our bare jumps as $\{a,a^\dagger\}$ and considering the Hamiltonian $H=\gamma N$  with $\gamma>0$ and $\widehat f(\gamma)\ne0$ . Then, given a function $\widehat{f}$ satisfying the conditions of Section \ref{sec:generalframe}, we conclude that
  \begin{align*}
L^+&:=\sum_{\nu\in B(H)} \widehat{f}(\nu)\, (a^\dagger)_\nu\\
&=\sum_{k\in\mathbb{Z}}\sum_{n\in\mathbb{N}_0}\widehat{f}(\gamma k)|n+k\rangle\langle n+k|a^\dagger |n\rangle\langle n|\\
&=\sum_{n\in\mathbb{N}_0}\, \widehat{f}(\gamma)\,\sqrt{n+1}|n+1\rangle\langle n|=\widehat{f}(\gamma)a^\dagger.
\end{align*} 
Similarly,
  \begin{align*}
L^-:=\sum_{\nu\in B(H)} \widehat{f}(\nu)\, a_\nu
=\sum_{k\in\mathbb{Z}}\sum_{n\in\mathbb{N}_0}\widehat{f}(\gamma k)|n+k\rangle\langle n+k|a |n\rangle\langle n|
=\sum_{n\in\mathbb{N}_0}\, \widehat{f}(-\gamma)\,\sqrt{n}|n-1\rangle\langle n|=\widehat{f}(-\gamma)a.
\end{align*} 
Computing the last two terms of \eqref{eq:cLonF} similarly, we easily derive
\begin{align}\label{eq:qOUgen}
\mathcal{L}_{\widehat{f},\gamma N}(\rho)=&|\widehat{f}(\gamma)|^2\,\Big(a^\dagger \rho a-\frac{1}{2}\{aa^\dagger,\rho\}\Big)+|\widehat{f}(-\gamma)|^2\,\Big(a\rho a^\dagger-\frac{1}{2}\{a^\dagger a,\rho\}\Big)\,.
\end{align}
This coincides with the generator of the quantum Ornstein-Uhlenbeck (qOU) semigroup \cite{OU}, with a birth rate $\nu_+:=|\widehat{f}(\gamma)|^2$ and a death rate $\nu_-:=|\widehat{f}(-\gamma)|^2$. Moreover, since $\overline{\widehat{f}(\gamma)}=\widehat{f}(-\gamma)e^{-\beta \gamma/2}$, it is clear that 
$\nu_+=e^{-\beta\gamma }\nu_-<\nu_-$. 

\smallskip

 Next, we choose $\gamma=1$ and now consider the generator  $L_{\widehat{f},\gamma N}$ on Hilbert-Schmidt operators, as defined  in Equation \eqref{def:L}. A direct computation yields \cite{OU}
\begin{align}
  \label{eq:HSgenerator}
L_{\widehat{f},N}(x) =& -\left( \frac{\nu_- + \nu_+}{2}(Nx + xN) + \nu_+ x \right)  + \sqrt{\nu_+ \nu_-}( a x a^{\dagger} +  a^{\dagger} x a) .
\end{align}
Cipriani, Fagnola, and Lindsay \cite[Thm.~7.2]{OU} showed that the spectrum of $ L_{\widehat{f},N} $ is
\[
\operatorname{Sp}(L_{\widehat{f},N}) = -\left\{ n\left( \frac{\nu_- - \nu_+}{2} \right) \;\middle|\; n \in \mathbb{N}_0 \right\}.
\]
In the limiting case $ \nu_- = \nu_+ $, the spectrum becomes continuous and fills the half-line   $ (-\infty,0] $ , corresponding to the so-called quantum Brownian motion \cite[Thm. 8.1]{OU}.  This equal-rate limit has no normalizable Gibbs state of the oscillator and is considered here as a limiting generator.

\medskip

\paragraph{\textbf{Why work in the Hilbert-Schmidt Space?}}
One might ask why we study $ L_{\widehat{f},N} $ on the Hilbert-Schmidt space instead of working with $ \mathcal{L}_{\widehat{f},N} $ directly on the set $ \mathscr{T}_1(\mathcal{H}) $ of trace-class operators. The reason is that $ \mathcal{L}_{\widehat{f},N} $ lacks a spectral gap in the space of trace-class operators, as the next theorem demonstrates. We recall that, for a linear operator $(A,D(A))$ on a Banach space $\mathscr{B}$, its resolvent $\rho(A)$ is the set of complex numbers $\lambda$ such that $A-\lambda I$ has a bounded inverse $B:\mathscr{B}\to D(A)$, i.e., $(A-\lambda I)B=I$ and $B(A-\lambda I)=I_{D(A)}$. The spectrum $\operatorname{Sp}(A)$ is the complementary set of $\rho(A)$ in $\mathbb{C}$. The spectrum of a closed operator is closed, and that of the generator of a strongly continuous  contraction  semigroup is included in $\mathbb{C}_-$. A necessary condition for the uniform exponential convergence of the strongly continuous contraction semigroup associated with $\mathcal{L}_{\widehat{f},N}$ is the existence of a spectral gap, namely a constant $\lambda_0<0$ such that for all $\lambda\in \operatorname{Sp}(\mathcal{L}_{\widehat{f},H})\backslash \{0\}$, $\operatorname{Re}(\lambda)\le \lambda_0$ \cite[Corollary 4.1.2]{salzmann2024robustness}.

\begin{proposition}\label{thm.qOUnogap}
The spectrum of the OU generator $ \mathcal{L}_{\widehat{f},N} $ on the space of trace-class operators is the entire closed left complex half-plane $ \mathbb{C}_- $.
\end{proposition}

\begin{proof}
We recall that the set $\sigma_{\mathrm{p}}(A)$ of eigenvalues $\lambda$ of a linear map $(A,D(A))$ is included in $\operatorname{Sp}(A)$.   As a comparison with the quantum calculation below, the  formal adjoint of the generator of the classical Ornstein-Uhlenbeck on $L^\infty(\mathbb{R})$ 
\[
G_{\mathrm{OU}} \varphi(x) = q \varphi''(x) - b x \varphi'(x),
\]
 is
\[
G_{\mathrm{OU}}' \varphi(x) = q \varphi''(x) + b x \varphi'(x) + b\varphi(x).
\]
To find its spectrum, we solve the eigenvalue problem
\[
G_{\mathrm{OU}}' \varphi = \lambda \varphi \quad \Rightarrow \quad q \varphi''(x) + b x \varphi'(x) - \mu \varphi(x) = 0,
\]
where $ \mu = \lambda - b $. 
Taking the Fourier transform of this equation yields
\[
- q \xi^2 \widehat{\varphi}(\xi) - b\left( \widehat{\varphi}(\xi) + \xi \widehat{\varphi}'(\xi) \right) - \mu \widehat{\varphi}(\xi) = 0.
\]
Solving this differential equation in Fourier space gives
\[
\widehat{\varphi}(\xi) = c e^{ -\frac{q \xi^2}{2b} } |\xi|^{ -\frac{\lambda}{b} },
\]
for some $ c \neq 0 $. Taking the inverse Fourier transform yields
\[
\varphi(x) = \tfrac{c}{2} \left( \tfrac{2b}{q} \right)^{\frac{b-\lambda}{2b}} \Gamma\left( \tfrac{b-\lambda}{2b}\right) \, {}_1F_1\left( \tfrac{b-\lambda}{2b}, \tfrac{1}{2}, -\tfrac{b x^2}{2q} \right),
\]
where $F_1$ denotes the confluent hypergeometric function. Asymptotically, we have
  \[ \vert \varphi(x)\vert = \mathcal O(|x|^{-1+\frac{\Re\lambda}{b}}).\] 
 Thus, $\varphi \in L^1(\mathbb R)$ for $\Re(\lambda)<0$  , and  the Fourier semigroup gives  $\mathbb C_- \subset  \operatorname{Sp}(G'_{\operatorname{OU}})=  \operatorname{Sp}(G_{\operatorname{OU}})$. Since $\mathcal  L_{\widehat{f},N}$ generates a contraction semigroup, we have  $\operatorname{Sp}(\mathcal L_{\widehat{f},N}) \subset \mathbb C_-.$ 
Finally, it was shown in \cite{OU} that, given the multiplication operator $ M_\varphi $ defined by   $ (M_\varphi u)(x) = \varphi(x)u(x) $,
 the weak$^*$ Heisenberg generator  $\mathcal{L}_{\widehat{f},N}^\dagger$ satisfies, on smooth test  functions,
\[
G_{\mathrm{OU}} \varphi(x) := \bigl[\mathcal{L}_{\widehat{f},N}^{\dagger}(M_\varphi)\bigr](x) = q \varphi''(x) - b x \varphi'(x),
\] 
with
\[
q := \frac{\nu_- + \nu_+}{4} > 0, \quad b := \frac{\nu_- - \nu_+}{2} > 0.
\]

We now prove the spectral inclusion for the quantum
generator directly. Identify a sequence
$p=(p_n)_{n\ge0}\in\ell^1(\mathbb N_0)$ with the diagonal operator
\[
X_p:=\sum_{n\ge0}p_n|n\rangle\langle n|,
\qquad \|X_p\|_1=\|p\|_{\ell^1}.
\]
For finitely supported $p$, equation~\eqref{eq:qOUgen}
gives $\mathcal L_{\widehat f,N}(X_p)=X_{Qp}$, where
\[
(Qp)_n=\nu_+n p_{n-1}+\nu_-(n+1)p_{n+1}
-\bigl(\nu_+(n+1)+\nu_-n\bigr)p_n,
\qquad p_{-1}:=0.
\]
Fix $\lambda\in\mathbb C$ with $\Re\lambda<0$, and set
$d:=\nu_--\nu_+>0$, $r:=\nu_+/\nu_-\in(0,1)$, and
$s:=-\lambda/d$, so that $\Re s>0$. Define $p_n$ by
\[
P_\lambda(z):=\sum_{n\ge0}p_nz^n
=(1-z)^s(1-rz)^{-s-1},\qquad |z|<1,
\]
using the analytic branches whose values at $z=0$ are
$1$. Differentiation gives
\[
(1-z)(\nu_--\nu_+z)P_\lambda'(z)
-\nu_+(1-z)P_\lambda(z)=\lambda P_\lambda(z).
\]
Comparing coefficients shows that $Qp=\lambda p$
componentwise. The binomial coefficients of $(1-z)^s$ are
$O((n+1)^{-1-\Re s})$; when $s$ is a positive integer, they
vanish for large $n$. The coefficients of $(1-rz)^{-s-1}$
decay exponentially since $r<1$. Their convolution therefore
satisfies
\[
|p_n|\le C_\lambda(n+1)^{-1-\Re s}.
\]
In particular, $p\in\ell^1$, $p_0=1$, and
$n|p_n|\to0$. To obtain vectors in the generator domain without
assuming a finite first moment of $p$, truncate the sequence:
$p^{(M)}_n:=p_n\mathbf1_{\{n\le M\}}$.
The recurrence cancels every component of
$(Q-\lambda)p^{(M)}$ except those at $M$ and $M+1$, giving
\[
\|(Q-\lambda)p^{(M)}\|_{\ell^1}
=(M+1)\bigl(\nu_-|p_{M+1}|+\nu_+|p_M|\bigr)
\longrightarrow0.
\]
Thus the finite-rank operators
$X_M:=X_{p^{(M)}}/\|p^{(M)}\|_{\ell^1}$ belong to
$D(\mathcal L_{\widehat f,N})$, have trace norm $1$, and satisfy
\[
\|\mathcal L_{\widehat f,N}(X_M)-\lambda X_M\|_1
\longrightarrow0,
\]
since $\|p^{(M)}\|_{\ell^1}\to\|p\|_{\ell^1}>0$.
A bounded inverse of $\mathcal L_{\widehat f,N}-\lambda I$
would force $\|X_M\|_1\to0$, a contradiction. Hence every
$\Re\lambda<0$ lies in the spectrum. Closedness of the spectrum
then includes the imaginary axis, while contractivity of the
semigroup gives the reverse inclusion. Consequently,
\[
\operatorname{Sp}(\mathcal L_{\widehat f,N})
=\mathbb C_-.
\]
 \end{proof}

 \subsection{A sharp criterion for one-mode downward rates}
\label{sec:GeneralNumberviaBirthDeath}
Consider a one-mode Hamiltonian
\begin{equation}\label{eq:NumPresHamiltonian}
H=h(N)=\sum_{n\ge0}h(n)|n\rangle\langle n|
\end{equation} 
with bare jumps $\{a,a^\dagger\}$, under the standing generation and
Gibbs hypotheses. Define the Gibbs weights and transition rates by
\begin{equation}\label{eq:BirthDeathRates}
\begin{aligned}
\pi_n&=Z_\beta^{-1}e^{-\beta h(n)},\qquad
b_n=(n+1)|\widehat f(h(n+1)-h(n))|^2,\\
d_n&=n|\widehat f(h(n-1)-h(n))|^2\quad(n\ge1),\qquad d_0=0.
\end{aligned}
\end{equation} 
KMS symmetry gives $\pi_nb_n=\pi_{n+1}d_{n+1}$.

The filtered jumps are
\[
L^+=a^\dagger\widehat f(h(N+1)-h(N)),\qquad
L^-=a\widehat f(h(N-1)-h(N)).
\] 
  Their squared moduli are diagonal in the Fock basis, so the coherent
correction vanishes. On finite Fock matrices the generator is
\begin{equation}\label{eq:LindbladNumPresHamiltonian}
\mathcal L_{\widehat f,h(N)}(\rho)
=L^+\rho(L^+)^\dagger+L^-\rho(L^-)^\dagger
-\tfrac12\{(L^+)^\dagger L^++(L^-)^\dagger L^-,\rho\}.
\end{equation} 
This is a quantum birth--death generator, as studied in
\cite{carboneFagnola_exponentiall2-convergence_2000}.

\begin{theorem}[Sharp downward-rate criterion]
\label{thm:DownwardRateCriterion}
Assume the standing generation and Gibbs hypotheses for $H=h(N)$,
the jumps $\{a,a^\dagger\}$ and a KMS-symmetric filter satisfying
$\widehat f(\nu)>0$ whenever $A_\nu^\alpha\ne0$ for either bare jump.
Fix $\beta>0$.
Suppose that $h(n)$ is non-decreasing for $n\ge n_0$ and that
there are $s\in\mathbb N$ and $\delta>0$ such that
\begin{equation}\label{eq:LowerboundEn+s-En}
h(m+s)-h(m)\ge\delta,\qquad m\ge n_0.
\end{equation} 
  Then
\begin{equation}\label{eq:DownwardRateCriterion}
\operatorname{gap}(L_{\widehat f,h(N)})>0
\quad\Longleftrightarrow\quad
d_*:=\inf_{n\ge1}n|\widehat f(h(n-1)-h(n))|^2>0.
\end{equation} 
More precisely, if $d_*>0$, then
\[
\operatorname{gap}(L_{\widehat f,h(N)})
\ge \frac{d_*}{4}\left(1-e^{-\beta\delta/(2s)}\right)^2
e^{-\beta[\delta(1+n_0/s)+\Delta_E]},
\qquad
\Delta_E:=\max_{0\le i,j\le2n_0}|h(i)-h(j)|.
\]

Moreover, $\sigma_\beta$ is the unique invariant density operator of the
quantum Markov semigroup generated by $\mathcal L_{\widehat f,h(N)}$, i.e.~$\ker\mathcal L_{\widehat f,h(N)}=\mathbb C\sigma_\beta$,
or equivalently $\ker L_{\widehat f,h(N)}=\mathbb C\sqrt{\sigma_\beta}.$\end{theorem}

\begin{proof}
Corollary~\ref{thm:uniqueGibbsSurgical} gives uniqueness of the invariant
density operator and $\ker L_{\widehat f,h(N)}=\mathbb C\sqrt{\sigma_\beta}$.
The convergence to the stationary projection established in
Section~\ref{sec:generalframe} then gives
$\ker\mathcal L_{\widehat f,h(N)}=\mathbb C\sigma_\beta$.
We therefore use the variational characterization of the spectral gap,
writing $\lambda:=\operatorname{gap}(L_{\widehat f,h(N)})$.
We first prove necessity. Suppose that $d_*=0$.
Let $T_n=\sum_{j\ge n}\pi_j$ and use the core vector
\[
x_n=\sum_{j\ge n}\sqrt{\pi_j}|j\rangle\langle j|
=\sqrt{\sigma_\beta}-\sum_{j<n}\sqrt{\pi_j}|j\rangle\langle j|.
\] 
  The diagonal restriction is the reversible birth--death form, so
$\mathcal E_{\widehat f,H}(x_n)=\pi_nd_n$ and
$\|x_n\|_2^2-|\langle\sqrt{\sigma_\beta},x_n\rangle|^2
=T_n(1-T_n)$. Thus the spectral gap satisfies
\begin{equation}\label{eq:TailGapBound}
0\le\lambda
\le\frac{\pi_nd_n}{T_n(1-T_n)}
\le\frac{d_n}{1-T_n}.
\end{equation} 
Strict positivity of the filter gives $d_n>0$ for every $n\ge1$,
so $d_*=0$ supplies a subsequence $n\to\infty$ with $d_n\to0$.
Since $T_n\to0$, the bound above gives $\lambda=0$.

For sufficiency, assume $d_*>0$.
For the proof, abbreviate
\[
r:=e^{-\beta\delta/s},\qquad
C:=e^{\beta[\delta(1+n_0/s)+\Delta_E]},\qquad
K:=\frac{C}{(1-\sqrt r)^2}.
\]
Block growth and eventual monotonicity imply
$h(n)-h(j)\ge(\delta/s)(n-j)-\delta(1+n_0/s)-\Delta_E$
for all $n\ge j\ge0$, by splitting at $n_0$ when necessary.
Consequently
\begin{equation}\label{eq:PolynomialGibbsRatios}
\frac{\pi_n}{\pi_j}\le Cr^{n-j},\qquad n\ge j\ge0.
\end{equation} 
  The same bound holds for
$\mu_n^{(k)}:=\sqrt{\pi_n\pi_{n+k}}$, uniformly in $k\ge0$.
For any of these weights $\mu$, telescoping and the $\ell^1$--$\ell^2$
convolution inequality give
\begin{equation}\label{eq:PolynomialHardy}
\sum_{n\ge0}\mu_n|u_n-u_0|^2
\le K\sum_{n\ge1}\mu_n|u_n-u_{n-1}|^2,
\qquad \sum_{n\ge0}\mu_n\le K\mu_0.
\end{equation} 
  Indeed,
$\sqrt{\mu_n}(u_n-u_0)=
\sum_{j=1}^n\sqrt{\mu_n/\mu_j}\sqrt{\mu_j}(u_j-u_{j-1})$,
and the convolution kernel is bounded by $\sqrt C\,r^{(n-j)/2}$.

Decompose finite Fock matrices into diagonals. On the $k$th upper
diagonal write
\[
x=\sum_{n\ge0}\sqrt{\mu_n^{(k)}}\,u_n|n\rangle\langle n+k|,
\qquad \|x\|_2^2=\sum_{n\ge0}\mu_n^{(k)}|u_n|^2.
\] 
Transforming \eqref{eq:LindbladNumPresHamiltonian} by \eqref{eqcLtoL}
and using detailed balance yields
\begin{equation}\label{eq:PolynomialSectorForm}
\begin{aligned}
\mathcal E_{\widehat f,H}(x)
&=\sum_{n\ge1}\mu_n^{(k)}\sqrt{d_nd_{n+k}}|u_n-u_{n-1}|^2\\
&\quad+\frac{1}{2}\sum_{n\ge0}\mu_n^{(k)}
\bigl[(\sqrt{b_n}-\sqrt{b_{n+k}})^2
+(\sqrt{d_n}-\sqrt{d_{n+k}})^2\bigr]|u_n|^2.
\end{aligned}
\end{equation} 
  To check the identity, expansion gives diagonal coefficient
$(b_n+b_{n+k}+d_n+d_{n+k})/2$; the neighbouring coefficients follow
from $\pi_nb_n=\pi_{n+1}d_{n+1}$.
For $k=0$, \eqref{eq:PolynomialHardy} 
gives
\[
\|x\|_2^2-|\langle\sqrt{\sigma_\beta},x\rangle|^2
\le\sum_n\pi_n|u_n-u_0|^2
\le\frac{K}{d_*}\mathcal E_{\widehat f,H}(x).
\] 
  For $k\ge1$, the $n=0$ term of the second sum in
\eqref{eq:PolynomialSectorForm} 
is at least
$\mu_0^{(k)}d_*|u_0|^2/2$, because $d_0=0$. Hence
\[
\begin{aligned}
\|x\|_2^2
&\le2K\left(\mu_0^{(k)}|u_0|^2
+\sum_{n\ge1}\mu_n^{(k)}|u_n-u_{n-1}|^2\right)\\
&\le\frac{4K}{d_*}\mathcal E_{\widehat f,H}(x).
\end{aligned}
\]
The lower diagonals follow by adjunction. Distinct diagonals are
orthogonal for the norm and form, so summation and passage through
the form core prove the asserted lower bound.
\end{proof}

In particular, the sufficient condition
$\sqrt n\,\widehat f(\pm[h(n+1)-h(n)])\to0$ implies the gap obstruction.
Only the downward rates are needed, and no Schwartz assumption on
the filter is required. For $h(n)=n^q$, $q>1$, filters with Gaussian
frequency decay satisfy the condition.

\begingroup

  For comparison with the Metropolis--Hastings choice, define
\begin{equation}\label{eq:Function}
\hatfM(\nu):=\exp\left[-\frac{\sqrt{1+\beta^2\nu^2}+\beta\nu}{4}\right].
\end{equation} 
  \begin{figure}[h!]
\color{black}
\centering
\includegraphics[width=0.3\linewidth]{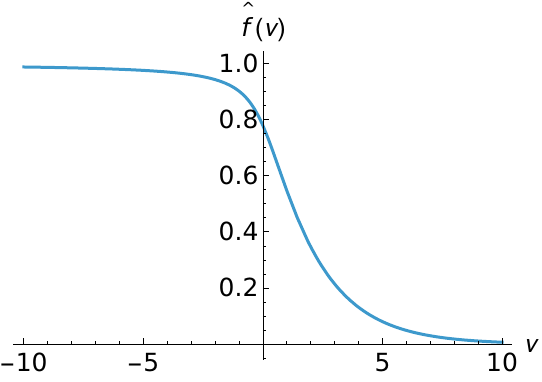}
\caption{Metropolis-type filter $\hatfM$ for $\beta=1$.}
\label{fig:placeholder}
\end{figure}

  \begin{corollary}[Metropolis filter]
\label[corollary]{thm:SpectralGapGeneralh(N)}
Assume the standing generation hypotheses for the Metropolis filter
and the Hamiltonian assumptions of Theorem~\mbox{%
\ref{thm:DownwardRateCriterion}}\hskip0pt%
.
Then
\begin{equation}\label{eq:UniformGapBound}
\operatorname{gap}(L_{\hatfM,h(N)})
\ge\frac{\left(1-e^{-\beta\delta/(2s)}\right)^2}{4}
\exp\left[-\beta\delta\left(1+\frac{n_0}{s}\right)
-\frac{\sqrt{1+\beta^2\Delta_E^2}+3\beta\Delta_E}{2}\right]>0.
\end{equation}

\end{corollary}
\begin{proof}
Writing $\Delta_n=h(n+1)-h(n)$, eventual monotonicity gives
$\Delta_n\ge-\Delta_E$ for all $n$. Since $\hatfM$ is decreasing,
\[
|\hatfM(-\Delta_n)|^2\ge|\hatfM(\Delta_E)|^2
=\exp\left[-\frac{\sqrt{1+\beta^2\Delta_E^2}+\beta\Delta_E}{2}\right]
=:c_0.
\]
Thus $d_n\ge c_0n$ and Theorem~\mbox{%
\ref{thm:DownwardRateCriterion} }\hskip0pt%
applies.
 \end{proof}

  \subsection{Polynomially decaying filters}
 \label{sec:PolynomialFilters}
  To obtain integrable time-domain filters, we damp the Metropolis
filter at negative frequencies. For $\eta\ge0$, set
\begin{equation}\label{eq:PolynomialFilter}
\widehat f_\eta(\nu):=
\frac{\hatfM(\nu)}{\langle\beta\nu\rangle^\eta}
=\frac{\exp[-(\sqrt{1+\beta^2\nu^2}+\beta\nu)/4]}{(1+\beta^2\nu^2)^{\eta/2}}.
\end{equation} 
  The denominator is even, and thus KMS symmetry is preserved.

  \begin{corollary}[Sharp polynomial decay criterion]
\label[corollary]{thm:SharpPolynomialFilter}
Assume the standing generation hypotheses for $\widehat f_\eta$
and the Hamiltonian assumptions of Theorem~\mbox{%
\ref{thm:DownwardRateCriterion}}\hskip0pt%
.
Write $\Delta_n=h(n+1)-h(n)$ and define
\begin{equation}\label{eq:PolynomialSpacingCriterion}
M_\eta(h):=\sup_{n\ge0}\frac{\langle\beta\Delta_n\rangle^{2\eta}}{n+1}.
\end{equation} 
  Then, for $\eta\ge0$,
\begin{equation}\label{eq:PolynomialGapCriterion}
\operatorname{gap}(L_{\widehat f_\eta,h(N)})>0
\quad\Longleftrightarrow\quad M_\eta(h)<\infty.
\end{equation} 
  When $M_\eta(h)<\infty$, the gap is at least
$c_{\beta,h}/M_\eta(h)$ with $c_{\beta,h}=c_0/(4K)$ independent of $\eta$.
\end{corollary}
 \begin{proof}
  The bounds $c_0\le|\hatfM(-\Delta_n)|^2\le1$ imply
\begin{equation}\label{eq:PolynomialDeathBound}
\frac{c_0}{M_\eta(h)}\le\inf_{n\ge1}d_n\le\frac{1}{M_\eta(h)},
\end{equation} 
  with $1/\infty=0$. Apply Theorem~\mbox{%
\ref{thm:DownwardRateCriterion}}\hskip0pt%
.
 \end{proof}
  For example,
\begin{equation}\label{eq:PolynomialPowerThreshold}
h(n)=n^q,\quad q>1:\qquad
\operatorname{gap}(L_{\widehat f_\eta,H})>0
\quad\Longleftrightarrow\quad0\le\eta\le\frac{1}{2(q-1)}.
\end{equation} 
  For $h(n)=n$, every finite exponent is admissible. The next lemma
provides the time-domain bounds needed for implementation.
\par\endgroup

 \begin{lemma}[Time-domain integrability]
\label{lem:PolynomialFilterTime}
For every $\eta>0$, $\widehat f_\eta$ is the Fourier transform
of a function $f_\eta\in L^1(\mathbb R)$. Moreover, for every
$0<a<1/\beta$, there is $C_{a,\eta,\beta}<\infty$ such that
\begin{equation}\label{eq:PolynomialFilterTimeTail}
|f_\eta(t)|\le C_{a,\eta,\beta}
\frac{e^{-a|t|}}{|t|},\qquad |t|\ge1.
\end{equation}
Near the origin,
\begin{equation}\label{eq:PolynomialFilterTimeOrigin}
|f_\eta(t)|\le C_{\eta,\beta}
\begin{cases}
1+|t|^{\eta-1},&0<\eta<1,\\
1+|\log|t||,&\eta=1,\\
1,&\eta>1,
\end{cases}
\qquad 0<|t|\le1.
\end{equation}
\end{lemma}

\begin{proof}
The principal branches in \eqref{eq:PolynomialFilter} are
analytic on $|\operatorname{Im}z|<1/\beta$, since
$\operatorname{Re}(1+\beta^2z^2)>0$ there. Uniformly on each
closed smaller strip,
\[
\widehat f_\eta(x+iy)=
\begin{cases}
O(|x|^{-\eta}),&x\to-\infty,\\
O(e^{-\beta x/2}x^{-\eta}),&x\to+\infty,
\end{cases}
\]
and its derivative is integrable along every horizontal line
in that strip. For $t\ne0$, the improper integral
\[
f_\eta(t)=\frac{1}{2\pi}\int_{\mathbb R}
\widehat f_\eta(\nu)e^{-it\nu}\,d\nu
\]
therefore converges by integration by parts. Shifting the
contour
to $\operatorname{Im}z=-a\operatorname{sgn}(t)$ and integrating
by parts proves \eqref{eq:PolynomialFilterTimeTail}; the vertical
sides vanish by the displayed decay.

At the negative end,
$\widehat f_\eta(-u)=\beta^{-\eta}u^{-\eta}
+O(u^{-\eta-1})$ as $u\to\infty$. Subtracting this leading
power on $u\ge1$ leaves an $L^1$ frequency function, whose
inverse transform is bounded. Split the remaining oscillatory
integral at $u=|t|^{-1}$ and integrate its tail by parts.
The first piece is bounded by
$\int_1^{|t|^{-1}}u^{-\eta}\,du$, and the second by a constant
times $|t|^{\eta-1}$. This proves
\eqref{eq:PolynomialFilterTimeOrigin}. Both time-domain bounds
are integrable, so $f_\eta\in L^1(\mathbb R)$.

To justify the Fourier identification, for $R\ge1$ set
\[
f_{\eta,R}(t):=\frac{1}{2\pi}\int_{-R}^R
\widehat f_\eta(\nu)e^{-it\nu}\,d\nu.
\]
The same split at $|t|^{-1}$ gives
\eqref{eq:PolynomialFilterTimeOrigin} for $f_{\eta,R}$ with a
constant independent of $R$. Integration by parts also gives
\[
|f_{\eta,R}(t)|\le
\frac{2\|\widehat f_\eta\|_\infty+
\|\widehat f_\eta'\|_1}{2\pi|t|},\qquad |t|\ge1.
\]
Since $f_{\eta,R}(t)\to f_\eta(t)$ for $t\ne0$, these
bounds allow dominated convergence against every Schwartz test
function. Thus $f_{\eta,R}\to f_\eta$ in
$\mathcal S'(\mathbb R)$, the space of tempered distributions.
On the frequency side,
$\mathbf1_{[-R,R]}\widehat f_\eta\to\widehat f_\eta$ in
$\mathcal S'(\mathbb R)$ by dominated convergence as well.
Continuity and invertibility of the Fourier transform on
$\mathcal S'(\mathbb R)$, as stated in
\cite[Proposition~9.6, p.~60]{Melrose2004FourierInversion},
therefore identify the distributional Fourier transform of
$f_\eta$ with $\widehat f_\eta$. For an $L^1$ function, the
distributional transform agrees with its ordinary Fourier
transform by Fubini's theorem. Both that transform and
$\widehat f_\eta$ are continuous, so their distributional equality
implies the pointwise identity
\[
\widehat f_\eta(\nu)=\int_{\mathbb R}
f_\eta(t)e^{it\nu}\,dt,\qquad\nu\in\mathbb R.
\]
\end{proof}

\section{Efficient implementation}\label{sec:implementations}

\noindent In Proposition \ref{propDirichlettoSchro}, we constructed the generator $\cL_{\smash{\widehat{f},H}}$ of a quantum dynamical semigroup associated with a function $\widehat{f}$ that satisfies minimal summability and symmetry assumptions in \Cref{eq:condAalphas}. In Section \ref{sec:spectralgaps}, we argued that the spectral gap of the corresponding evolution—which governs its mixing time—crucially depends on the choice of $\widehat{f}$. In particular, choosing $\widehat{f}\equiv 
    \hatfM(\nu)$ as defined in \eqref{eq:Function}, we found in       Corollary~\ref{thm:SpectralGapGeneralh(N)}  that the Gibbs sampler     corresponding to $H = N^2$ is gapped; however, this is no longer     true for rapidly decaying functions, cf.~\Cref{thm:DownwardRateCriterion}.     However, the particular choice of \eqref{eq:Function} appears to     hinder algorithmic implementations of the dynamics, as the lack     of decay of $\hatfM$ as $\nu\to-\infty$ implies singularity of     its Fourier transform $\fM$ at $0$. Another serious issue     compared to the finite-dimensional setting concerns the possible     unboundedness of the bare Hamiltonian $H$. 

The polynomially decaying filters $\widehat f_\eta$ of
\eqref{eq:PolynomialFilter} resolve the first difficulty:
  Corollary~\ref{thm:SharpPolynomialFilter}  gives their sharp gap
criterion, while Lemma~\ref{lem:PolynomialFilterTime} shows that,
for every $\eta>0$, their time-domain kernels belong to
$L^1(\mathbb R)$ and decay exponentially at infinity.   

 In order to address the 
unboundedness of the Hamiltonian, we consider an alternative family of generators  for which the  efficiency results established in Section \ref{sec:gaussexamples}  remain valid, and we show that these generators can be implemented  using infinite-dimensional extensions of the LCU technique. Our  starting point is the following generator, which is well defined on  $\mathscr{F}$ for any   $\sigma_E\in[0,\infty]$  , under \Cref{eq:condAalphas} ,  where the Gaussian factor is interpreted as  $\mathbf1_{\{\nu_1=\nu_2\}}$ at $\sigma_E=0$ and as $1$ at  $\sigma_E=\infty$ :
\begin{align}
\label{eq:DefLsigma_E}
\mathcal{L}_{\sigma_E,\widehat{f},H}(\rho)\!:=\!\sum_\alpha&\sum_{\nu_1,\nu_2}e^{-\frac{(\nu_1-\nu_2)^2}{8\sigma_E^2}}\,\overline{\widehat{f}(\nu_1)}\widehat{f}(\nu_2)A^\alpha_{\nu_2}\rho (A^\alpha_{\nu_1})^\dagger+G_{\sigma_E}\rho+\rho G_{\sigma_E}^\dagger\,,
\end{align}
where 
\begin{align*}
 G_{\sigma_E}
&=-\sum_{\substack{\alpha\in\mathcal{A}\\ \nu_1,\nu_2}}\!e^{-\frac{(\nu_1-\nu_2)^2}{8\sigma_E^2}}\!\!\! \frac{1}{1+e^{\frac{\beta(\nu_2-\nu_1)}{2}}}\overline{\widehat{f}(\nu_1)}\widehat{f}(\nu_2)(A^\alpha_{\nu_1})^\dagger A^\alpha_{\nu_2}.
\end{align*}
Formally, we retrieve the generator $\mathcal{L}_{\widehat{f},H}$ of Proposition \ref{propDirichlettoSchro} by setting $\sigma_E=\infty$. It is not hard to see, by adapting the generation results of Section \ref{sec:generalframe}, that the above expression gives rise to a KMS-symmetric quantum Markov semigroup for any $\sigma_E\ge0$, whose generator in the Hilbert-Schmidt setting we denote by $L_{\sigma_E,\widehat{f},H}$. On the other hand, taking the limit $\sigma_E\to 0$, since $B^\alpha_{\nu\nu}=0$,  we obtain a Davies-type generator $\mathcal{L}_{0,\widehat{f},H}\equiv \mathcal{L}_{\operatorname{D}}$, with
\begin{align}
\label{eq:DefDavies}
\mathcal{L}_{\operatorname{D}}(\rho)=\sum_{\substack{\alpha \\ \nu\in B(H)}}|\widehat{f}(\nu)|^2\,\big(A^\alpha_\nu\rho (A^\alpha_\nu)^\dagger-\frac{1}{2}\{(A^\alpha_\nu)^\dagger A^\alpha_\nu,\rho\}\big),
\end{align}
that is, the rates $\Upsilon$ introduced in \eqref{eq:Davies} all coincide with $|\widehat{f}|^2$. In other words, the family $\{\mathcal{L}_{\sigma_E,\widehat{f},H}\}_{\sigma_E\ge 0}$ interpolates between $\mathcal{L}_{\widehat{f},H}$ and $\mathcal{L}_{\operatorname{D}}$. The next sections are organized as follows: first, in Section \ref{monotonicitygap}, we briefly justify the well-posedness of the generators $\cL_{\sigma_E,\widehat{f},H}$ and $L_{\sigma_E,\widehat{f},H}$, and prove that introducing the Gaussian envelope cannot decrease a positive spectral gap, provided $\ker L_{\widehat f,H}=\mathbb C\sqrt{\sigma_\beta}$. Under this assumption, the positive gap lower bounds of Section \ref{sec:spectralgaps} derived for $L_{\widehat f,H}$ directly hold for $L_{\sigma_E,\widehat f,H}$ for every $0<\sigma_E\le\infty$. In Section \ref{sec:IntegralRepresentation}, we derive an integral formulation of the generator $\mathcal{L}_{\smash{\sigma_E,\widehat{f},H}}$ 
for $f\in L^1(\mathbb R)$.  The rest of this Section is devoted to the derivation of an implementation scheme for the evolution generated by $\cL_{\smash{\sigma_E,\widehat{f},H}}$.

\subsection{Gaussian-convoluted generators}\label{monotonicitygap}

We start by considering the following densely defined, symmetric operator defined on $\mathscr{F}_{\sigma_\beta}$:

\begin{align}
L_{\sigma_E,\widehat{f},H}(\lambda x+\mu\,\sqrt{\sigma_\beta}):=-\lambda \sum_{\substack{\alpha\in\mathcal{A}\\ \nu_1,\nu_2\in B(H)}}&e^{-\frac{(\nu_1-\nu_2)^2}{8\sigma_E^2}}\,\frac{\overline{\widehat{f}(\nu_1)}\widehat{f}(\nu_2)\,e^{\beta(\nu_1+\nu_2)/4}}{2\cosh((\nu_1-\nu_2)\beta/4)}\, (\delta^{\alpha}_{\nu_1})^\dagger\delta^\alpha_{\nu_2}(x)
\end{align}
with associated form
\begin{align}\label{calEGibbs1}
\mathcal{E}_{\sigma_E,\widehat{f},H}(\lambda\,x+\mu\sqrt{\sigma_\beta}):=|\lambda|^2\sum_{\substack{\alpha\in\mathcal{A}\\ \nu_1,\nu_2\in B(H)}}&e^{-\frac{(\nu_1-\nu_2)^2}{8\sigma_E^2}}\,\frac{\overline{\widehat{f}(\nu_1)}\widehat{f}(\nu_2)\,e^{\beta(\nu_1+\nu_2)/4}}{2\cosh((\nu_1-\nu_2)\beta/4)}\, \langle \delta^\alpha_{\nu_1}(x),\delta^\alpha_{\nu_2}(x)\rangle.
\end{align}
By Fourier integration, we notice that for all $x\in\mathscr{F}$
  \begin{align}
\mathcal{E}_{\sigma_E,\widehat{f},H}(x) &=\sigma_E\sqrt{\frac{2}{\pi}}\sum_{\substack{\alpha\\\nu_1,\nu_2}}\int_{-\infty}^\infty \!\!\!e^{-2\sigma_E^2s^2}  e^{-is (\nu_1-\nu_2)} \, \frac{\overline{\widehat{f}(\nu_1)}\widehat{f}(\nu_2)\,e^{\beta(\nu_1+\nu_2)/4}}{2\cosh((\nu_1-\nu_2)\beta/4)}\,\langle \delta^\alpha_{\nu_1}(x),\delta^\alpha_{\nu_2}(x)\rangle\,ds\nonumber \\
&=\sigma_E\sqrt{\frac{2}{\pi}}\sum_{\substack{\alpha\\\nu_1,\nu_2}}\int_{-\infty}^\infty \!\!\!e^{-2\sigma_E^2s^2}  \frac{\overline{\widehat{f}(\nu_1)}\widehat{f}(\nu_2)\,e^{\beta(\nu_1+\nu_2)/4}}{2\cosh((\nu_1-\nu_2)\beta/4)}\,\langle \delta^\alpha_{\nu_1}(x(s)),\delta^\alpha_{\nu_2}(x(s))\rangle\,ds\nonumber \\
&=\sigma_E{\sqrt{\frac{2}{\pi}}}\,\int_{-\infty}^\infty e^{-2\sigma_E^2s^2}\,\mathcal{E}_{\widehat{f},H}(x(-s))\,ds\label{integralrepLse}
\end{align} 
for $x(t):=e^{iHt}x e^{-iHt}$, where we also used that $e^{it H}A^\alpha_\nu e^{-itH}=e^{it\nu}A^\alpha_\nu$.
The identity also holds for
$x\in\mathscr F_{\sigma_\beta}$: adding a multiple of
$\sqrt{\sigma_\beta}$ leaves both forms unchanged, and this Gibbs
square root is fixed by $x\mapsto x(t)$. 

This directly shows
that $L_{\smash{\sigma_E,\widehat{f},H}}$ is negative. By Friedrichs extension, $\mathcal{E}_{\smash{\sigma_E,\widehat{f},H}}$ is closable and defines a self-adjoint extension of $L_{\smash{\sigma_E,\widehat{f},H}}$. Both the close form and operator are denoted as $\mathcal{E}_{\smash{\sigma_E,\widehat{f},H}}$ and $L_{\smash{\sigma_E,\widehat{f},H}}$, by slight abuse of notations. 
\begin{proposition}\label{prop.QMSforLinsigmaEappendix}
The closed quadratic form $\mathcal{E}_{\smash{\sigma_E,\widehat{f},H}}$ is completely Dirichlet. Therefore, the strongly continuous semigroup $\{e^{tL_{\sigma_E,\widehat{f},H}}\}_{t\ge 0}$ is completely Markov with respect to the state $\sigma_\beta$.
\end{proposition}

\begin{proof}

By Proposition \ref{GibbsDirichlet}, we know that $\mathcal{E}_{\widehat{f},H}$ is completely Dirichlet. Therefore, for any $\mathbf{x}\in\mathscr{M}_n\otimes  \mathscr{F}_{\sigma_\beta}$, we have that $\mathcal{E}^{(n)}_{\widehat{f},H}(\mathbf{x})\ge \mathcal{E}_{\widehat{f},H}^{(n)}(\mathbf{x}_+),\mathcal{E}_{\widehat{f},H}^{(n)}(\mathbf{x}_\wedge)$. The result follows since $\mathbf{x}_+(s)=(\mathbf{x}(s))_+$, $\mathbf{x}_\wedge(s)=(\mathbf{x}(s))_\wedge$, the integral representation \eqref{integralrepLse} and the form core property of $\mathscr{F}_{\sigma_\beta}$.

\end{proof}

\noindent Next, we aim at deriving an expression for the generator $\cL_{\sigma_E,\widehat{f},H}$ of the strongly continuous semigroup induced on $\mathscr{T}_1(\cH)$ by extension of the reasoning of   Lemma~\ref{T1semigroupgeneration} . In place of the operator $G$ defined in \eqref{eq:GinGenerationTheore}, we consider the densely defined    operator $G_{\sigma_E}$ on $\mathcal{F}$ as
\begin{align}\label{GsigmaEdeffirst}
 G_{\sigma_E}
&=-\sum_{\substack{\alpha\in\mathcal{A}\\ \nu_1,\nu_2}}\!e^{-\frac{(\nu_1-\nu_2)^2}{8\sigma_E^2}}\!\!\! \frac{1}{1+e^{\frac{\beta(\nu_2-\nu_1)}{2}}}\overline{\widehat{f}(\nu_1)}\widehat{f}(\nu_2)(A^\alpha_{\nu_1})^\dagger A^\alpha_{\nu_2}
\end{align}
and the map $\Phi_{\sigma_E,\widehat{f},H}$ defined on $\mathscr{F}$ as
\begin{align*}
\Phi_{\sigma_E,\widehat{f},H}(\ket{E_i}\bra{E_j})&:=\sum_{\substack{\alpha\in\mathcal{A}\\ \nu_1,\nu_2}}e^{-\frac{(\nu_1-\nu_2)^2}{8\sigma_E^2}}\, \overline{\widehat{f}(\nu_1)}\widehat{f}(\nu_2)\, A^\alpha_{\nu_2} \ket{E_i} (A^\alpha_{\nu_1}\ket{E_j})^\dagger.
\end{align*}
\noindent Indeed, that $G_{\sigma_E}$, resp.~$\Phi_{\sigma_E,\widehat{f},H}$, is well-defined on $\mathcal{F}$, resp.~on $\mathscr{F}$, follows directly from the estimates in the proof of \Cref{lem:definL} in the limit $\sigma_E\to \infty$ and the fact that $|e^{\smash{-\frac{(\nu_1-\nu_2)^2}{8\sigma_E^2}}}|\le 1$ uniformly over $\nu_1,\nu_2$. Moreover, by direct computation we have that, as argued in the proof of \Cref{propDirichlettoSchro}, 
\begin{align*}
\cL_{\sigma_E,\widehat{f},H}(\ket{E_i}\bra{E_j})= G_{\sigma_E}\ket{E_i}\bra{E_j}+\ket{E_i}(G_{\sigma_E}\ket{E_j})^\dagger +\Phi_{\sigma_E,\widehat{f},H}(\ket{E_i}\bra{E_j}).
\end{align*}

\noindent Next, we show that a positive spectral gap cannot decrease as
$\sigma_E$ decreases, provided the larger-width generator has a simple
kernel (see also \cite{slezak2026polynomial} for a finite-dimensional
analogue of the result):

\begin{proposition}\label{lem.gapmonoton}
For $0<\sigma_E\le\sigma'_E\le\infty$, assume that $\ker L_{\sigma'_E,\widehat f,H}=\mathbb C\sqrt{\sigma_\beta}$ and $\operatorname{gap}(L_{\sigma'_E,\widehat f,H})>0$.
Then $\ker L_{\sigma_E,\widehat f,H}=\mathbb C\sqrt{\sigma_\beta}$ and
\[
\operatorname{gap}(L_{\sigma_E,\widehat f,H})\ge\operatorname{gap}(L_{\sigma'_E,\widehat f,H}).
\]
Moreover, $\sigma_\beta$ is the unique invariant density operator
of the semigroup generated by $\mathcal L_{\sigma_E,\widehat f,H}$, i.e.~$\ker\mathcal L_{\sigma_E,\widehat f,H}=\mathbb C\sigma_\beta.$
\end{proposition}
\begin{proof}
We make use of the variational formulation of the gap:
\begin{align*}
\operatorname{gap}(L_{\sigma'_E,\widehat{f},H})=\inf_{x\in D(\mathcal{E}_{\sigma'_E,\widehat{f},H})\backslash \mathbb{C}\sigma_\beta^{\frac{1}{2}}}\frac{\mathcal{E}_{\sigma'_E,\widehat{f},H}(x)}{\Big\|x-\langle \sigma_\beta^{\frac{1}{2}},x\rangle\,\sigma_\beta^{\frac{1}{2}}\Big\|_{2}^2},
\end{align*}
where $\mathcal{E}_{\smash{\sigma'_E,\widehat{f},H}}$ is the Dirichlet form associated with $L_{\smash{\sigma'_E,\widehat{f},H}}$ and we used the assumption $\ker L_{\sigma'_E,\widehat f,H}=\mathbb C\sqrt{\sigma_\beta}$.
The case $\sigma_E=\sigma'_E$ is immediate, so assume that $\sigma_E<\sigma'_E$.
Therefore, denoting the probability density $g_{\sigma_E}(s):=\sigma_E\sqrt{2/\pi}e^{-2\sigma_E^2s^2}$ and using \eqref{integralrepLse} as well as the form core property of $\mathscr{F}_{\sigma_\beta}$, for all $x\in D(\mathcal{E}_{\sigma_E,\widehat{f},H})\backslash \mathbb{C}\sigma_\beta^{\smash{\frac{1}{2}}}$, there exists a sequence of elements $x_n\in\mathscr{F}_{\sigma_\beta}\backslash \mathbb{C}\sigma_\beta^{\smash{\frac{1}{2}}}$ such that $\|x-x_n\|_{D(\mathcal{E}_{\sigma_E,\widehat{f},H})}\to 0$ as $n\to\infty$ and
\begin{align*}
\frac{\mathcal{E}_{\sigma_E,\widehat{f},H}(x)}{\Big\|x-\langle\sigma_\beta^{\frac{1}{2}},x\rangle\,\sigma_\beta^{\frac{1}{2}}\Big\|_2^2}
&=\lim_{n\to\infty}\frac{\mathcal{E}_{\sigma_E,\widehat{f},H}(x_n)}{\Big\|x_n-\langle\sigma_\beta^{\frac{1}{2}},x_n\rangle\,\sigma_\beta^{\frac{1}{2}}\Big\|_2^2}\\
&=\lim_{n\to\infty}\int_{-\infty}^{\infty}\frac{g_{\sigma_E}(s)\mathcal{E}_{\widehat{f},H}(x_n(-s))}{\Big\|x_n(-s)-\langle\sigma_\beta^{\frac{1}{2}},x_n(-s)\rangle\,\sigma_\beta^{\frac{1}{2}}\Big\|_2^2}\,ds,
\end{align*}
where we also used the invariance of the denominator under the unitaries $e^{isH}$.
Now, for any such $x$, denoting $\mathcal{E}_n(s):=\mathcal{E}_{\smash{\widehat{f},H}}(x_n(s))$, the above equations show that $\mathcal{E}_{\sigma_E,\widehat{f},H}(x_n)=\mathcal{E}_n\ast g_{\sigma_E}(0)$.
For $\sigma'_E<\infty$, since $g_{\sigma_E}=g_{\sigma'_E}\ast g_{\sigma_E^*}$ with $(\sigma_E^*)^2:=\sigma_E^2(\sigma'_E)^2/((\sigma'_E)^2-\sigma_E^2)$, and by associativity of the convolution, we get
\begin{align*}
\mathcal{E}_{\sigma_E,\widehat{f},H}(x_n)
&=\mathcal{E}_n\ast g_{\sigma_E}(0)=\mathcal{E}_n\ast g_{\sigma'_E}\ast g_{\sigma_E^*}(0)\\
&=(s\mapsto\mathcal{E}_{\sigma'_E,\widehat{f},H}(x_n(s)))\ast g_{\sigma_E^*}(0).
\end{align*}
For $\sigma'_E=\infty$, the same identity holds with $\sigma_E^*:=\sigma_E$.
Thus, using the variational formulation above and that $g_{\sigma_E^*}$ is a probability density, we obtain
\begin{align*}
\frac{\mathcal{E}_{\sigma_E,\widehat{f},H}(x_n)}{\Big\|x_n-\langle\sigma_\beta^{\frac{1}{2}},x_n\rangle\,\sigma_\beta^{\frac{1}{2}}\Big\|_2^2}
&=\int_{-\infty}^{\infty}\frac{g_{\sigma_E^*}(s)\mathcal{E}_{\sigma'_E,\widehat{f},H}(x_n(-s))}{\Big\|x_n(-s)-\langle\sigma_\beta^{\frac{1}{2}},x_n(-s)\rangle\,\sigma_\beta^{\frac{1}{2}}\Big\|_2^2}\,ds\\
&\ge\operatorname{gap}(L_{\sigma'_E,\widehat{f},H}).
\end{align*}
Passing to the limit $n\to\infty$, the same lower bound holds for the quotient at $x$.
Since $\operatorname{gap}(L_{\sigma'_E,\widehat{f},H})>0$, this gives $\ker L_{\sigma_E,\widehat f,H}=\mathbb C\sqrt{\sigma_\beta}$.
We can therefore apply the variational formulation to $L_{\sigma_E,\widehat{f},H}$, from which we directly read off that $\operatorname{gap}(L_{\sigma_E,\widehat{f},H})\ge\operatorname{gap}(L_{\sigma'_E,\widehat{f},H})$.
Finally, convergence to the stationary projection established in Section~\ref{sec:generalframe} gives $\ker\mathcal L_{\sigma_E,\widehat f,H}=\mathbb C\sigma_\beta$, and hence uniqueness of the invariant density operator.
\end{proof}

  \begingroup
\begin{corollary}[Downward rates under Gaussian convolution]
\label[corollary]{cor:SharpPolynomialFilterGaussian}
Under the hypotheses of Theorem~\mbox{%
\ref{thm:DownwardRateCriterion}}\hskip0pt%
and with $d_*$ defined in \eqref{eq:DownwardRateCriterion},
for every $0<\sigma_E\le\infty$,
\[
\operatorname{gap}(L_{\sigma_E,\widehat f,h(N)})>0
\quad\Longleftrightarrow\quad d_*>0.
\] 
  When $d_*>0$, the gap is at least $d_*/(4K)$, independently of $\sigma_E$.
In particular, the polynomial filters $\widehat f_\eta$ defined in
\eqref{eq:PolynomialFilter}, $\eta\ge0$, satisfy
\begin{equation}\label{eq:PolynomialGapGaussian}
\operatorname{gap}(L_{\sigma_E,\widehat f_\eta,h(N)})>0
\quad\Longleftrightarrow\quad M_\eta(h)<\infty,
\end{equation}
where $M_\eta(h)$ is defined in \eqref{eq:PolynomialSpacingCriterion},
with lower bound $c_{\beta,h}/M_\eta(h)$ whenever $M_\eta(h)<\infty$.
 \end{corollary}
\begin{proof}
  Proposition~\mbox{%
\ref{lem.gapmonoton} }\hskip0pt%
transfers the lower bound of
Theorem~\mbox{%
\ref{thm:DownwardRateCriterion} }\hskip0pt%
to every Gaussian width.
If $d_*=0$, strict positivity of the individual rates supplies a
subsequence tending to zero. The necessity proof of Theorem~\mbox{%
\ref{thm:DownwardRateCriterion}
}\hskip0pt%
applies unchanged for every $0<\sigma_E\le\infty$: its diagonal test
vectors $x_n$ commute with $H$, so \eqref{integralrepLse} gives
$\mathcal E_{\sigma_E,\widehat f,H}(x_n)
=\mathcal E_{\widehat f,H}(x_n)=\pi_nd_n$ for finite $\sigma_E$.
Corollary~\ref{thm:uniqueGibbsSurgical} gives
$\ker L_{\sigma_E,\widehat f,h(N)}=\mathbb C\sqrt{\sigma_\beta}$
at every such width. The same variational bound therefore yields zero
spectral gap at every width.
The result for the polynomial filter follows from
\mbox{%
\eqref{eq:PolynomialDeathBound}}\hskip0pt%
.
 \end{proof}

  \par\endgroup
 \subsection{Integral representation}\label{sec:IntegralRepresentation}

Next, we provide an integral formulation of the generator $\mathcal{L}_{\sigma_E,\widehat{f},H}$ for any $\sigma_E\in(0,\infty)$ and sufficiently nice functions $\widehat{f}$. This form will play a crucial role when considering an implementation of the associated evolution on a discrete variable quantum platform.   We assume $f\in L^1(\mathbb R)$,
with the Fourier-transform convention stated after
\eqref{def:Lalpha}.
All operator integrals below are understood strongly: after
application to a vector in the stated domain, they are Bochner
integrals in $\mathcal H$; the completely positive part is
integrated in trace norm. No smoothness of $f$ is required.

We start by providing integral representations for each constituent of the generator $\mathcal{L}_{\sigma_E,\widehat{f},H}$.

\begin{proposition}\label{integralrep}
Assume Condition \ref{eq:condAalphas} is satisfied for some $f\in
L^1(\mathbb R)$ with Fourier transform $\widehat{f}(\nu):=\int_{-\infty}^\infty f(s)e^{is\nu}ds$. Then for any $\alpha\in\mathcal{A}$, $(L^\alpha,\mathcal{F})$ extends to an operator in $D(\widetilde{H}^{\gamma}),$ which is relatively $\widetilde H^\gamma$-bounded, and for any $\ket{\psi}\in D(\widetilde{H}^{\gamma})$, 
\begin{align*}
L^\alpha\ket{\psi}=\int f(s)\,e^{isH}A^\alpha e^{-isH}\ket{\psi}\,ds=\sum_{E,E'\in\operatorname{Sp}(H)}\widehat{f}(E'-E)P_{E'}A^\alpha P_{E}\ket{\psi}.
\end{align*}
where the latter expression is to be understood as a limit of the finite sums $\sum_{\smash{E\le M,E'\le M'}}$ in $\mathcal{H}$. Moreover, $\widetilde{H}^{\gamma }L^\alpha \widetilde{H}^{-\mu}$ extends to a bounded operator with integral representation
\begin{align*}
\widetilde{H}^{\gamma }L^\alpha \widetilde{H}^{-\mu}:=\int f(s) e^{isH}\widetilde{H}^{\gamma } A^\alpha \widetilde{H}^{-\mu}e^{-isH}\,ds.
\end{align*}
Next, $(G_{\sigma_E},\mathcal{F})$ extends to an operator on $D(\widetilde{H}^{\mu}),$ which is relatively $\widetilde H^\mu$-bounded, and for any {$\sigma_E\in(0,\infty)$,} $\ket{\psi}\in D(\widetilde{H}^{\mu})$, $\ket{\psi}\in D((L^\alpha)^\dagger L^\alpha )\cap D(G_{\sigma_E})$ and
\begin{align*}
G_{\sigma_E}\ket{\psi}&=-\sum_{\alpha\in\mathcal{A}}\int_{-\infty}^\infty g(t)\,e^{itH}(L^\alpha)^\dagger L^\alpha e^{-itH}\ket{\psi}\,dt\\
&=-\!\!\!\!\!\!\!\!\!\!\!\sum_{\substack{\alpha\in\mathcal{A}\\E,E',E''\in\operatorname{Sp}(H)}}\!\!\!\!\!\!\!\!\!\!\widehat g(E'-E'')\overline{\widehat{f}(E-E')}\widehat{f}(E-E'')P_{E'}(A^\alpha)^\dagger P_{E}A^\alpha P_{E''}\ket{\psi},
\end{align*}
with $g(t)=\frac{1}{2\pi}\int_{-\infty}^\infty \frac{e^{-\nu^2/8\sigma_E^2}}{1+e^{\beta\nu/2}} e^{-i\nu t}d\nu$. Similarly, $\Phi_{\sigma_E,\widehat{f},H}$ extends to an operator on $D(\mathcal{W}_H^{\gamma,\gamma})$, which is relatively $\mathcal{W}^{\gamma,\gamma}_H$ bounded, and for any $x\in D(\mathcal{W}_H^{\gamma,\gamma})$ 
\begin{align}
\Phi_{\sigma_E,\widehat{f},H}(x)=\sigma_E\sqrt{\frac{2}{\pi}}\sum_{\alpha}\int_{\mathbb{R}}e^{-2\sigma_E^2 s^2}\, \, X^{\alpha}_s \cdot x\cdot  (X_s^\alpha)^\dagger \,ds
\end{align}
with $X^{\alpha}_s:=e^{isH}L^\alpha e^{-isH}$ and where we recall that $X_s^\alpha\cdot x\cdot (X_s^\alpha)^\dagger \equiv X_s^\alpha\widetilde{H}^{-\gamma}(X^\alpha_s\widetilde{H}^{-\gamma}\mathcal{W}_H^{\gamma,\gamma}(x))^\dagger)^\dagger$.
\end{proposition}
\begin{proof}
Strong continuity of the unitary group,
first on the weighted domains and then after applying the bounded
weighted jump operators, ensures measurability. Absolute
integrability follows from
\[
\begin{aligned}
\|L^\alpha\widetilde H^{-\gamma}\|
&\le\|f\|_1\|A^\alpha\widetilde H^{-\gamma}\|,\\
\|\widetilde H^\gamma L^\alpha\widetilde H^{-\mu}\|
&\le\|f\|_1
\|\widetilde H^\gamma A^\alpha\widetilde H^{-\mu}\|.
\end{aligned}
\]
The analogous first bound holds for the adjoints, using
$\overline f$ and $(A^\alpha)^\dagger$. These estimates justify
the dominated-convergence arguments below under the sole filter
assumption $f\in L^1(\mathbb R)$.

For $\ket{\psi}\in \mathcal{F}$ the operator $L^\alpha$ defined in \eqref{eq:LalphainGenerationTheory} satisfies 
\begin{align*}
L^\alpha\ket{\psi}&=\sum_{E,E'\in\spec(H) }\widehat{f}(E'-E)P_{E'}A^\alpha P_{E}\ket{\psi}=\sum_{E,E'\in \operatorname{Sp}(H)}\int_{-\infty}^\infty f(s)e^{is(E'-E)}ds P_{E'}A^\alpha P_E\ket{\psi}\\
&= \lim_{M\to \infty} \lim_{M'\to \infty}\int_{-\infty}^\infty f(s) \sum_{\substack{E,E'\in \operatorname{Sp}(H)\\E\le M,\,E'\le M'}}e^{is(E'-E)}ds P_{E'}A^\alpha P_E\ket{\psi} \\
&= \lim_{M\to \infty} \lim_{M'\to \infty} \int_{-\infty}^\infty f(s)\, P_{H\le M'}e^{isH}A^\alpha e^{-isH} P_{H\le M}\ket{\psi}\,ds  
\end{align*}
where we denoted for $M\in\N_0$ the projection $ P_{H\le M} =  \sum_{\substack{E\in \operatorname{Sp}(H)\\E\le M}}P_E.$ Note that
\begin{align*}
   \| P_{H\le M'}e^{isH}A^\alpha e^{-isH} P_{H\le M}\ket{\psi}\| \le    \|A^\alpha\widetilde H^{-\gamma}\| \|\widetilde    H^\gamma\ket{\psi}\|
\end{align*}
for all $M,M'\in\N$ and therefore, using $f\in L^1(\R),$ we can use dominated convergence to exchange the limit $M,M'\to\infty$ with the integral and obtain from the above\begin{align*}
L^\alpha\ket{\psi}=\int_{-\infty}^\infty\,f(s)\sum_{E,E'\in\operatorname{Sp}(H)}e^{is(E'-E)}P_{E'}A^\alpha P_E\ket{\psi}\,ds=\int_{-\infty}^\infty f(s)e^{isH}A^\alpha e^{-isH}\ket{\psi}\,ds\,.
\end{align*}
Clearly, the above swapping of sums and integral extend to $\ket{\psi}\in D(\widetilde{H}^{\gamma})$, . We argue similarly for $G_{\sigma_E}$ as defined in \eqref{GsigmaEdeffirst}: for any $\ket{\psi}\in \mathcal{F}$,
  \begin{align*}
&G_{\sigma_E}\ket{\psi}\\
&\, =-\sum_{\alpha\in\mathcal{A}}\lim_{M,M'\to\infty}\int_{-\infty}^\infty\sum_{\substack{E,E',E''\in\operatorname{Sp}(H)\\E\le M,E'\le M'}}\!\!\!g(t)\overline{\widehat{f}(E-E')}\widehat{f}(E-E\prime\prime)\\
&\hspace{30mm}\times e^{itH}P_{E'}(A^\alpha)^\dagger P_{E}A^\alpha e^{-itH}P_{E''}\ket{\psi}dt\\
&\, =\!-\!\sum_{\alpha\in\mathcal{A}}\lim_{M,M'\to\infty}\int_{-\infty}^\infty \!\!P_{H\le M'}\!\!\!\!\!\!\!\!\!\!\!\!\sum_{E,E',E''\in\operatorname{Sp}(H)}\!\!\!\!\!\!\!\!\!\!g(t)\overline{\widehat{f}(E-E')}\widehat{f}(E-E'')\\
&\hspace{30mm}\times e^{itH}P_{E'}(A^\alpha)^\dagger P_{E}P_{H\le M}A^\alpha e^{-itH}P_{E''}\ket{\psi}dt\\
&\,=\!-\!\sum_{\alpha\in\mathcal{A}}\lim_{M,M'\to\infty}\int_{-\infty}^\infty \!\,g(t)e^{itH}P_{H\le M'}(L^\alpha)^\dagger P_{H\le M}L^\alpha e^{-itH}\ket{\psi}dt.
\end{align*} 
Moreover, the term above is uniformly integrable, since $g\in L^1(\mathbb{R})$ and 
  \begin{align*}
\|P_{H\le M'}(L^\alpha)^\dagger P_{H\le M} L^\alpha e^{-itH}\ket{\psi}\|\le  \|\widetilde{H}^{-\gamma}L^\alpha\|\,\|\widetilde{H}^\gamma L^\alpha \widetilde{H}^{-\mu}\|\,\|\widetilde{H}^\mu\ket{\psi}\|<\infty
\end{align*} 
Therefore, using the dominated convergence theorem to exchange the limits $M,M'\to \infty$ with the integral gives
   \begin{align*}
G_{\sigma_E}\ket{\psi}&=-\sum_{\alpha}\int_{-\infty}^{\infty}g(t)\,\sum_{E,E',E''\in\operatorname{Sp}(H)}\overline{\widehat{f}(E-E')}\widehat{f}(E-E\prime\prime)e^{itH}P_{E'}(A^\alpha)^\dagger P_E A^\alpha e^{-itH}P_{E\prime\prime}\ket{\psi}\\
&=-\sum_{\alpha}\int_{-\infty}^{\infty}g(t)\,e^{itH}(L^\alpha)^\dagger  L^\alpha e^{-itH}\ket{\psi}
\end{align*} 
where in the last line we used that $L^\alpha\ket{\psi}\in D(\widetilde{H}^{\gamma})$. 
Once again, the above swapping of sums and integrals extend to $\ket{\psi}\in D(\widetilde{H}^{\mu})$.

For the CP-term $\Phi_{\sigma_E,\widehat f, H}$ we apply a similar dominated convergence argument as for the $G_{\sigma_E}$ operator to exchange the Bohr frequency series with the integral: given $x\in \mathscr{F}$
\begin{align*}
&\Phi_{\sigma_E,\widehat{f},H}(x)\\
&\quad =\sum_{\substack{\alpha\in\mathcal{A}\\ E_1,E_1',E_2,E_2'\in\operatorname{Sp}(H)}}e^{-\frac{(E_1-E_1'-E_2+E_2')^2}{8\sigma_E^2}}\, \overline{\widehat{f}(E_1-E_1')}\widehat{f}(E_2-E_2')\, P_{E_2} A^\alpha P_{E_2'} x P_{E_1'}(A^\alpha)^\dagger P_{E_1}\\
&\quad =\sigma_E\sqrt{\frac{2}{\pi}}\!\!\!\!\!\!\!\!\!\!\!\!\!\!\!\sum_{\substack{\alpha\in\mathcal{A}\\ E_1,E_1',E_2,E_2'\in\operatorname{Sp}(H)}}\!\!\!\!\!\!\!\!\!\!\!\!\!\!\!\int_{-\infty}^\infty e^{-2\sigma_E^2 s^2}  e^{is (-E_1+E_1'+E_2-E_2')} \overline{\widehat{f}(E_1-E_1')}\widehat{f}(E_2-E_2')\, P_{E_2} A^\alpha P_{E_2'} x P_{E_1'}(A^\alpha)^\dagger P_{E_1}ds\\
&\quad =\sigma_E\sqrt{\frac{2}{\pi}}\!\!\!\!\!\!\!\!\!\!\!\!\!\!\!\sum_{\substack{\alpha\in\mathcal{A}\\ E_1,E_1',E_2,E_2'\in\operatorname{Sp}(H)}}\!\!\!\!\!\!\!\!\!\!\!\!\!\!\!\!\int_{-\infty}^\infty \!\!\!\!\!e^{-2\sigma_E^2 s^2}   \overline{\widehat{f}(E_1-E_1')}\widehat{f}(E_2-E_2')\, P_{E_2} e^{isH} A^\alpha e^{-isH}P_{E_2'} x e^{isH}P_{E_1'}(A^\alpha)^\dagger e^{-isH}P_{E_1}ds.
\end{align*}
Arguing by truncation as above, it suffices to show that 
\begin{align*}
&P_{H\le M}\sum_{E_1,E_2,E_1',E_2'\in\operatorname{Sp}(H)} \overline{\widehat{f}(E_1-E_1')}\widehat{f}(E_2-E_2')\, P_{E_2} e^{isH} A^\alpha e^{-isH}P_{E_2'} x e^{isH}P_{E_1'}(A^\alpha)^\dagger e^{-isH}P_{E_1} P_{H\le M'}\\
&\qquad =P_{H\le M}\,\sum_{E_2',E_1'\in\operatorname{Sp}(H)}\,e^{isH}L^\alpha  e^{-isH}P_{E_2'}x P_{E_1'}e^{isH}(L^\alpha)^\dagger e^{-isH}P_{H\le M'}\\
&\qquad =P_{H\le M}\,\,e^{isH}L^\alpha  e^{-isH}x e^{isH}(L^\alpha)^\dagger e^{-isH}P_{H\le M'}
\end{align*}
with norm
  \begin{align*}
\|P_{H\le M}\,\,e^{isH}L^\alpha  e^{-isH}x e^{isH}(L^\alpha)^\dagger e^{-isH}P_{H\le M'}\|_1\le \|x\|_{\mathcal{W}_H^{\gamma,\gamma}}\|L^\alpha \widetilde{H}^{-\gamma}\|^2<\infty.
\end{align*} 
Thus, once again, by the dominated convergence theorem, we conclude that 
\begin{align*}
\Phi_{\sigma_E,\widehat{f},H}(x)&=\sigma_E\sqrt{\frac{2}{\pi}}\sum_{\alpha\in\mathcal{A}}\int_{-\infty}^\infty e^{-2\sigma_E^2 s^2}   \,e^{isH}L^\alpha  e^{-isH}\cdot x\cdot  e^{isH}(L^\alpha)^\dagger e^{-isH}ds\\
&=\sigma_E\sqrt{\frac{2}{\pi}}\sum_{\alpha}\int_{\mathbb{R}}e^{-2\sigma_E^2 s^2}\, \, X_s^{\alpha} \cdot x\cdot  (X_s^\alpha)^\dagger \,ds.
\end{align*}
Again, the same swap holds and extends to $x\in \mathcal{W}_H^{\gamma,\gamma}$. 

\end{proof}

\begin{corollary}
 Assume that \Cref{eq:condAalphas} is satisfied for some  $f\in
L^1(\mathbb R)$ with Fourier transform  $\widehat{f}(\nu):=\int_{-\infty}^\infty f(s)e^{is\nu}ds$. Then,  $D(\mathcal{W}_H^{\gamma,\gamma})\cap D(\mathcal{W}_H^{0,\mu})\cap  D(\mathcal{W}_H^{\mu,0})\subseteq  D(\mathcal{L}_{\sigma_E,\widehat{f},H})$, and for any $x\in  D(\mathcal{W}_H^{\gamma,\gamma})\cap D(\mathcal{W}_H^{0,\mu})\cap  D(\mathcal{W}_H^{\mu,0})$,
 \begin{align*}
\mathcal{L}_{\sigma_E,\widehat{f},H}(x):=\Phi_{\sigma_E,\widehat{f},H}(x)+G_{\sigma_E}\cdot x+x\cdot G_{\sigma_E}^\dagger.
 \end{align*}
\end{corollary}
\begin{proof}
This follows the same lines as in \Cref{prop.generationtheorem}, by invoking the closedness of $\mathcal{L}_{\sigma_E,\widehat{f},H}$ and density of $\mathscr{F}$ in $D(\mathcal{W}_H^{\gamma,\gamma})\cap D(\mathcal{W}_H^{0,\mu})\cap D(\mathcal{W}_H^{\mu,0})$:
given  $x\in D(\mathcal{W}_H^{\gamma,\gamma})\cap D(\mathcal{W}_H^{0,\mu})\cap D(\mathcal{W}_H^{\mu,0})$, $x_n:=P_{H\le n}x P_{H\le n}\in \mathcal{W}_{H}^{\mu,0}$ with $x_n\to x$ in $\mathscr{T}_1(\cH)$ as $n\to\infty$. 
Then, by \Cref{integralrep},
\begin{align*}
G_{\sigma_E}\cdot x_n=-\sum_{\alpha\in\mathcal{A}} \int_{-\infty}^\infty g(t) e^{itH}(L^\alpha)^\dagger L^\alpha e^{-itH}x_n dt\to -\sum_{\alpha\in \mathcal{A}}\int_{-\infty}^\infty g(t) e^{itH}(L^\alpha)^\dagger L^\alpha e^{-itH}x dt\equiv G_{\sigma_E} \cdot x
\end{align*}
by dominated convergence theorem. Similarly, we show  that $x_n\cdot  G_{\sigma_E}^\dagger \to  x\cdot G_{\sigma_E}^\dagger$ and 
\begin{align*}
\Phi_{\sigma_E,\widehat{f},H}(x_n)=\sigma_E\sqrt{\frac{2}{\pi}}\sum_{\alpha\in\mathcal{A}}\int_{\mathbb{R}}e^{-2\sigma_E^2s^2} X_s^\alpha \cdot x_n\cdot  (X_s^\alpha)^\dagger \,ds\to \sigma_E\sqrt{\frac{2}{\pi}}\sum_{\alpha\in\mathcal{A}}\int_{\mathbb{R}}e^{-2\sigma_E^2s^2} X_s^\alpha \cdot x\cdot  (X_s^\alpha)^\dagger \,ds.
\end{align*}
The result follows.
\end{proof}

\subsection{Finite-dimensional truncations}
\label{sec:FiniteTruncBareJumps}
In this section we consider a finite-dimensional\footnote{On the original Hilbert space the truncated generator is bounded, generally not a finite-rank superoperator. Its restriction to the reducing system register is finite dimensional.} truncation of the unbounded generator $\mathcal{L}_{\sigma_E,\widehat{f},H}$ which for truncation parameter $M\in\N$ is denoted by $\mathcal{L}^{\le M}_{\sigma_E,\widehat{f},H_{\le M}}.$
We show in the following that $\mathcal{L}_{\sigma_E,\widehat{f},H}$ is well-approximated by the generator $\mathcal{L}^{\le M}_{\sigma_E,\widehat{f},H_{\le M}}$
 on certain energy constrained input states and for truncation  parameter $M$ large enough. 

 In Section~\ref{sec:TruncateJumps}, we first introduce the bounded  generator $\mathcal{L}^{\le M}_{\sigma_E,\widehat{f},H}$ by  truncating the bare jumps $A^\alpha$ within the unbounded generator  and continue to show closeness of this bounded generator to  $\mathcal{L}_{\sigma_E,\widehat{f},H}.$ Then in  Section~\ref{sec:FinitDimGenerator} we further truncate the  Hamiltonian to define the finite-dimensional generator  $\mathcal{L}^{\le M}_{\sigma_E,\widehat{f},H_{\le M}}$ and show  closeness to $\mathcal{L}^{\le M}_{\sigma_E,\widehat{f},H}.$  Lastly, we show in Section~\ref{sec:FiniteDimPreparationIntegrable}  that the dynamics $e^{t\mathcal{L}_{\sigma_E,\widehat{f},H}}$ is  well approximated by the finite-dimensional one  $e^{t\mathcal{L}^{\le M}_{\sigma_E,\widehat{f},H_{\le M}}}.$

\subsubsection{Truncating the bare jumps}
\label{sec:TruncateJumps}
In the following section we focus on approximating bare jumps $\{A^\alpha\}$ by finite-dimensional\footnote{Technically, the single-mode  truncated bare jumps  are finite rank  on   the original  infinite-dimensional space.   Their tensor extensions to several modes are bounded local operators and generally have infinite rank .  A finite-dimensional system register is introduced in Section~\ref{sec:FinitDimGenerator}. } bare jumps which will be useful for implementing the Lindlabian dynamics considered in the previous sections on finite-dimensional hardware.

To illustrate this, let us focus first on bare jumps being the  annihilation and creation operators. In particular for truncation parameter $M\in\N,$ we consider $\pi_M:= \sum_{n=0}^M\kb{n}$ and
\begin{align}
\label{eq:DefTruncAnnCrea}
    a^{\le M} :=  \pi_Ma\pi_M=  \sum_{n=1}^M     \sqrt{n}\ket{n-1}\!\bra{n},\quad\text{and}\quad\left(a^{\le     M}\right)^\dagger = \pi_Ma^\dagger     \pi_M=\sum_{n=0}^{M-1}\sqrt{n+1} \ket{n+1}\!\bra{n}.
\end{align}
More generally, for $k\in\N,$ we can also consider truncations of higher order bare jumps like $a^k$ and $(a^k)^\dagger$ which are defined by
\begin{align*}
    \left(a^k\right)^{\le M} &:= \pi_Ma^k\pi_M=\sum_{n=k}^M     \sqrt{\frac{n!}{(n-k)!}}\ket{n-k}\!\bra{n},\quad\text{and}
\\\left((a^k)^{\le M}\right)^\dagger &=  \pi_M (a^\dagger)^k\pi_M= \sum_{n=0}^{M-k}\sqrt{\frac{(n+k)!}{n!}}\ket{n+k}\!\bra{n}.
\end{align*}

In the following lemma we see that these finite-dimensional truncations approximate the untruncated bare jumps well when applied on energy constrained input states.

\begin{lemma}
\label{lem:FiniteTruncBareJumps}
Let $\kappa\in(0,1],$ $k\in\N$ and $M\in\N$ be such that  $M\ge \left(\frac{k}{2\kappa}\right)^{1/\kappa}+k.$  Then for $\ket{\psi}\in D(e^{N^\kappa})$ we have
\begin{align*}
    \left\|\left(a^k - (a^k)^{\le M}\right) \ket{\psi}\right\| &\le     M^{k/2}\,e^{-M^{\kappa}}\| e^{N^\kappa} \ket{\psi}\|,     \\\left\|\left((a^k)^\dagger - \left((a^k)^{\le     M}\right)^\dagger\right) \ket{\psi}\right\| &\lesssim M^{k/2}
    \,e^{-M^{\kappa}}\| e^{N^{\kappa}} \ket{\psi}\|,
\end{align*}
Furthermore, for $\psi\in D(e^{2N^{\kappa}})$ we have
\begin{align*}    \left\|e^{N^{\kappa}}(a^k - (a^k)^{\le M}) \ket{\psi}\right\| &\le M^{k/2}\,e^{-M^{\kappa}}\| e^{2N^{\kappa}} \ket{\psi}\|,\\
     \left\|e^{N^{\kappa}}\left((a^k)^\dagger - \left((a^k)^{\le      M}\right)^\dagger\right) \ket{\psi}\right\| &\lesssim M^{k/2}
    \,e^{-M^{\kappa}}\| e^{2N^{\kappa}} \ket{\psi}\|.
\end{align*}
 The constants hidden in the $\lesssim$-notation in the above  inequalities depend on $k$ and $\kappa$ but on no other variables.
\end{lemma}

\begin{proof}

The
 first inequality follows from noting
    \begin{align*}
    \left\|(a^k- (a^k)^{\le M}) e^{-N^{\kappa}}\right\| = \sup_{n\ge     \max\{M+1,k\}} \sqrt{\frac{n!}{(n-k)!}} e^{-n^\kappa}\le     \sup_{n\ge \max\{M+1,k\}} n^{k/2} e^{-n^\kappa}\le     M^{k/2}e^{-M^{\kappa}},
\end{align*}
where for the last inequality we have used that $M\ge \left(\frac{k}{2\kappa}\right)^{1/\kappa}+k$ and  that the function $ x^{k/2}e^{-x^{\kappa}}$ is non-increasing for $x\ge \left(\frac{k}{2\kappa}\right)^{1/\kappa}.$

Similarly, we see
\begin{align*}
\left\|\left((a^k)^\dagger - \left((a^k)^{\le M}\right)^\dagger\right)e^{-N^\kappa}\right\|&\le \sup_{n\ge M-k+1}\sqrt{\frac{(n+k)!}{n!}}e^{-n^{\kappa}} \le \sup_{n\ge M-k+1}(n+k)^{k/2}e^{-n^{\kappa}}   \\&\le e^{\kappa(k-1)}(M+1)^{k/2}e^{-M^{\kappa}},
\end{align*}
where for the last inequality we have used that the function $(x+k)^{k/2}e^{-x^\kappa}$ is non-increasing for $x\ge \left(\frac{k}{2\kappa}\right)^{1/\kappa}$ and furthermore that $e^{M^{\kappa}-(M-k+1)^\kappa}\le e^{\kappa(k-1)}.$

The third and fourth inequality follow similarly while noting for the fourth inequality that   $e^{(n+k)^{\kappa}-n^\kappa}\le e^{k\kappa}$ for $n\ge1$, which is the range of the tail indices in that estimate.
 \end{proof}

Going beyond, we want to consider a finite-dimensional truncation scheme for a set of general bare jumps $\{A^\alpha\}_{\alpha\in\mathcal{A}}$ on some separable Hilbert space $\cH:$ 
 For $M\in \N$ and $\alpha\in\cA$ we consider  projections $\pi^{\alpha}_M$    of rank $M+1$ on the  acted-on factor; in a multimode system they may be tensored with  identities on the other factors.  Assuming that  $\pi^{\alpha}_M\cH\subseteq D(A^{\alpha})$ we then define the  truncated bare jump $\left(A^{\alpha}\right)^{\le M} :=  \pi^{\alpha}_M A^{\alpha}\pi^{\alpha}_M.$ We assume that we have  good control on the operator norm of the truncated jump, i.e.  precisely 
  \begin{align}
\label{eq:NormTruncatedJump}
    \left\|(A^\alpha)^{\le M}\right\| \le q'(M),
\end{align} 
for some polynomially bounded function $q'(M).$
 We choose $(A^{\bar\alpha})^{\le M}=((A^\alpha)^{\le M})^\dagger$ whenever $A^{\bar\alpha}=(A^\alpha)^\dagger$. 

    For the remaining section we assume that $A^{\alpha}$ is     well-approximated by the truncated bare jump on energy     constrained inputs: Precisely, to measure energy we consider a     self-adjoint and positive semidefinite operator, $\NA,$ and     assume that for $l=1,2$ and some $\kappa\in(0,1/2)$ we have
\begin{align}
\label{eq:AbstractFiniteTruncBareJumps}  
    \left\|e^{(l-1)\NA^{\kappa}}(A^{\alpha}-(A^{\alpha})^{\le     M})e^{-l\NA^{\kappa}}\right\| \le q(M)e^{-M^{\kappa}},
\end{align}
where $q(M)$ is some polynomially bounded function.  For certain parts of the proofs of Section~\ref{sec:FinitDimGenerator} we also assume that
\begin{align}
\label{eq:NormTruncatedJump2}
\left\|e^{\NA^{\kappa}}(A^{\alpha})^{\le M} e^{-\NA^{\kappa}}\right\| \le q'(M)\quad\text{and}\quad\left\|e^{-\NA^{\kappa}}(A^{\alpha})^{\le M} e^{\NA^{\kappa}}\right\| \le q'(M)
\end{align}
for the polynomially bounded function $q'(M)$ which also appeared in \eqref{eq:NormTruncatedJump}.

To illustrate the above truncation scheme, we consider the example of an $m$-mode bosonic system on the Hilbert space $\left(L^2(\R)\right)^{\otimes m}$: In this case, we can consider multi-indices $\alpha$ taken from the set $\cA=\{(i,+),(i,-)\}_{i=1}^m$ and bare jumps $\{A^\alpha\}_{\alpha\in \cA}=\left\{a^k_i,(a_i^k)^\dagger\right\}_{i=1}^m $ for some $k\in \N.$ Furthermore, for $\alpha= (i,+)$ or $\alpha=(i,-)$, we consider the specific choice of   projections with local rank $M+1$,  given by local truncations in the Fock basis, i.e., $\pi^{\alpha}_M \equiv \pi^i_M = \1_{1,\cdots,i-1}\otimes \pi_M\otimes \1_{i+1,\cdots,m} $, where $\pi_M = \sum_{n=0}^{M} \kb{n}.$ 
 In this case, \eqref{eq:NormTruncatedJump} holds with $M^{k/2}$; to satisfy \eqref{eq:NormTruncatedJump2} with the same polynomial, take $q'(M)=C_{k,\kappa}(M+1)^{k/2}$. 
Lemma~\ref{lem:FiniteTruncBareJumps} ensures for this particular choice of bare jumps and truncations that the relation~\eqref{eq:AbstractFiniteTruncBareJumps} holds, where in this case we take $\NA = \Ntot = \sum_{i=1}^m N_i.$ Alternatively, in \cite{BeckerRouzeSalzmannSchroedinger} we consider $\NA$ to be a finite rank perturbation of $\Ntot$ and verify relation~\eqref{eq:AbstractFiniteTruncBareJumps} for this choice as well. 

Next, we consider the generator $\mathcal{L}_{\sigma_E,\widehat{f},H}$ with bare jumps $\{A^{\alpha}\}_{\alpha\in\cA}$, for which Proposition~\ref{integralrep} provided an explicit integral representation. We see in the following that $\mathcal{L}_{\sigma_E,\widehat{f},H}$ is close to the generator $\mathcal{L}^{\le M}_{\sigma_E,\widehat{f},H}$ which is defined by replacing $A^\alpha$ by $\left(A^{\alpha}\right)^{\le M}.$  For this we assume that energy measured with respect $e^{2\NA^\kappa}$ can only increase subexponentially with respect to time when evolved with the unitary dynamics generated by $H$: More precisely we assume for $k=2,4$ and $\kappa\in(0,1/2)$ as above that there exists $r\ge 0$ such that for all $t\in\R$ we have\footnote{The reason for the choice of range of $\kappa,$ in particular the constraint $\kappa<1/2$, is that this enables us to prove approximation of $\mathcal{L}_{\sigma_E,\widehat{f},H}$ by $\mathcal{L}^{\le M}_{\sigma_E,\widehat{f},H}$ for all $\beta>0.$ The reason for this is that the function $e^{r|t|^{2\kappa}/2}$ needs to be dominated by the function $g(t),$ which is featured in the integral representations of the generators stated in Proposition~\ref{integralrep}, and it can be shown that   $|g(t)|
\lesssim
e^{-\tfrac{\pi}{\beta}|t|}$  for large $t.$ 
Throughout the truncation and circuit results, the condition
$0<\kappa<1/2$, $r\ge0$ may alternatively be replaced by
$0<\kappa\le1$, $r=0$, with all other hypotheses retained.
Indeed, every time-growth factor $e^{r|t|^{2\kappa}}$ is then one,
so neither subadditivity of $|t|^{2\kappa}$ nor its domination by the
exponentially decaying time kernel is needed. The ladder estimates
use only concavity of $x^\kappa$ and remain valid up to $\kappa=1$.
In particular, when $[H,\NA]=0$ and the required Gibbs moment is finite,
$\kappa=1$ gives an error polynomial in $M$ times $e^{-M}$ and a cutoff
of logarithmic, rather than a higher logarithmic-power, order.}
\begin{align}\label{eq:ExpBound}
e^{-itH}e^{k\NA^\kappa} e^{itH}\le e^{r|t|^{2\kappa}}\,e^{k\NA^\kappa}.
\end{align}
Such a condition is natural and holds, e.g., for the choice $\NA=\Ntot = \sum_{i=1}^m N_i$ and the Bose-Hubbard model and for all $\kappa\in(0,1]$ and $r=0$ since the corresponding Hamiltonian commutes with $\Ntot$, and, as we see in Lemma~\ref{lem:mfBHEnergyLimited} below, for the mean field Bose-Hubbard Hamiltonian, which is studied in the companion paper \cite{BeckerRouzeSalzmannBose}, for all $\kappa\in[0,1/2]$ and some $r>0$ depending only on $\kappa$ and the interaction strength $\psi.$ 

Another natural choice of $\NA$ is given by $\NA = \widetilde H$ where $\widetilde H = H +h_0+1$ and $H \ge -h_0.$ For this choice \eqref{eq:ExpBound} is trivially satisfied for $r=0.$ The alternative range $0<\kappa\le1$ therefore applies. On the other hand, conditions~\eqref{eq:AbstractFiniteTruncBareJumps} and \eqref{eq:NormTruncatedJump2} are non-trivial and need to be verified explicitly in this case.

We now proceed to show that under conditions \eqref{eq:AbstractFiniteTruncBareJumps} and \eqref{eq:ExpBound}, the generators $\mathcal{L}^{\le M}_{\sigma_E,\widehat{f},H}$ and $\mathcal{L}_{\sigma_E,\widehat{f},H}$ are close. For this, we first prove the following technical Lemma. In what follows, we denote for some function $f\in
L^1(\mathbb R)$
\begin{align*}L^\alpha:=\int f(s)e^{isH}A^\alpha e^{-isH}ds,\qquad 
L^{\alpha,\le M}&:=\int f(s)e^{isH}(A^{\alpha})^{\le M} e^{-isH}ds.
\end{align*}

\begin{lemma}\label{lem:integralrepapproxTrunc}
Let $f\in
L^1(\mathbb R)$ and set of bare jumps $\{A^\alpha\}_{\alpha\in\cA}$ and assume  \Cref{eq:condAalphas} and, for some self-adjoint and positive semidefinite $\NA$ and $\kappa\in(0,1/2),$  \eqref{eq:AbstractFiniteTruncBareJumps}  are satisfied and further that $e^{\NA^\kappa}A^{\alpha}e^{-2\NA^{\kappa}},$ $e^{-\NA^{\kappa}}A^{\alpha}$ and $A^{\alpha}e^{-\NA^{\kappa}}$ are bounded.
Furthermore, assume that \eqref{eq:ExpBound} holds for $r\ge 0$ and $\kappa\in(0,1/2)$ as above and that
\begin{align}
\label{eq:C_f}
    C_f := \int_{-\infty}^{\infty} e^{r|t|^{2\kappa}} |f(t)|dt <     \infty.
\end{align}
Then for $M\in \N$ we have
\begin{align*}
&\|(L^{\alpha}\!-\!L^{\alpha,\le M})e^{-\NA^\kappa}\|\!\le\! C_f\,q(M) e^{-M^{\kappa}},\\
&\|((L^\alpha)^\dagger L^{\alpha}-\left(L^{\alpha,\le M}\right)^\dagger L^{\alpha,\le M})e^{-2\NA^{\kappa}}\|\le C^2_f C_{A}\,q(M)e^{-M^\kappa},\, \\
\end{align*}
where we denoted $C_A := \max_{\alpha\in\cA}\left(\|e^{\NA^{\kappa}}A^{\alpha}e^{-2\NA^{\kappa}}\|+\|A^{\alpha}e^{-\NA^{\kappa}}\|+\|e^{-\NA^{\kappa}}A^{\alpha}\|\right)+2Q<\infty$, $Q:=\sup_{M\ge1}q(M)e^{-M^\kappa}<\infty$. \end{lemma}

\begin{proof}

 We use \eqref{eq:AbstractFiniteTruncBareJumps} together with  \eqref{eq:ExpBound} to see for $\ket{\psi}\in  D\left(e^{\NA^{\kappa}}\right)$ that we have
\begin{align}
\label{eq:NLLeqTrunc}
\nn\|(L^{\alpha}-L^{\alpha,\le M})\ket{\psi}\|
 \nn&\le \int |f(s)|\|(A^{\alpha}-(A^{\alpha})^{\le  M})e^{-isH}\ket{\psi}\|\,ds\\&\nn\le  \int |f(s)|  \left\|(A^{\alpha}-(A^{\alpha})^{\le  M})e^{-\NA^{\kappa}}\right\|\left\|e^{\NA^{\kappa}}e^{-isH}\ket{ \psi}\right\| ds\\&\nn\le  \int |f(s)|  e^{\tfrac{r|s|^{2\kappa}}{2}}\left\|(A^{\alpha}-(A^{\alpha})^{\le  M})e^{-\NA^{\kappa}}\right\|\left\|e^{\NA^{\kappa}}\ket{\psi}\right \| ds\\&\le C_f\,q(M)e^{-M^{\kappa}}\left\|e^{\NA^{\kappa}}\ket{ \psi}\right\| ,
\end{align}
which shows the first inequality in Lemma~\ref{lem:integralrepapproxTrunc}. 

Next, we control squares of the $L$ operators for $\ket{\varphi}\in D(e^{2\NA^{\kappa}})$:
\begin{align}\label{eq:twotermstoboundTrunc}
&\nn\|((L^\alpha)^\dagger L^{\alpha}-\left(L^{\alpha,\le M}\right)^\dagger L^{\alpha,\le M})\ket{\varphi}\|\\
&\nn\quad\le \|(L^\alpha)^\dagger (L^{\alpha}-L^{\alpha,\le M})\ket{\varphi}\|\\
&\qquad+\|((L^\alpha)^\dagger-\left(L^{\alpha,\le M}\right)^\dagger)L^{\alpha,\le M}\ket{\varphi}\|.
\end{align}
The first term above can be controlled by realizing that by assumption we have for $\ket{\psi}\in D(e^{\NA^{\kappa}})$ that $\ket{\psi}\in D((L^\alpha)^\dagger)$ with
\begin{align}\label{eq:LiboundTrunc}
\|(L^\alpha)^\dagger\ket{\psi}\|&\le \int |f(s)|\,\|(A^{\alpha})^\dagger e^{-isH}\ket{\psi}\|\,ds\le \int |f(s)|\,\|(A^{\alpha})^\dagger e^{-\NA^{\kappa}}\|\| e^{\NA^{\kappa}}e^{-isH}\ket{\psi}\|\,ds\nn\\&
\le C_f \|e^{-\NA^{\kappa}}A^{\alpha}\| \|e^{\NA^{\kappa}}\ket{\psi}\|. 
\end{align}
Hence, combining with an analogous argument to \eqref{eq:NLLeqTrunc}, we get
for $\ket{\phi}\in D(e^{2\NA^{\kappa}})$
\begin{align}\label{eq:firstboundLLTrunc}
\|(L^\alpha)^\dagger\, (L^{\alpha}-L^{\alpha,\le M})\ket{\varphi}\|\!\le \! C^2_f\,q(M)e^{-M^\kappa}\left\|e^{-\NA^{\kappa}}A^{\alpha}\right\|\|e^{2\NA^{\kappa}}\ket{\varphi}\!\|.
\end{align}
To control the second term in \eqref{eq:twotermstoboundTrunc}, we note that for $\ket{\phi}\in D(e^{2\NA^{\kappa}})$ we have
  \begin{align*}
    \left\|e^{\NA^{\kappa}}L^{\alpha,\le M}     \ket{\phi}\right\| &\le  \int     |f(s)|e^{\frac{r|s|^{2\kappa}}{2}}     \left\|e^{\NA^{\kappa}}e^{isH}(A^{\alpha})^{\le     M}e^{-2\NA^{\kappa}}\right\|ds\,\left\|e^{2\NA^{\kappa}}\ket{    \phi}\right\| \\ &\le \int |f(s)| e^{r|s|^{2\kappa}} ds     \,\left(\left\|e^{\NA^{\kappa}}A^{\alpha}e^{-2\NA^{\kappa}}\right    \|+Q\right)\left\|e^{2\NA^{\kappa}}\ket{\phi}\right\|\\&    = C_f  \left(\left\|e^{\NA^{\kappa}}A^{\alpha}e^{-2\NA^{\kappa}}    \right\|+Q\right)\left\|e^{2\NA^{\kappa}}\ket{\phi}\right\|,
\end{align*} 
    From this 
and \eqref{eq:NLLeqTrunc} we bound the second term in \eqref{eq:twotermstoboundTrunc} as 
  \begin{align*}
    \|((L^\alpha)^\dagger-\left(L^{\alpha,\le     M}\right)^\dagger)L^{\alpha,\le M}\ket{\varphi}\| \le C^2_f\,     q(M)e^{-M^{\kappa}} \left(\left\|e^{\NA^{\kappa}}A^{\alpha}e^{-2    \NA^{\kappa}}\right\|+Q\right)\left\|e^{2\NA^{\kappa}}\ket{    \varphi}\right\|,
\end{align*} 
which finishes the proof.

\end{proof}

Next, we show that the generators $\mathcal{L}^{\le M}_{\sigma_E,\widehat{f},H}$ and $\mathcal{L}_{\sigma_E,\widehat{f},H}$ are close when evaluated on states $\rho$ satisfying the superpolynomial energy constraint 
\begin{align*}
    \Tr\left(e^{4\NA^{\kappa}}\rho\right) <\infty.
\end{align*}

\begin{proposition}\label{prop:generatorboundTrunc}
Let $f\in
L^1(\mathbb R)$ and set of bare jumps $\{A^\alpha\}_{\alpha\in\cA}$ and assume  \Cref{eq:condAalphas} and, for some self-adjoint and positive semidefinite $\NA$ and $\kappa\in(0,1/2),$  \eqref{eq:AbstractFiniteTruncBareJumps}  are satisfied and further that $e^{\NA^\kappa}A^{\alpha}e^{-2\NA^{\kappa}},$ $e^{-\NA^{\kappa}}A^{\alpha}$ and $A^{\alpha}e^{-\NA^{\kappa}}$ are bounded. Furthermore, assume that \eqref{eq:ExpBound} holds for $r\ge 0$ and $\kappa$ as above and that
\begin{align*}
    C_f := \int_{-\infty}^{\infty} e^{r|t|^{2\kappa}}  |f(t)|dt <     \infty.
\end{align*}
Then for all $\beta,\,\sigma_E>0,$ $M\in\N$ and states $\rho$ satisfying $E_k:=\Tr(e^{k\NA^{\kappa}}\rho)<\infty$, $k\in\{2,4\}$, we have
\begin{align}
&\nn\left\|(\mathcal{L}_{\sigma_E,\widehat{f},H}-\mathcal{L}^{\le M}_{\sigma_E,\widehat{f},H})(\rho)\right\|_1\\
&\quad\le C_{\beta,r,\kappa}C^2_f\, C_A \,|\cA| \,q(M)e^{-M^{\kappa}}\left(\sigma_E
\exp\!\left(\tfrac{\pi^2}{8\beta^2\sigma_E^2}\right)\sqrt{E_4}+\exp\left(r+\tfrac{r^2}{8\sigma^2_E}\right)\,E_2\right),
\end{align}
for some constant $C_{\beta,r,\kappa}\ge 0$ and with $C_A$ being defined in Lemma~\ref{lem:integralrepapproxTrunc}.
\end{proposition}

\begin{proof}

We make use of the integral representation of Proposition \ref{integralrep}. First, we denote the operator 
\begin{align*}
G^{\le M}_{\sigma_E}\!:=\!-\sum_{\alpha\in\cA}\!\int_{-\infty}^\infty\! \!\!\!\!g(t)e^{itH} (L^{\alpha,\le M})^\dagger L^{\alpha,\le M}e^{-itH}\!dt.
\end{align*}
Denoting $\Delta (L^{\alpha})^2:=(L^{\alpha,\le M})^\dagger L^{\alpha,\le M}-(L^{\alpha})^\dagger L^{\alpha},$ we directly obtain from Lemma \ref{lem:integralrepapproxTrunc} that for any state $\rho$ for which  $\Tr(e^{4\NA^{\kappa}}\rho)\le E_4<\infty$ we have 
  \begin{align*}
\|(G^{\le M}_{\sigma_E}-G_{\sigma_E})\rho\|_1
  &\le \sum_{\alpha\in\cA} \int |g(t)|\, \|\Delta   (L^{\alpha})^2 e^{-itH}\rho\|_1\\
&\le \sum_{\alpha\in\cA}\|\Delta   (L^{\alpha})^2 e^{-2\NA^{\kappa}}\|\int   |g(t)|\left\|e^{2\NA^{\kappa}}e^{-itH}\rho\right\|_1dt\\
&\le C^2_fC_A\,|\cA| q(M)e^{-M^{\kappa}}\sqrt{E_4}\int |g(t)|e^{r|t|^{2\kappa}/2}dt\\& \le 
C_{\beta,r,\kappa}C^2_fC_A\,|\cA|\,\sigma_E
\exp\!\left(\tfrac{\pi^2}{8\beta^2\sigma_E^2}\right)q(M)e^{-M^{\kappa}}\sqrt{E_4},
\end{align*} 
where in the last inequality we have used\footnote{Recall the definition $g(t)=\frac{1}{2\pi}\int_{-\infty}^\infty \frac{e^{-\nu^2/8\sigma_E^2}}{1+e^{\beta\nu/2}} e^{-i\nu t}d\nu$, and note $\tfrac{e^{-\nu^2/8\sigma_E^2}}{1+e^{\beta\nu/2}}$ is analytic for $|\Im{\nu}| <\frac{2\pi}{\beta}$. Hence, by shifting the Contour of integration in the definition of $g(t)$ to $\nu\mapsto \nu -i\tfrac{\pi}{\beta}$ for $t>0$ and to $\nu\mapsto \nu +i\tfrac{\pi}{\beta}$ for $t<0$ a straigtforward calculation yields   $|g(t)|
\le
\sigma_E\sqrt{\frac{2}{\pi}}
\exp\!\left(\frac{\pi^2}{8\beta^2\sigma_E^2}\right)
e^{-\tfrac{\pi}{\beta}|t|}.$  From this and $\kappa\in(0,1/2),$ another straightforward calculation yields \eqref{eq:gkappaIntegralBound}. Note that   $C_{\beta,r,\kappa}<\infty$  for   every $\beta>0$  and   $r\ge0$. }
\begin{align}
\label{eq:gkappaIntegralBound}
\int_{\mathbb{R}} |g(t)|\, e^{r |t|^{2\kappa}/2}\, dt
\;\le C_{\beta,r,\kappa} \,
\sigma_E\,\,
\exp\!\left(\frac{\pi^2}{8\beta^2\sigma_E^2}\right),\,
\end{align}
for some finite and $\beta,r$ and $\kappa$ dependent constant $C_{\beta,r,\kappa}\ge 0.$
Similarly, we obtain that 
\begin{align*}
    &\|\rho(G^{\dagger}_{\sigma_E}-\left(G^{\le     M}_{\sigma_E}\right)^\dagger)\|_1\le     C_{\beta,r,\kappa}C^2_fC_A\,|\cA|\,\sigma_E
\exp\!\left(\tfrac{\pi^2}{8\beta^2\sigma_E^2}\right)q(M)e^{-M^{\kappa}}\sqrt{E_4}.\end{align*}
Next, for 
\begin{align}
 X^\alpha_s:=e^{isH}L^{\alpha} e^{-isH}
,\quad  X^{\alpha,\le M}_s:=e^{isH}L^{\alpha,\le M} e^{-isH}\label{eq:XsaiTrunc}
\end{align}
we get for any $s\in\mathbb{R}$, using an analogous argument to \eqref{eq:LiboundTrunc} and the fact that $|t+s|^{2\kappa}\le |t|^{2\kappa}+|s|^{2\kappa}$ that 
\begin{align*}
\|X^{\alpha}_s\sqrt{\rho}\|_2 &\le C_A\sqrt{E_2} \int |f(t)| e^{r|t+s|^{2\kappa}/2}dt \le C_A C_f\sqrt{E_2} \,e^{r|s|^{2\kappa}/2}\\
    \|X^{\alpha,\le M}_s\sqrt{\rho}\|_2 &\le C_A\sqrt{E_2} \int     |f(t)| e^{r|t+s|^{2\kappa}/2}dt\le     C_AC_f\sqrt{E_2}\,e^{r|s|^{2\kappa}/2}
\end{align*}
and furthermore, by an analogous argument to \eqref{eq:NLLeqTrunc} that 
  \begin{align*}
 &\left\|(X^{\alpha}_s-X^{\alpha,\le  M}_s)\sqrt{\rho}\right\|_2\le C_f\,q(M)e^{-M^{\kappa}} \sqrt{E_2}  \,e^{r|s|^{2\kappa}/2}
\end{align*} 
which in total gives
  \begin{align*}
&\left\|X^{\alpha}_s\rho (X^{\alpha})_s^\dagger-X^{\alpha,\le M}_s\rho(X^{\alpha,\le M}_s)^\dagger\right\|_1
\\&\qquad\le \|(X^{\alpha}_s-X^{\alpha,\le M}_s)\sqrt{\rho}\|_2\|\sqrt{\rho}(X^{\alpha}_s)^\dagger\|_2+ \|X^{\alpha,\le M}_s\sqrt{\rho}\|_2\|\sqrt{\rho}((X^{\alpha}_s)^\dagger-(X^{\alpha,\le M}_{s})^\dagger)\|_2
\\&\qquad\le 2C_AC^2_f\,E_2 e^{r|s|^{2\kappa}}q(M)e^{-M^{\kappa}}\,
\end{align*} 

Therefore, denoting the map
\begin{align}\label{phialphasigmaEmapTrunc}
\Phi^{\le M}_{\sigma_E,\widehat f,H}(\rho):=\sigma_E\sqrt{\frac{2}{\pi}}\sum_{\alpha\in\cA}\int_{\mathbb{R}}e^{-2\sigma_E^2 s^2}X^{\alpha,\le M}_s \rho (X^{\alpha,\le M}_s)^\dagger ds,
\end{align}
 we get that
  \begin{align*}
\|(\Phi^{\le M}_{\sigma_E,\widehat f,H}-\Phi_{\sigma_E,\widehat{f},H})(\rho)\|_1\le 4 |\cA| C_AC^2_f\,E_2 q(M)e^{-M^{\kappa}}\exp\left(r+\tfrac{r^2}{8\sigma^2_E}\right), 
\end{align*} 
where we used that  
\begin{align*}
  \sqrt{\frac{2}{\pi}}\,\sigma_E\int e^{-2\sigma^2_Es^2}   e^{r|s|^{2\kappa}}ds\le e^{r} \sqrt{\frac{2}{\pi}}\,\sigma_E\int   e^{-2\sigma^2_Es^2} e^{r|s|}ds\le   2\exp\left(r+\tfrac{r^2}{8\sigma^2_E}\right).
\end{align*}

\end{proof}
To end this section, we show that the mean field Bose Hubbard Hamiltonian, \[H_{\text{mfBH}}= -\mu N +\tfrac{U}{2} \tfrac{N(N-1)}{2} -\overline{\psi} a - \psi a^\dagger  + \vert \psi \vert^2,\] satisfies the assumption \eqref{eq:ExpBound}. More generally, we consider one-mode Hamiltonians of the form
\begin{align*}
    H= h(N) + \overline{\psi} a + \psi a^\dagger
\end{align*}
for some function $h: \N_0 \to \R$ and $\psi\in\C.$  In the following Lemma we find that for $\kappa \in[0,1/2]$ and $k\ge 0$ these Hamiltonians are energy limited with respect to  $e^{kN^{\kappa}}.$  
\begin{lemma}
\label{lem:mfBHEnergyLimited}
Let $\psi \in\C$ and   let $H = h(N)+ \overline{\psi} a + \psi a^\dagger$ have a self-adjoint realization for which the finite Fock span is an operator core.
 Then for $0\le \kappa\le 1/2$ and $k\ge 0$ and $t\in\R$ we have
\begin{align}
    e^{itH} e^{kN^{\kappa}}e^{-itH} \le e^{r|t|^{2\kappa}}      e^{kN^\kappa}
\end{align}
for some $r\ge 0$ depending only on $k,\kappa$ and $\psi.$
\end{lemma}
\begin{proof}
 The cases $k=0$ and $\kappa=0$ are immediate. The following form estimates are justified by finite Fock compressions and closure.
 We consider $G(t):=e^{itH} e^{k(N+1+t^2)^\kappa}e^{-itH}$ and differentiate 
\begin{align}
\label{eq:GoperatorDeriv}
    \frac{d}{dt}G(t) &=2k\kappa te^{itH} (N+1+t^2)^{\kappa-1}     e^{k(N+1+t^2)^\kappa}e^{-itH} + ie^{itH}     \left[H,e^{k(N+1+t^2)^\kappa}\right]e^{-itH}.
\end{align}
In the following we estimate both operators appearing on the right hand side respectively. For the first term we use $\kappa -1\le 0$ which gives
\begin{align}
\label{eq:FirstTermBound}
   \nn 2k\kappa|t|\langle \psi,e^{itH} (N+1+t^2)^{\kappa-1}    e^{k(N+1+t^2)^\kappa}e^{-itH}\psi\rangle &\le 2k |t|^{2\kappa    -1}\langle \psi,e^{itH}  e^{k(N+1+t^2)^\kappa}e^{-itH}\psi\rangle   \nn\\&\le 2k |t|^{2\kappa -1}\langle \psi,G(t)\psi\rangle
\end{align}
For the second term we use
\begin{align}
\label{eq:CommBound}\left[H,e^{k(N+1+t^2)^\kappa}\right] &= \overline{\psi} \left(e^{k(N+2+t^2)^\kappa} -e^{k(N+1+t^2)^\kappa} \right)a + \psi \left(e^{k(N+t^2)^\kappa} -e^{k(N+1+t^2)^\kappa} \right)a^\dagger
\end{align}
and bound each term individually. For the first term we expand $\ket{\varphi} = \sum_{n=0}^{\infty} \varphi_n\ket{n}$ and note
\begin{align*}
    &|\bra{\varphi} \left(e^{k(N+2+t^2)^{\kappa}}     -e^{k(N+1+t^2)^{\kappa}} \right)a\ket{\varphi}| \le      \sum_{n=1}^{\infty} |\varphi_{n}||\varphi_{n-1}|     \sqrt{n}\left|e^{k(n+1+t^2)^{\kappa}} -e^{k(n+t^2)^{\kappa}}     \right| \\ &\le k\kappa\sum_{n=1}^{\infty}     |\varphi_{n}||\varphi_{n-1}| \sqrt{n}     (n+t^2)^{\kappa-1}e^{k(n+1+t^2)^{\kappa}} \le k\kappa     |t|^{2\kappa -1} \sum_{n=1}^{\infty}     |\varphi_{n}||\varphi_{n-1}| e^{k(n+1+t^2)^{\kappa}}
    \\&\le C |t|^{2\kappa-1}\bra{\varphi}     e^{k(N+1+t^2)^\kappa}\ket{\varphi}, 
\end{align*}
 for some $C\ge 0$ which depends on $k$ and $\kappa$ but is  independent of all other variables and where we have used in the  second inequality we have used the mean value theorem, in the third  the fact that $\kappa \le 1/2$ and in the fourth the Cauchy-Schwarz  inequality. The second term in \eqref{eq:CommBound} can be treated  similarly giving the bound
 \begin{align*}
     |\bra{\varphi} \left(e^{k(N+1+t^2)^{\kappa}}      -e^{k(N+t^2)^{\kappa}} \right)a^\dagger\ket{\varphi}| \le      C|t|^{2\kappa-1} \bra{\varphi}      e^{k(N+1+t^2)^\kappa}\ket{\varphi}.
 \end{align*}
 Combining these estimates with \eqref{eq:GoperatorDeriv} and  \eqref{eq:FirstTermBound} we get
   \begin{align*}
     -C'|t|^{2\kappa-1}G(t)\le\frac{d}{dt} G(t) \le      C'|t|^{2\kappa-1} G(t) 
 \end{align*} 
 for some constant $C'\ge 0$ depending only on $k,\kappa$ and  $\psi.$ Using Gr\"onwall's inequality in the case $t\ge0$ on the  interval $[0,t]$ or for the case $t\le 0$ on the interval $[0,-t]$  proves 
   \begin{align*}
     G(t) \le e^{r |t|^{2\kappa}} G(0) = e^{r |t|^{2\kappa}}      e^{k(N+1)^\kappa}
 \end{align*} 
 for some $r\ge 0$ depending again only on $k,\kappa$ and $\psi.$  Using $e^{kN^{\kappa}} \le e^{k(N+1+t^2)^{\kappa}}$   and  $(N+1)^\kappa\le N^\kappa+1$ gives $|t|\ge1$; the same commutator  estimate with $e^{kN^\kappa}$ gives $|t|\le1$ .
\end{proof}

\subsubsection{Truncating the Hamiltonian}
\label{sec:FinitDimGenerator}

In the following we provide a   bounded Lindbladian with a finite-dimensional system-register restriction , denoted by $\mathcal{L}^{\le M}_{\sigma_E,\widehat{f},H_{\le M}},$ which is good approximation of the generator $\mathcal{L}_{\sigma_E,\widehat{f},H}$ studied in Section~\ref{monotonicitygap} and \ref{sec:IntegralRepresentation}. For that we consider the finite-dimensional approximations, $\{(A^{\alpha})^{\le M}\}_{\alpha\in\cA},$ of the bare jumps $\{A^\alpha\}_{\alpha\in\cA}$ analysed in Section~\ref{sec:FiniteTruncBareJumps} and replace the Hamiltonian $H$ within the corresponding generator $\mathcal{L}^{\le M}_{\sigma_E,\widehat{f},H}$ by its finite-dimensional truncation, $H_{\le M},$ to obtain $\mathcal{L}^{\le M}_{\sigma_E,\widehat{f},H_{\le M}}.$ As in Section~\ref{sec:FiniteTruncBareJumps} we have already seen that $\mathcal{L}^{\le M}_{\sigma_E,\widehat{f},H}$ is a good approximation of $\mathcal{L}_{\sigma_E,\widehat{f},H}$ on certain energy constraint input states, we focus in this section on the approximation of $\mathcal{L}^{\le M}_{\sigma_E,\widehat{f},H}$ by the fully finite-dimensional\footnote{  On the original Hilbert space , $\mathcal{L}^{\le M}_{\sigma_E,\widehat{f},H_{\le M}}$ is a   bounded  generator  , generally not a finite-rank superoperator .   Its restriction to the reducing system register is finite dimensional .} generator $\mathcal{L}^{\le M}_{\sigma_E,\widehat{f},H_{\le M}}.$

For that, we start with defining the truncated, finite rank Hamiltonian $H_{\le M}$ for some self-adjoint, possibly unbounded, Hamiltonian $H$ on some separable Hilbert space $\cH:$ More precisely we consider a parametrised family of self-adjoint and finite rank projections  $\left(P_M\right)_{M\in\N}$ on $\cH$ such that $P_M\cH \subseteq D(H)$ and with that the finite rank truncation 
\begin{align}\label{HleMpolybound}
    H_{\le M} := P_MHP_M\qquad \text{ with }\qquad  \|H_{\le M}\|\le     p_5(|\mathcal{A}|,M)
\end{align}
for some polynomial $p_5$ of the number, $|\mathcal{A}|,$ of jumps and truncation $M$. As described in the Introduction, we refer to the finite dimensional space $\operatorname{im}(P_M)$ as the \emph{system register}. In many-body or multi-mode systems, one usually considers system registers whose dimensions satisfies
$
\log\left( \dim\bigl(\operatorname{im}(P_M)\bigr)\right)
= \mathcal{O}\bigl(|\mathcal{A}|\log(M)\bigr),
$ but we are in principle free to leave the dimension of the system register unspecified at this stage. 

We assume in the following that $H$ is well-approximated by $H_{\le M}$ on energy constraint input states; specifically, we have for some $\kappa\in(0,1/2)$
\begin{align}
 \label{eq:HamiltonianLeakTrunc}  
 \left\|\left(H-H_{\le M}\right)e^{-\NA^\kappa}\right\| \le  p(|\cA|,M) e^{-M^\kappa},
\end{align}
where $p(|\cA|,M)$ is some polynomially bounded function and $\NA$ is some self-adjoint and positive semidefinite operator. 

Let us illustrate the above with the example of a $m$-mode bosonic system on the Hilbert space $\left(L^2(\R)\right)^{\otimes m},$ with multi-index $\alpha$ taken from the set $\cA=\{(i,+),(i,-)\}_{i=1}^m$ and bare jumps $\{A^\alpha\}_{\alpha\in \cA}=\left\{a^k_i,(a_i^k)^\dagger\right\}_{i=1}^m$ for some $k\in\N:$
In this case, we take $P_M = \pi^{\otimes m}_M$ with $\pi_M = \sum_{n=0}^M\kb{n}$ being the local truncations onto the first $M+1$ Fock states of a fixed mode. Furthermore, for this multimode case, we consider $\NA\equiv\Ntot:=\sum_{i=1}^m a^\dagger_ia_i$ such that \eqref{eq:HamiltonianLeakTrunc} is naturally satisfied for Hamiltonians $H$ being polynomials in $a_i$ and $a^\dagger_i$ with degree $\mathcal{O}(1)$ in the number of modes $m$, e.g., the Bose-Hubbard Hamiltonian. 

Additionally to \eqref{eq:HamiltonianLeakTrunc}, we assume as in Section~\ref{sec:FiniteTruncBareJumps} that $H$ satisfies \eqref{eq:ExpBound}, i.e.~that for $k=2,4$ and $\kappa\in(0,1/2)$ as above there exists $r\ge 0$ such that for all $t\in\R$ we have\footnote{The reason for the choice of range of $\kappa,$ in particular the constraint $\kappa<1/2$, is that this enables us to prove approximation of $\mathcal{L}^{\le M}_{\sigma_E,\widehat{f},H}$ by $\mathcal{L}^{\le M}_{\sigma_E,\widehat{f},H_{\le M}}$ for all $\beta>0.$ The reason for this is that the function $e^{r|t|^{2\kappa}/2}$ needs to be dominated by the function $g(t),$ which is featured in the integral representations of the generators stated in Proposition~\ref{integralrep}, and it can be shown that   $|g(t)|
\lesssim
e^{-\tfrac{\pi}{\beta}|t|}$  for large $t.$ }
\begin{align}\label{eq:ExpBound2}
e^{-itH}e^{k\NA^\kappa} e^{itH}\le e^{r|t|^{2\kappa}}\,e^{k\NA^\kappa}.
\end{align}

Under both of these assumptions, we see in the following lemma that the unitary evolutions of the unbounded and truncated Hamiltonians are close on energy constraint inputs:

Following the notation of Section~\ref{sec:FiniteTruncBareJumps}, we show in the following that under the assumptions \eqref{eq:ExpBound2} and \eqref{eq:HamiltonianLeakTrunc} the   bounded generator with finite-dimensional register restriction , $\mathcal{L}^{\le M}_{\sigma_E,\widehat{f},H_{\le M}},$ is close to $\mathcal{L}^{\le M}_{\sigma_E,\widehat{f},H}.$
For this, we first prove the following technical Lemma. In what follows, we denote for some function $f\in
L^1(\mathbb R)$ the integrated jump operators 
\begin{align}
\label{eq:TruncatedL}
L^{\alpha,\le M}:=\int f(s)e^{isH}(A^{\alpha})^{\le M}e^{-isH}ds\quad\text{and}\quad
L^{\alpha,\le M}_{\le M}:=\int f(s)e^{isH_{\le M}}(A^{\alpha})^{\le M} e^{-isH_{\le M}}ds.
\end{align}
By construction, we have that
\begin{align}\label{eqLalphaMMbound}
\|L^{\alpha,\le M}_{\le M}\|\le \|f\|_{L^1(\mathbb{R})}\|(A^\alpha)^{\le M}\|\le \|f\|_{L^1(\mathbb{R})}\, q'(M),
\end{align}
where the polynomial $q'(M)$ is defined in \eqref{eq:NormTruncatedJump}.

\begin{lemma}\label{lem:integralrepapproxFullyFinite}
Let $f\in
L^1(\mathbb R)$ and assume that \Cref{eq:condAalphas}, \eqref{eq:NormTruncatedJump} and \eqref{eq:NormTruncatedJump2} are satisfied for the set of bare jumps $\{A^\alpha\}_{\alpha\in\cA}$. Furthermore, for $\NA$ self-adjoint and positive semidefinite assume that \eqref{eq:HamiltonianLeakTrunc} and  \eqref{eq:ExpBound2} holds for some $\kappa\in(0,1/2)$ and $r\ge0$ and that
\begin{align*}
    C'_f := \int_{-\infty}^{\infty} |t|e^{r|t|^{2\kappa}} |f(t)|dt <     \infty.
\end{align*}

Then for $M\in \N$ we have
\begin{align*}
\|(L^{\alpha,\le M}\!-\!L^{\alpha,\le M}_{\le M})e^{-\NA^{\kappa}}\|\!&\le\! C'_fp_1(|\cA|,M)e^{-M^{\kappa}},\\
\|((L^{\alpha,\le M})^\dagger L^{\alpha,\le M}-\left(L^{\alpha,\le M}_{\le M}\right)^\dagger L^{\alpha,\le M}_{\le M})e^{-\NA^{\kappa}}\|&\le C_fC'_fp_2(|\cA|,M)e^{-M^{\kappa}},\, 
\end{align*}
where $p_1(|\cA|,M)=2q'(M)p(|\cA|,M)$ and $p_2(|\cA|,M) =4(q'(M))^2p(|\cA|,M)$ with $q'(M)$ being the polynomially bounded function appearing in \eqref{eq:NormTruncatedJump} and \eqref{eq:NormTruncatedJump2} and $C_f$ being defined in \eqref{eq:C_f}.

\end{lemma}

\begin{proof}
 \begingroup
We first record the Hamiltonian estimate needed here and below.
Duhamel's formula on the weighted domain, justified first on the
operator domain and then by the assumed weighted bounds, together with
\mbox{%
\eqref{eq:HamiltonianLeakTrunc} }\hskip0pt%
and \mbox{%
\eqref{eq:ExpBound2}}\hskip0pt%
, gives
\begin{equation}\label{eq:HamEvolTrunc}
\begin{aligned}
\left\|(e^{-itH}-e^{-itH_{\le M}})e^{-\NA^\kappa}\right\|
&\le\int_0^{|t|} e^{rs^{2\kappa}/2}\,ds\,
\left\|(H-H_{\le M})e^{-\NA^\kappa}\right\|\\
&\le |t|e^{r|t|^{2\kappa}/2}p(|\cA|,M)e^{-M^\kappa}.
\end{aligned}
\end{equation}
\par\endgroup

For the first inequality in Lemma~\ref{lem:integralrepapproxFullyFinite} we use \eqref{eq:ExpBound2} and   \eqref{eq:HamEvolTrunc}  to see for $\ket{\psi}\in D\left(e^{\NA^{\kappa}}\right)$ that we have
\begin{align}
\label{eq:NLLeqFullyFinite}
\nn&\|(L^{\alpha,\le M}-L^{\alpha,\le M}_{\le M})\ket{\psi}\|
\\&\quad\le \int |f(s)|\Big(\left\|(A^{\alpha})^{\le M}\left(e^{-isH} -e^{-isH_{\le M}}\right)\ket{\psi}\right\|\nn+ \left\|\left(e^{isH} -e^{isH_{\le M}}\right)(A^{\alpha})^{\le M} e^{-isH}\ket{\psi}\right\|\Big)ds
\\&\quad\le 2q'(M)p(|\cA|,M)e^{-M^{\kappa}}
\int|f(s)| |s|e^{r|s|^{2\kappa}}ds\|e^{\NA^{\kappa}}\ket{\psi}\| = 2C'_fq'(M)p(|\cA|,M)e^{-M^{\kappa}}
\|e^{\NA^{\kappa}}\ket{\psi}\|,
\end{align}
where we used \eqref{eq:NormTruncatedJump} and \eqref{eq:NormTruncatedJump2} in the second inequality.

Next, we control squares of the $L$ operators for $\ket{\varphi}\in D(e^{\NA^{\kappa}})$:
\begin{align}\label{eq:twotermstoboundTrunc2}
&\nn\left\|\left(\left(L^{\alpha,\le M}\right)^\dagger L^{\alpha,\le M}-\left(L^{\alpha,\le M}_{\le M}\right)^\dagger L^{\alpha,\le M}_{\le M}\right)\ket{\varphi}\right\|\\&\le \left\|\left(L^{\alpha,\le M}_{\le M}\right)^\dagger (L^{\alpha,\le M}-L^{\alpha,\le M}_{\le M})\ket{\varphi}\right\|+\left\|\left(\left(L^{\alpha,\le M}\right)^\dagger-\left(L^{\alpha,\le M}_{\le M}\right)^\dagger\right)L^{\alpha,\le M}\ket{\varphi}\right\|.
\end{align}
The first term above can be controlled by using that by \eqref{eq:NormTruncatedJump} we have $\left\|\left(L_{\le M}^{\alpha,\le M}\right)^\dagger\right\| \le q'(M)\|f\|_1\le q'(M)C_f$ and therefore using \eqref{eq:NLLeqFullyFinite} we have
\begin{align*}&\left\|\left(L^{\alpha,\le M}_{\le M}\right)^\dagger (L^{\alpha,\le M}-L^{\alpha,\le M}_{\le M})\ket{\varphi}\right\| \le 2C_fC'_f (q'(M))^2p(|\cA|,M)e^{-M^{\kappa}}
\|e^{\NA^{\kappa}}\ket{\varphi}\|.
\end{align*}
To control the second term in \eqref{eq:twotermstoboundTrunc2}, we note that by \eqref{eq:ExpBound2} and \eqref{eq:NormTruncatedJump2} we have 
\begin{align}
\label{eq:ThereAreLoadsOfInequalitiesHere}
\nn&\left\|e^{\NA^{\kappa}}L^{\alpha,\le M} e^{-\NA^{\kappa}}\right\| \le  \int |f(s)|e^{r|s|^{2\kappa}/2} \left\|e^{\NA^{\kappa}}e^{isH}(A^\alpha)^{\le M}e^{-\NA^{\kappa}}\right\|ds\,\\ &\le q'(M)\int |f(s)| e^{r|s|^{2\kappa}} ds = q'(M)C_f,
\end{align}
where $C_f$ is defined in \eqref{eq:C_f}. From this 
and the adjoint version of \eqref{eq:NLLeqFullyFinite} we bound the second term in \eqref{eq:twotermstoboundTrunc2} as 
\begin{align*}
    &\left\|\left(\left(L^{\alpha,\le     M}\right)^\dagger-\left(L^{\alpha,\le M}_{\le     M}\right)^\dagger\right)L^{\alpha,\le M}\ket{\varphi}\right\|     \le 2 C_fC'_f (q'(M))^2p(|\cA|,M)e^{-M^{\kappa}}     \left\|e^{\NA^{\kappa}}\ket{\varphi}\right\|,
\end{align*}
which finishes the proof.
\end{proof}
\noindent  We are now ready to show in the following proposition that $\mathcal{L}^{\le M}_{\sigma_E,\widehat{f},H}$ can be well approximated by the bounded generator $\mathcal{L}^{\le M}_{\sigma_E,\widehat{f},H_{\le M}}$, whose system-register restriction is finite dimensional.
\begin{proposition}\label{prop:generatorboungFullyFinite}
Let $f\in
L^1(\mathbb R)$ and assume that \Cref{eq:condAalphas}, \eqref{eq:NormTruncatedJump} and \eqref{eq:NormTruncatedJump2} are satisfied for the set of bare jumps $\{A^\alpha\}_{\alpha\in\cA}$. Furthermore, for $\NA$ self-adjoint and positive semidefinite assume that \eqref{eq:HamiltonianLeakTrunc} and  \eqref{eq:ExpBound2} holds for some $\kappa\in(0,1/2)$ and $r\ge0$ and that
\begin{align*}
    C'_f := \int_{-\infty}^{\infty} |t|e^{r|t|^{2\kappa}} |f(t)|dt <     \infty.
\end{align*}
Then, for all $\beta,\,\sigma_E>0,$ $M\in\N$ and state $\rho$ satisfying $E_2:=\Tr(e^{2\NA^{\kappa}}\rho)<\infty$, we have
\begin{align}
&\nn\left\|(\mathcal{L}^{\le M}_{\sigma_E,\widehat{f},H}-\mathcal{L}^{\le M}_{\sigma_E,\widehat{f},H_{\le M}})(\rho)\right\|_1\\
&\quad\le C_{\beta,r,\kappa} \widetilde C_f \sqrt{E_2}\,p_3(|\cA|,M)e^{-M^{\kappa}}\nn
\nn\left(\sigma_E
\exp\!\left(\tfrac{\pi^2}{8\beta^2\sigma_E^2}\right)+\exp\left(r+\tfrac{(r+1)^2}{8\sigma^2_E}\right)\right),
\end{align}
for some constant $C_{\beta,r,\kappa}\ge 0$ and 
where $p_3(|\cA|,M):=(q'(M))^2|\cA|\,p(|\cA|,M)$ is  polynomially bounded, $\widetilde C_f := C_f\max\{C_f,C'_f\}$ and $C_f$ is defined in \eqref{eq:C_f}.
\end{proposition}
\begin{proof}
We make use of the integral representation of Proposition \ref{integralrep}. First, we denote the operators
\begin{align*}
G^{\le M}_{\sigma_E}\!&:=\!-\sum_{\alpha\in\cA}\!\int_{-\infty}^\infty\! \!\!\!\!g(t)e^{itH} (L^{\alpha,\le M})^\dagger L^{\alpha,\le M}e^{-itH}\!dt\\G^{\le M}_{\sigma_E,\le M}\!&:=\!-\sum_{\alpha\in\cA}\!\int_{-\infty}^\infty\! \!\!\!\!g(t)e^{itH_{\le M}} (L_{\le M}^{\alpha,\le M})^\dagger L_{\le M}^{\alpha,\le M}e^{-itH_{\le M}}\!dt,
\end{align*}
    where $G^{\le M}_{\sigma_E}$ appears in the definition of     $\mathcal{L}^{\le M}_{\sigma_E,\widehat{f},H}$ and $G^{\le     M}_{\sigma_E,\le M}$ appears in the definition of     $\mathcal{L}^{\le M}_{\sigma_E,\widehat{f},H_{\le M}}.$

We see
\begin{align}
\label{eq:MoreGApproxiii}
\nn\left\|\left(G^{\le M}_{\sigma_E}-G^{\le M}_{\sigma_E,\le M}\right)\rho\right\|_1 &\le \sum_{\alpha\in\cA }\int |g(t)| \left\|\left(e^{itH}-e^{itH_{\le M}}\right) (L^{\alpha,\le M})^\dagger L^{\alpha,\le M} e^{-itH}dt\,\rho \right\|_1\\&\quad\nn+\sum_{\alpha\in\cA }\int |g(t)| \left\|\left((L^{\alpha,\le M})^\dagger L^{\alpha,\le M} -(L^{\alpha,\le M}_{\le M})^\dagger L^{\alpha,\le M}_{\le M}\right)e^{-itH}dt\,\rho \right\|_1\\&\quad+\sum_{\alpha\in\cA }\int |g(t)| \left\|(L^{\alpha,\le M}_{\le M})^\dagger L^{\alpha,\le M}_{\le M}\left(e^{-itH}-e^{-itH_{\le M}}\right)dt\,\rho \right\|_1.
\end{align}
For the first term we note using \eqref{eq:ExpBound2} and \eqref{eq:NormTruncatedJump2} 
 together with an analogous argument as for  \eqref{eq:ThereAreLoadsOfInequalitiesHere} that 
\begin{align*}
    \left\|e^{\NA^{\kappa}}(L^{\alpha,\le M})^\dagger L^{\alpha,\le     M}e^{-\NA^{\kappa}}\right\| \le      \left\|e^{\NA^{\kappa}}(L^{\alpha,\le M})^\dagger     e^{-\NA^{\kappa}}\right\|\left\|e^{\NA^{\kappa}} L^{\alpha,\le     M}e^{-\NA^{\kappa}}\right\| \le C^2_f (q'(M))^2 
\end{align*}
and therefore using   \eqref{eq:HamEvolTrunc}
   \begin{align*}
&\sum_{\alpha\in\cA }\int |g(t)| \left\|\left(e^{itH}-e^{itH_{\le M}}\right) (L^{\alpha,\le M})^\dagger L^{\alpha,\le M} e^{-itH}dt\,\rho \right\|_1\\&\le C^2_f|\cA|(q'(M))^2\sqrt{E_2} \int |g(t)| e^{r|t|^{2\kappa}/2} \left\|\left(e^{itH}-e^{itH_{\le M}}\right)e^{-\NA^{\kappa}}\right\| dt\\&\le C^2_f|\cA|(q'(M))^2\sqrt{E_2}\,p(|\cA|,M)e^{-M^\kappa}\int |g(t)| |t|e^{r|t|^{2\kappa}}dt\\& \le C_{\beta,r,\kappa}C^2_f\sqrt{E_2} \,
\sigma_E\,\,
\exp\!\left(\frac{\pi^2}{8\beta^2\sigma_E^2}\right) |\cA|(q'(M))^2p(|\cA|,M)e^{-M^{\kappa}}
\end{align*} 
for some constant $C_{\beta,r,\kappa}\ge0$ and
where in the last inequality we used \eqref{eq:gkappaIntegralBound} with a larger value of $r$ in there. The third term in \eqref{eq:MoreGApproxiii} can be estimated analogously and the second as well utilising additionally Lemma~\ref{lem:integralrepapproxFullyFinite} which in total gives
\begin{align*}
    \left\|\left(G^{\le M}_{\sigma_E}-G^{\le M}_{\sigma_E,\le     M}\right)\rho\right\|_1\le C_{\beta,r,\kappa}\widetilde     C_f\sqrt{E_2} \,
\sigma_E\,\,
\exp\!\left(\frac{\pi^2}{8\beta^2\sigma_E^2}\right) |\cA|p_2(|\cA|,M)e^{-M^{\kappa}},
\end{align*}
where we increased the constant $C_{\beta,r,\kappa}$ by a factor 3 compared to the above and, furthermore, denoted $\widetilde C_f := C_f\max\{C_f,C'_f\}.$

Next, for 
\begin{align*}
 X^{\alpha,\le M}_s:=e^{isH}L^{\alpha,\le M} e^{-isH}
,\quad  X^{\alpha,\le M}_{s,\le M}:=e^{isH_{\le M}}L^{\alpha,\le M}_{\le M} e^{-is{H_{\le M}}}
\end{align*}
using that by \eqref{eq:NormTruncatedJump} we have   $\|X^{\alpha,\le M}_s\|,\, \|X^{\alpha,\le M}_{s,\le M}\|\le \|f\|_1 q'(M)\le C_f q'(M),$
we see 
\begin{align*}
    &\left\| X^{\alpha,\le M}_s\rho (X^{\alpha,\le     M}_s)^\dagger- X^{\alpha,\le M}_{s,\le M}\rho (X^{\alpha,\le     M}_{s,\le M})^\dagger\right\|_1 \\
    &\qquad\le C_f\,q'(M)\left( \left\|\left(X^{\alpha,\le     M}_s-  X^{\alpha,\le M}_{s,\le     M}\right)\rho\right\|_1+\left\|\rho\left((X^{\alpha,\le     M}_s)^\dagger-  (X^{\alpha,\le M}_{s,\le     M})^\dagger\right)\right\|_1\right).
\end{align*} 
We focus on the term as the second term can be estimated the same way:
Using \eqref{eq:NormTruncatedJump}
and   \eqref{eq:HamEvolTrunc}  we see
\begin{align*}
    \left\|\left(X^{\alpha,\le M}_s-  X^{\alpha,\le M}_{s,\le     M}\right)\rho\right\|_1 &\le \int     |f(t)|\Big(\left\|\left(e^{i(t+s)H}-e^{i(t+s)H_{\le     M}}\right)(A^{\alpha})^{\le M} e^{-i(t+s)H}\rho\right\|_1     \\&\qquad+\left\|e^{i(t+s)H_{\le M}}(A^{\alpha})^{\le M}     \left(e^{-i(t+s)H}-e^{-i(t+s)H_{\le M}}\right)\rho\right\|_1     \Big)dt\\&\le 2 q'(M)\sqrt{E_2}\,p(|\cA|,M)e^{-M^{\kappa}}\int     |f(t)| |t+s| e^{r|t+s|^{2\kappa}}dt \\&\le2\max\{C_f,C'_f\}     \sqrt{E_2}\, q'(M) p(|\cA|,M)e^{-M^\kappa}     (1+|s|)e^{r|s|^{2\kappa}}.
\end{align*}
Using this and denoting the map
\begin{align}\label{phialphasigmaEmapTruncDouble}
\Phi^{\le M}_{\sigma_E,\widehat f,H_{\le M}}(\rho):=\sigma_E\sqrt{\frac{2}{\pi}}\sum_{\alpha\in\cA}\int_{\mathbb{R}}e^{-2\sigma_E^2 s^2}X^{\alpha,\le M}_{s,\le M} \rho (X^{\alpha,\le M}_{s,\le M})^\dagger ds,
\end{align}
we see that
\begin{align*}
\left\|\left(\Phi^{\le M}_{\sigma_E,\widehat f,H}-\Phi^{\le M}_{\sigma_E,\widehat f,H_{\le M}}\right)(\rho)\right\|_1\lesssim  \,\widetilde C_f\,\sqrt{E_2} |\cA|(q'(M))^2p(|\cA|,M)e^{-M^{\kappa}} \exp\left(r+\tfrac{(r+1)^2}{8\sigma^2_E}\right)  , 
\end{align*}
where we used that  
\begin{align*}
\,\sigma_E\int e^{-2\sigma^2_Es^2}(1+|s|) e^{r|s|^{2\kappa}}ds&\le  e^{r} \,\sigma_E\int e^{-2\sigma^2_Es^2}(1+|s|) e^{r|s|}ds\\&\lesssim e^{r} \,\sigma_E\int e^{-2\sigma^2_Es^2} e^{(1+r)|s|}ds \lesssim \exp\left(r+\tfrac{(r+1)^2}{8\sigma^2_E}\right)   
\end{align*}
and where we only hide constants independent of all parameters with the $\lesssim$-notation.

\end{proof}

\subsubsection[Finite-dimensional Lindblad dynamics for
integrable filters]{Finite-dimensional Lindblad dynamics for
{ integrable} filter functions}
\label{sec:FiniteDimPreparationIntegrable}

In this section we combine the results of the previous two sections, in particular Propositions~\ref{prop:generatorboundTrunc} and~\ref{prop:generatorboungFullyFinite}, to see in Theorem~\ref{thm:IntegrableGibbsDynamics} that the dynamics $e^{t\cL_{\sigma_E,\widehat f,H}}$ 
for integrable filters satisfying the weighted time-moment hypothesis in Theorem~\ref{thm:IntegrableGibbsDynamics} is
well approximated by $e^{t\cL^{\le M}_{\sigma_E,\widehat f,H_{\le M}}}.$ 
Moreover,  under the assumptions of a simple kernel and a positive spectral gap of the corresponding unbounded generator on the Hilbert-Schmidt space,  we then apply Theorem~\ref{thm:IntegrableGibbsDynamics} to show efficient finite-dimensional preparation of the Gibbs state of $H$ in Corollary~\ref{cor:IntegrableGibbs}. 

We first state both results and then give their proofs at the end of the section.

\begin{theorem}[Finite-dimensional approximation of $e^{t\cL_{\sigma_E,\widehat f,H}}$]
\label{thm:IntegrableGibbsDynamics}
Let $f\in
L^1(\mathbb R)$ and set of bare jumps $\{A^\alpha\}_{\alpha\in\cA}$ and assume  \Cref{eq:condAalphas} is satisfied. 
Let further $\NA$ self-adjoint and positive semidefinite and either $0<\kappa<1/2$, $r\ge0$, or $0<\kappa\le1$, $r=0$, and assume that \eqref{eq:NormTruncatedJump}, \eqref{eq:AbstractFiniteTruncBareJumps}, \eqref{eq:NormTruncatedJump2} are satisfied and that $C_A := \max_{\alpha\in\cA}\left(\|e^{\NA^{\kappa}}A^{\alpha}e^{-2\NA^{\kappa}}\|+\|A^{\alpha}e^{-\NA^{\kappa}}\|+\|e^{-\NA^{\kappa}}A^{\alpha}\|\right)+2Q<\infty$, $Q:=\sup_{M\ge1}q(M)e^{-M^\kappa}<\infty$. Furthermore, assume that for $\kappa$ as above \eqref{eq:HamiltonianLeakTrunc} holds and, for this $r$, the condition \eqref{eq:ExpBound} is satisfied and that $\int_{-\infty}^{\infty} |t|e^{r|t|^{2\kappa}} |f(t)|dt <\infty.$ 
For $\beta>0$ we assume that the Gibbs state of $H,$ $\sigma_\beta,$ satisfies
\begin{align}E_{\operatorname{Gibbs}}:=\Tr\left(e^{4\NA^\kappa}\, \sigma_\beta\right)<\infty
\end{align}
and consider state $\rho$ such that 
\begin{align}
\label{eq:InputDmaxConstraint}
    \rho \le \mathfrak{c}\,\sigma_\beta
\end{align}
for some $\mathfrak{c}\ge 1.$ Then for all $\sigma_E\in(0,\infty),$ $t\ge 0$ and $M\in\N$ we have
    \begin{align}
\left\|\left(e^{t\cL_{\sigma_E,\widehat f,H}}-e^{t\cL^{\le M}_{\sigma_E,\widehat f,H_{\le M}}}\right)(\rho)\right\|_1\le t\,C_1(\beta,r,\kappa,\sigma_E,f,C_A)\,\mathfrak{c} \,p_4(|\mathcal{A}|,M)E_{\operatorname{Gibbs}}\, e^{-M^\kappa} , 
    \end{align}
    where $C_1(\beta,r,\kappa,\sigma_E,f,C_A)\ge 0$ is some constant     depending only on the displayed parameters and     $p_4(|\mathcal{A}|,M)$ is some polynomially bounded function.     Therefore, for $\eps>0,$ we can achieve
      \begin{align}\label{eq:Metal}
       &\left\|\left(e^{t\cL_{\sigma_E,\widehat        f,H}}-e^{t\cL^{\le M}_{\sigma_E,\widehat f,H_{\le        M}}}\right)(\rho)\right\|_1\le \eps,\qquad\text{with}\qquad M        = \widetilde{\mathcal{O}}        \left(\left(\log\left(2+\frac{t\,\mathfrak{c}\,E_{       \operatorname{Gibbs}}|\cA|}{        \eps}\right)\right)^{1/\kappa}\right),
    \end{align} 
    where the $\widetilde{\mathcal{O}}$ notation hides constants     independent of the displayed parameters and additionally     suppresses subdominant $\operatorname{poly}\log\log$ factors.
\end{theorem}
\begin{corollary}[Gibbs state preparation]
\label{cor:IntegrableGibbs}
Under the same assumptions as in Theorem~\ref{thm:IntegrableGibbsDynamics} and assuming additionally that  $L_{\sigma_E,\widehat{f},H},$ the self-adjoint generator on $\mathscr{T}_2(\cH)$ associated to the Lindbladian $\cL_{\sigma_E,\widehat f,H},$ has a positive spectral gap 
$\lambda_2 \equiv \operatorname{gap}\left(L_{\sigma_E,\widehat f,H}\right) >0$ and $\ker L_{\sigma_E,\widehat f,H}=\mathbb C\sqrt{\sigma_\beta}$,  
we can achieve for all $0<\eps<1$
\begin{align}
\nn\left\|e^{t\cL^{\le M}_{\sigma_E,\widehat f,H_{\le M}}}(\rho) - \sigma_\beta\right\|_1
&\le \eps,\\
\nn t&=\mathcal{O}\left(\frac{1}{\lambda_2}\log\left(2+\frac{\mathfrak{c}}{\eps}\right)\right),\\
\nn M&=\widetilde{\mathcal{O}}\left(\left(\log\left(2+\frac{(1+\lambda_2^{-1})\mathfrak{c}\,E_{\operatorname{Gibbs}}|\cA|}{\eps}\right)\right)^{1/\kappa}\right),
\end{align}
where the $\mathcal{O}$ and $\widetilde{\mathcal{O}}$ notations hide constants independent of the displayed parameters and $\widetilde{\mathcal{O}}$ additionally suppresses subdominant $\operatorname{poly}\log\log$ factors.
\end{corollary}

\begin{proof}[Proof of Theorem~\ref{thm:IntegrableGibbsDynamics}]
We use that 
\begin{align*}
    \left(e^{t\cL^{\le M}_{\sigma_E,\widehat f,H_{\le M}}}     -e^{t\cL_{\sigma_E,\widehat f,H}}\right)(\rho) = \int_0^t     e^{(t-s)\cL^{\le M}_{\sigma_E,\widehat f,H_{\le     M}}}\left(\cL^{\le M}_{\sigma_E,\widehat f,H_{\le M}}     -\cL_{\sigma_E,\widehat f,H}\right)e^{s\cL_{\sigma_E,\widehat     f,H}}(\rho)ds
\end{align*}
and therefore, denoting $\rho(s) :=e^{s\cL_{\sigma_E,\widehat f,H}}(\rho),$ using Proposition~\ref{prop:generatorboundTrunc} and~\ref{prop:generatorboungFullyFinite} we see
\begin{align*}
    &\left\|\left(e^{t\cL^{\le M}_{\sigma_E,\widehat f,H_{\le M}}}     -e^{t\cL_{\sigma_E,\widehat f,H}}\right)(\rho)\right\|_1 \\
    & \qquad \le   t\sup_{s\in[0,t]}\left\|\left(\cL^{\le     M}_{\sigma_E,\widehat f,H_{\le M}}-\cL_{\sigma_E,\widehat     f,H}\right)(\rho(s))\right\|_1 \\
 &    \qquad\le t\sup_{s\in[0,t]}\left\|\left(\cL^{\le  M}_{\sigma_E,\widehat f,H_{\le M}}-\cL^{\le M}_{\sigma_E,\widehat  f,H}\right)(\rho(s))\right\|_1  +t\sup_{s\in[0,t]}\left\|\left(\cL^{\le M}_{\sigma_E,\widehat  f,H}-\cL_{\sigma_E,\widehat f,H}\right)(\rho(s))\right\|_1\\
    &\qquad \le t\,C_1(\beta,r,\kappa,\sigma_E,f,C_A)     p_4(|\cA|,M)e^{-M^\kappa} \sup_{s\in[0,t]}\max\left\{\sqrt{\Tr    \left(e^{4\NA^\kappa}\rho(s)\right)}\,,\,     \Tr\left(e^{2\NA^\kappa}\rho(s)\right) \right\}\\
   &  \qquad\le t\,\mathfrak{c}\,C_1(\beta,r,\kappa,\sigma_E,f,C_A)    p_4(|\mathcal{A}|,M) e^{-M^\kappa} E_{\operatorname{Gibbs}},
\end{align*}
where $C_1(\beta,r,\kappa,\sigma_E,f,C_A)\ge 0$ is some constant depending only on the displayed parameters, $p_4(|\mathcal{A}|,M)$ is some polynomially bounded function and where we used in the last inequality that by \eqref{eq:InputDmaxConstraint}, positivity of the map $e^{s\cL_{\sigma_E,\widehat f,H}}$ and the fact that $\cL_{\sigma_E,\widehat f,H}$ is KMS-symmetric we have
\begin{align*}
   \rho(s) \le \mathfrak{c}\, e^{s\cL_{\sigma_E,\widehat    f,H}}(\sigma_\beta ) = \mathfrak{c}\,\sigma_\beta. 
\end{align*}
\end{proof}

\begin{proof}[Proof of Corollary~\ref{cor:IntegrableGibbs}]
We split
\begin{align}
\label{eq:MajorTriangle}
   \left\|e^{t\cL^{\le M}_{\sigma_E,\widehat f,H_{\le M}}}(\rho) -    \sigma_\beta\right\|_1  &\le \left\|e^{t\cL_{\sigma_E,\widehat    f,H}}(\rho) - \sigma_\beta\right\|_1+\left\|\left(e^{t\cL^{\le    M}_{\sigma_E,\widehat f,H_{\le M}}} -e^{t\cL_{\sigma_E,\widehat    f,H}}\right)(\rho)\right\|_1  
\end{align}
For the first term, we argue as in Section~\ref{sec.fastconv} using positivity of spectral gap of $L_{\sigma_E,\widehat f,H}$ together with \eqref{eq:InputDmaxConstraint} to get
\begin{align*} \left\|e^{t\cL_{\sigma_E,\widehat{f},H}}(\rho)- \sigma_\beta\right\|_1 &\le  e^{-\lambda_2 t} \left\|\sigma_\beta^{-1/4}\rho\sigma^{-1/4}_\beta\right\|_2 = e^{-\lambda_2 t} \sqrt{\Tr\left(\sigma^{-1/2}_\beta\rho\sigma^{-1/2}_\beta\rho\right)} \\&\le \sqrt{\mathfrak{c}}\,e^{-\lambda_2 t} \sqrt{\Tr(\rho)}= \sqrt{\mathfrak{c}}\,e^{-\lambda_2 t},
\end{align*}
which gives the desired bound on the mixing time $t.$ The result then follows by Theorem~\ref{thm:IntegrableGibbsDynamics}.
\end{proof}

\paragraph{}
The warm-start hypothesis $\rho\le\mathfrak c\,\sigma_\beta$ for some $\mathfrak c\ge1$ in \eqref{eq:InputDmaxConstraint} can alternatively be replaced by
\[
R:=\|\sigma_\beta^{-1/4}\rho\sigma_\beta^{-1/4}\|_2<\infty,
\qquad E_8:=\Tr(e^{8\NA^\kappa}\sigma_\beta)<\infty.
\]
Writing $x_s=e^{sL_{\sigma_E,\widehat f,H}}
(\sigma_\beta^{-1/4}\rho\sigma_\beta^{-1/4})$, contractivity gives
$\|x_s\|_2\le R$. For bounded spectral truncations $B$ of
$e^{4\NA^\kappa}$, Schatten Cauchy--Schwarz gives
\[
\Tr(B\rho(s))\le R\|\sigma_\beta^{1/4}B\sigma_\beta^{1/4}\|_2
\le R\,[\Tr(\sigma_\beta B^2)]^{1/2}.
\]
The last inequality follows by applying Hilbert--Schmidt Cauchy--Schwarz
to $B\sigma_\beta^{1/2}$ and its adjoint. Monotone convergence therefore
bounds $\Tr(e^{4\NA^\kappa}\rho(s))$ by $RE_8^{1/2}$.
Moreover, normalization gives
$\Tr(e^{2\NA^\kappa}\rho(s))^2\le\Tr(e^{4\NA^\kappa}\rho(s))$.
The proof of Theorem~\ref{thm:IntegrableGibbsDynamics} thus yields
\[
\|(e^{t\cL_{\sigma_E,\widehat f,H}}
-e^{t\cL^{\le M}_{\sigma_E,\widehat f,H_{\le M}}})(\rho)\|_1
\le t C_1p_4(|\mathcal A|,M)e^{-M^\kappa}R^{1/2}E_8^{1/4}.
\]
For the original warm-start branch the same argument sharpens its
prefactor $\mathfrak c E_{\operatorname{Gibbs}}$ to
$\sqrt{\mathfrak c E_{\operatorname{Gibbs}}}$.
The generator-domain requirements in the Duhamel argument hold because,
for every $p\ge0$, the same estimate gives
$\Tr(\widetilde H^p\rho(s))\le
R[\Tr(\widetilde H^{2p}\sigma_\beta)]^{1/2}<\infty$.
With a simple kernel and gap $\lambda_2$, the mixing estimate is
\[
\|e^{t\cL_{\sigma_E,\widehat f,H}}(\rho)-\sigma_\beta\|_1
\le\sqrt{R^2-1}\,e^{-\lambda_2t};
\]
in the warm-start branch it is at most
$\sqrt{\mathfrak c-1}\,e^{-\lambda_2t}$.
All subsequent circuit bounds also apply to these inputs, with
$\mathfrak c E_{\operatorname{Gibbs}}$ in the resource logarithms
replaced by $RE_8$, provided the normalized finite-register preparation
oracle is supplied. The hypothesis on $E_8$ is stronger than the original
fourth-moment hypothesis, which remains sufficient for warm starts.
This enlarges the input class: for $H=N$ and $0<\kappa<1$, a thermal
input at inverse temperature $\beta/2<\beta'<\beta$ has finite $R$ but
is not bounded above by any finite multiple of $\sigma_\beta$.

\subsection{Finite-dimensional circuit implementation}
\label{sec:FullFiniteDimPipeline}

\noindent In this final implementation section we show how the dynamics resulting from the generator $\cL_{\smash{\sigma_E,\widehat{f},H_{\le M}}}^{\le M}$ can be implemented by a quantum circuit through discretizations of the integral representations of its jump and coherent parts. By the integral representation in Proposition~\ref{integralrep}, we note that the generator $\cL_{\smash{\sigma_E,\widehat{f},H_{\le M}}}^{\le M}$ is constructed from a continuous family of jump operators. In Section~\ref{sec:IntegralDiscretization} below we, therefore, first establish that this generator can be written as a Gaussian integral over certain generators of the form \eqref{eq:GKLS} which can then be approximated by a Lindblad generator with a finite number of jumps using Gauss-Hermite quadratures, c.f.~Proposition ~\ref{prop:DiscreteGauss}. After this discretization, and assuming the filter function satisfies \cite[Assumption 13]{ding2025efficient}, as recalled in \Cref{cond:Window} and \ref{condfgevrey}, we can invoke \cite[Theorem 18]{ding2025efficient} to obtain an efficient circuit implementation of the finite-dimensional Lindblad dynamics.

Finally, combining this with
Theorem~\ref{thm:IntegrableGibbsDynamics} yields an efficient
finite-dimensional circuit implementation of
  $e^{t\cL_{\sigma_E,\widehat f,H}}$ for integrable filters satisfying
the stated weighted-moment and Gevrey assumptions, in
Theorem~\ref{thm:CircuitDynamicsIntegrable}. Given a simple kernel and a positive spectral gap
of the associated generators on the space of Hilbert-Schmidt operators,
Corollary~\ref{cor:GibbsCircuitIntegrable} then gives
efficient Gibbs state preparation. We verify below that the
polynomial filters $\widehat f_\eta$, $\eta>0$, satisfy the filter
assumptions directly.

\subsubsection{Integral discretizations}
\label{sec:IntegralDiscretization}

We start by recalling \cite[Assumption 13]{ding2025efficient} for the cut-off part of the filter function.
\begin{definition}[Gevrey function]
 A complex-valued $C^\infty$ function
$h : \mathbb{R} \to \mathbb{C}$ is called a \emph{Gevrey function of order} $s \ge 0$ if there exist
constants $C_1, C_2 > 0$ such that, for every nonnegative integer
$n $, the derivatives of $h$ satisfy
  \begin{equation}
\label{eq:gevrey-def}
\| h^{(n)}\|_{L^\infty(\mathbb{R})}
\le
C_1 C_2^{n} n^{n s}.
\end{equation} 
For fixed constants $C_1, C_2, s$, the set of such Gevrey functions is denoted by $\mathcal{G}^s_{C_1,C_2}$. Furthermore, we define
\[
\mathcal{G}^s = \bigcup_{C_1,C_2>0} \mathcal{G}^s_{C_1,C_2}.
\]

\end{definition}

\begin{condition}
\label{cond:Window}
 In what follows, we also consider a cut-off function $w\in  \mathcal{G}^s_{\xi_q,\xi_w}$ for some  $s>1$  and  $\xi_w\ge 1$, with $w(\nu) = \overline{w(-\nu)}$  for all $\nu\in\mathbb{R}$ and such that 
    \[
    w(\nu) = 1 \quad \text{when } |\nu| \le \tfrac12,
    \qquad
    w(\nu) = 0 \quad \text{when } |\nu| \ge 1.
    \]
    Finally, we denote by $\kappa(\nu)=w(\nu/S)$ for some $S>0$     which we will choose later.
\end{condition}

\noindent Next, for   $f\in L^1(\mathbb R)$ we show  that the
generator   $\cL_{\smash{\sigma_E,\widehat f,H_{\le M}}}^{\le M}$
from Section~\ref{sec:FinitDimGenerator}  can be approximated by a
  Lindblad generator with a finite  number of jumps  , depending on
the requested accuracy and truncation parameters . We recall that
\begin{align*}
\cL_{\smash{\sigma_E,\widehat{f},H_{\le M}}}^{\le M}(\rho)=\sum_{\alpha\in\cA} \int_\R \gamma(t) X_{t,\le M}^{\alpha,\le M}\rho \left(X_{t,\le M}^{\alpha,\le M}\right)^\dagger dt\, + G^{\le M}_{\sigma_E,\le M}\rho +\rho \left(G^{\le M}_{\sigma_E,\le M}\right)^\dagger 
\end{align*}
with 
  \begin{align*}
G^{\le M}_{\sigma_E,\le M} &=-\sum_{\alpha\in\cA} \int_\R g(t) (X_{t,\le M}^{\alpha,\le M})^\dagger X_{t,\le M}^{\alpha,\le M}dt,\quad &
X_{t,\le M}^{\alpha,\le M}:=e^{itH_{\le M}}L_{\le M}^{\alpha,\le M}e^{-itH_{\le M}},\\\gamma(t)&:=\sigma_E\sqrt{\frac{2}{\pi}}e^{-2\sigma_E^2 t^2}\qquad\qquad\qquad\text{and}&g(t):=\frac{1}{2\pi}\int _\R \frac{e^{-\nu^2/8\sigma_E^2}}{1+e^{\beta\nu/2}}e^{-i\nu t}d\nu
\end{align*} 
and $L^{\alpha,\le M}_{\le M}$ has been defined in \eqref{eq:TruncatedL}.
  By decomposing the function $g$ into real and imaginary part as
\begin{align*}
g(t)&=\frac{1}{2\pi}\int \frac{e^{-\nu^2/8\sigma_E^2}}{1+e^{\beta\nu/2}}e^{-i\nu t}d\nu=\frac{1}{2}\gamma(t)+\frac{1}{2\pi}\int {e^{-\nu^2/8\sigma_E^2}}e^{-i\nu t}  \frac{1-e^{\beta\nu/2}}{2(1+e^{\beta\nu/2})} d\nu \\&=\frac{1}{2}\gamma(t)+\frac{1}{4\pi}\int {e^{-\nu^2/8\sigma_E^2}}e^{-i\nu t} \tanh(-\beta \nu/4) d\nu,
\end{align*}
we get that 
  \begin{align*}
\cL_{\smash{\sigma_E,\widehat{f},H_{\le M}}}^{\le M}(\rho)&=\sum_\alpha \int \gamma(t)\left( X_{t,\le M}^{\alpha,\le M}\rho \left(X_{t,\le M}^{\alpha,\le M}\right)^\dagger-\frac{1}{2}\left\{\left(X_{t,\le M}^{\alpha,\le M}\right)^\dagger X_{t,\le M}^{\alpha,\le M},\rho\right\}\right) dt-i\big[B^{\le M}_{\sigma_E},\rho\big],
\end{align*} 
with
\begin{align*}
B_{\sigma_E}^{\le M}:=\sum_{\alpha\in\mathcal{A}}\frac{(-i)}{4\pi}\int\int {e^{-\nu^2/8\sigma_E^2}}e^{-i\nu t}  \tanh(-\beta\nu/4)  e^{itH_{\le M}}(L_{\le M}^{\alpha,\le M})^\dagger L_{\le M}^{\alpha,\le M} e^{-itH_{\le M}}\,d\nu dt=(B_{\sigma_E}^{\le M})^\dagger,
\end{align*}
where the self-adjointness follows from the antisymmetry of $\tanh$. Decomposing into the eigenbasis of $H_{\le M} = \sum_{E\in\spec(H_{\le M})} E \,P_E$ and denoting $\big[(L_{\le M}^{\alpha,\le M})^\dagger L_{\le M}^{\alpha,\le M}\big]_\mu=\sum_{E'-E=\mu}P_{E'}(L_{\le M}^{\alpha,\le M})^\dagger L_{\le M}^{\alpha,\le M}P_{E}$, we get
\begin{align*}
B_{\sigma_E}^{\le M}&=\sum_{\alpha\in\mathcal{A}}\sum_{\mu\in B(H_{\le M})}\frac{(-i)}{4\pi}\int\int {e^{-\nu^2/8\sigma_E^2}}e^{-i(\nu-\mu) t}  \tanh(-\beta\nu/4)  \big[(L_{\le M}^{\alpha,\le M})^\dagger L_{\le M}^{\alpha,\le M} \big]_\mu\,d\nu dt\\
&=\sum_{\alpha\in\mathcal{A}}\sum_{\mu\in B(H_{\le M})}\frac{(-i)}{2} {e^{-\mu^2/8\sigma_E^2}}  \tanh(-\beta\mu/4)  \big[(L_{\le M}^{\alpha,\le M})^\dagger L_{\le M}^{\alpha,\le M} \big]_\mu.
\end{align*}
Using the notation from \Cref{cond:Window}, choose $S\ge4\|H_{\le M}\|$ and define the Schwartz function
\[
 \widehat t_\kappa(\mu):=-\frac{i}{2}\tanh(-\beta\mu/4)\kappa(\mu).
\]
Every Bohr frequency $\mu\in B(H_{\le M})$ satisfies
$|\mu|\le2\|H_{\le M}\|\le S/2$, so
$\widehat t_\kappa(\mu)=-\frac{i}{2}\tanh(-\beta\mu/4)$. Hence
  \begin{align*}
B_{\sigma_E}^{\le M}&=\sigma_E\sqrt{\frac{2}{\pi}}\sum_{\alpha\in\mathcal{A}}\sum_{\mu\in B(H_{\le M})}\int e^{-i\mu t-2\sigma_E^2 t^2 }   \widehat{t}_\kappa(\mu)  \big[(L_{\le M}^{\alpha,\le M})^\dagger L_{\le M}^{\alpha,\le M} \big]_\mu dt\\
&=\sigma_E\sqrt{\frac{2}{\pi}}\sum_{\alpha\in\mathcal{A}}\sum_{\mu\in B(H_{\le M})}\int e^{-2\sigma_E^2 t^2 }   \widehat{t}_\kappa(\mu) e^{-itH_{\le M}} \big[(L_{\le M}^{\alpha,\le M})^\dagger L_{\le M}^{\alpha,\le M} \big]_\mu e^{itH_{\le M}}\,dt\\
&=\sum_{\alpha\in\mathcal{A}}\int \gamma(t)   {t}_\kappa(s) e^{i(s-t)H_{\le M}} (L_{\le M}^{\alpha,\le M})^\dagger L_{\le M}^{\alpha,\le M}  e^{-i(s-t)H_{\le M}}\,ds\,dt\\
&= \int \gamma(t)\,e^{-it H_{\le M}}B^{\le M}e^{itH_{\le M}}dt ,
\end{align*} 
where in the last equality we have denoted 
\begin{align*}
   B^{\le M} := \sum_{\alpha\in\cA} \int_\R t_\kappa(s)    \,e^{isH_{\le M}}\left(L^{\alpha,\le M}_{\le M}\right)^\dagger    L^{\alpha,\le M}_{\le M}e^{-isH_{\le M}} \,ds. 
\end{align*}
 Evenness of $\gamma$ makes the two time conventions equivalent.
 All in all, we have found that 
\begin{align}\label{convcomb}
\mathcal{L}_{\sigma_E,\widehat{f},H_{\le M}}^{\le M}(\rho)=\int \gamma(t) \,\mathcal{L}_{\widehat{f},H_{\le M}}^{\le M,\,t}(\rho)\, dt,
\end{align}
where $\mathcal{L}_{\widehat{f},H_{\le M}}^{\le M,\,t}$ coincides with the generator $\mathcal{L}_{\widehat{f},H_{\le M}}$ defined in \eqref{eq:GKLS} with filter function $\widehat{f}$, jumps $\left(A^{\alpha}\right)^{\le M}_t:= e^{itH_{\le M}} (A^{\alpha})^{\le M} e^{-itH_{\le M}}$ and Hamiltonian $H_{\le M}$. 

Next, we need to argue that, without loss of generality, the filter function $\widehat{f}$ can be replaced by  $\widehat{f}_{\kappa}(\nu):=\widehat{f}(\nu)\kappa(\nu)$ with $\kappa$ satisfying \Cref{cond:Window} for some choice of parameter $S$. This is done by observing that
\begin{align}
\label{eq:LMkappaVsNoKappa}
\nn L_{\le M}^{\alpha,\le M}=\int f(s)\,e^{isH_{\le M}}(A^{\alpha})^{\le M}e^{-isH_{\le M}}\,ds&=\sum_{\nu\in B(H_{\le M})}\widehat{f}(\nu) ((A^{\alpha})^{\le M})_{\nu}\\\nn
&=\sum_{\nu\in B(H_{\le M})}\widehat{f}_\kappa(\nu) ((A^{\alpha})^{\le M})_{\nu}\\
&=\int f_\kappa(s)\,e^{is H_{\le M}}(A^{\alpha})^{\le M}e^{-is H_{\le M}} ds,
\end{align}
for $S\ge 4\|H_{\le M}\|$.
Indeed, if $k_S$ is the inverse Fourier
transform of $\kappa$, then $k_S$ is Schwartz and
$f_\kappa=f*k_S\in L^1(\mathbb R)$, by the convolution theorem
and Young's inequality. Thus the integral identities remain
valid for any $f\in L^1$, even if it has an integrable singularity
at the origin. The cutoff does not change any Bohr-frequency
coefficient of the finite-dimensional generator.
 Thus, without loss of generality, we can
replace $\mathcal{L}_{\widehat{f},H_{\le M}}^{\le M,\,t}$ by $\mathcal{L}_{\widehat{f}_\kappa,H_{\le M}}^{\le M,\,t}$ in \Cref{convcomb}.

We invoke Gaussian-Hermite quadratures, recalled in the following supporting lemma, to find a discretization of the Gaussian integral in \eqref{convcomb} in Proposition~\ref{prop:DiscreteGauss} below.

\begin{lemma}
\label{lem:GaussHermite}
Let $X$ be a Banach space, $F:\mathbb R\to X$ be $C^\infty$ and assume that
\[
\|F^{(m)}(t)\|_X \le K\,C^m
\qquad
\text{for all }m\in\mathbb N_0,\ t\in\mathbb R,
\]
for some constants $C,\,K>0 $.  Let further $\sigma_E,\varepsilon>0.$ Let $\{(x_k,w_k)\}_{k=1}^n \subset \R\times \R$ be the $n$-point Gauss--Hermite nodes and weights with $\max_k|x_k|\le 2\sqrt{n}$ and $\max_k w_k=\mathcal{O}(n^{-1/2})$, and define
  \[
Q_n(F)
:=
\frac{1}{\sqrt{\pi}}
\sum_{k=1}^n
w_k\,F\!\left(\frac{x_k}{\sqrt2\,\sigma_E}\right).
\] 
Then, 
  \[
n \ge \max\left\{1,\frac{C^2}{2\sigma_E^2},\log_2\!\left(\frac{K}{\varepsilon}\right)\right\}\quad \Rightarrow\quad 
\left\|
\int_{\mathbb R} \gamma(t)\,F(t)\,dt
-
Q_n(F)
\right\|_X
\le \varepsilon.
\] 
\end{lemma}

 \begingroup
 \begin{proof}
  If $X$ is complex, regard it as a real Banach space with the same norm.
For any bounded real-linear functional $\varphi:X\to\mathbb R$ with
$\|\varphi\|\le1$, apply the scalar Gauss--Hermite remainder formula
\mbox{%
\cite[\S3.5(v), equations (3.5.19) and (3.5.28)]{DLMFQuadratureGreen}
}\hskip0pt%
to the real-valued function $g_\varphi(x)=\varphi(F(x/(\sqrt2\sigma_E)))$.
Its derivatives satisfy
$|g_\varphi^{(2n)}(x)|\le K(C/(\sqrt2\sigma_E))^{2n}$.
Bounded real-linear functionals commute with the integral and the
quadrature sum. Taking the supremum over these functionals and using
the real Hahn--Banach theorem therefore gives
\[
\left\|\int_{\mathbb R}\gamma(t)F(t)\,dt-Q_n(F)\right\|_X
\le K\frac{n!}{(2n)!}\left(\frac{C}{2\sigma_E}\right)^{2n}
\le K\left(\frac{C^2}{4\sigma_E^2 n}\right)^n.
\] 
  The stated conditions on $n$ make the last expression at most
$K2^{-n}\le\varepsilon$.
 \end{proof}
 \par\endgroup

\begin{proposition}
\label{prop:DiscreteGauss}
 Let $H$ be such that the truncated Hamiltonian $H_{\le M}$, defined  in Section~\ref{sec:FinitDimGenerator}, satisfies    $\|H_{\le M}\| \le  p_5(|\cA|,M)$  for some  polynomially bounded function $p_5.$ Moreover, let  $\{A^\alpha\}_{\alpha\in\cA}$ be a set of bare jumps satisfying  \eqref{eq:NormTruncatedJump}, i.e. $\|\left(A^\alpha\right)^{\le  M}\|\le q'(M)$ for some polynomially bounded function $q'(M).$
 Let $\kappa$ satisfy  \Cref{cond:Window} for   $s>1$  and $\xi_q,\xi_w\ge 0$ and  
   $S=\max\{1,4p_5(|\cA|,M)\}.$ 
  Let furthermore $\beta>0$ and  
$f\in L^1(\mathbb R)$  and denote 
  \begin{align*}
    K:=|\mathcal{A}|\xi_q\,\bigl(1+\log\left(2+(\beta+    \xi_w/S)\max\{S,1\}\right)\bigr)   \,\left(q'(M)\right)^2     \min\left\{\|f\|^2_{L^1(\R)},\|f_\kappa\|^2_{L^1(\R)}\right\}.
\end{align*}  
Then for $\beta,\eps>0$ and $t\ge 0$ we can achieve
  \begin{align}
\label{eq:DiscreteGaussEvolution}
 \Big\|e^{t\mathcal{L}^{\le M}_{\sigma_E,\widehat{f},H_{\le  M}}}-e^{t\mathcal{L}^{\le M,\mathfrak{n}}_{\sigma_E,\widehat{f},H_{ \le M}}}\Big\|_{1\to 1}\le \epsilon \qquad \text{with}\qquad  \mathfrak{n}=\mathcal{O}\left(   1+\frac{S^2}{\sigma_E^2}+\log\Big(2+\frac{K t}{\epsilon}\Big)\right)
\end{align} 
with, given the $\mathfrak{n}$-point Gauss-Hermite nodes and weights $\{(x_k,w_k)\}_{k=1}^{\mathfrak{n}}$ with $\max_{k}w_k=\mathcal{O}(\mathfrak{n}^{-1/2})$ and   $\max_k|x_k|\le 2\mathfrak{n}^{1/2}$ , 
\begin{align}
\label{eq:ImplementedGenerator}
\mathcal{L}^{\le M,\mathfrak{n}}_{\sigma_E,\widehat{f},H_{\le M}}:=\frac{1}{\sqrt{\pi}}
\sum_{k=1}^{\mathfrak{n}}
w_k\,\mathcal{L}_{\widehat{f},H_{\le M}}^{\le M,x_k/(\sqrt{2}\sigma_E)}
\end{align}
\end{proposition}

\begin{proof}
We use
  \begin{align*}
e^{t\mathcal{L}^{\le M}_{\sigma_E,\widehat{f},H_{\le M}}}-e^{t\mathcal{L}^{\le M,\mathfrak{n}}_{\sigma_E,\widehat{f},H_{\le M}}} = \int^t_0 e^{s\mathcal{L}^{\le M}_{\sigma_E,\widehat{f},H_{\le M}}}\left(\mathcal{L}^{\le M}_{\sigma_E,\widehat{f},H_{\le M}} - \mathcal{L}^{\le M,\mathfrak{n}}_{\sigma_E,\widehat{f},H_{\le M}}\right) e^{(t-s)\mathcal{L}^{\le M,\mathfrak{n}}_{\sigma_E,\widehat{f},H_{\le M}}} \,ds
\end{align*} 
and, therefore, as both dynamics are contractive in trace norm and the $1\to 1$ norm is submultiplicative, we have
  \begin{align*}
\Big\|e^{t\mathcal{L}^{\le M}_{\sigma_E,\widehat{f},H_{\le M}}}-e^{t\mathcal{L}^{\le M,\mathfrak{n}}_{\sigma_E,\widehat{f},H_{\le M}}}\Big\|_{1\to 1} \le t\left\|\mathcal{L}^{\le M}_{\sigma_E,\widehat{f},H_{\le M}} - \mathcal{L}^{\le M,\mathfrak{n}}_{\sigma_E,\widehat{f},H_{\le M}}\right\|_{1\to 1}
\end{align*} 
Hence, using Lemma~\ref{lem:GaussHermite} for the function $t\mapsto F(t):=\mathcal{L}_{\widehat{f},H_{\le M}}^{\le M,t}$, it suffices to compute estimate the  derivatives as  
  \begin{align*}
\|F^{(m)}(t)\|_{1\to 1}&\le 2|\mathcal{A}|\, (4\|H_{\le M}\|)^m\max_{\alpha\in\mathcal{A}}\left\|L_{\le M}^{\alpha,\le M}\right\|^2+2(4\|H_{\le M}\|)^m \|B^{\le M}\|\\
&\le 2|\mathcal{A}|\, (4\|H_{\le M}\|)^m\max_{\alpha\in\mathcal{A}}\left\|L_{\le M}^{\alpha,\le M}\right\|^2(1+\|t_\kappa\|_{L^1(\mathbb{R})})\\
&\le 2|\mathcal{A}|\, (4\|H_{\le M}\|)^m\min\left\{\|f\|^2_{L^1(\R)},\|f_\kappa\|^2_{L^1(\R)}\right\} \max_{\alpha\in\mathcal{A}}\left\|\left(A^\alpha\right)^{\le M}\right\|^2(1+\|t_\kappa\|_{L^1(\mathbb{R})})\\ &\le 2|\mathcal{A}|\, (4\|H_{\le M}\|)^m\min\left\{\|f\|^2_{L^1(\R)},\|f_\kappa\|^2_{L^1(\R)}\right\} (q'(M))^2(1+\|t_\kappa\|_{L^1(\mathbb{R})}),
\end{align*} 
where in the second to last inequality we have used \eqref{eq:LMkappaVsNoKappa} to bound the operator norm of $L_{\le M}^{\alpha,\le M}.$
Using now \Cref{cond:Window} combined with \cite[Lemma 30]{ding2025efficient}, we see
  \begin{align*}
\|t_\kappa\|_{L^1(\mathbb{R})} = \mathcal{O}\bigl(\xi_q\,(1+\log(2+(\beta+\xi_w/S)\max\{S,1\}))\bigr),
\end{align*} 
which finishes the proof.
\end{proof}

We can now employ the \cite[Theorem 18]{ding2025efficient} to provide a circuit implementation of the finite-dimensional dynamics. For that we first recall the condition on the filter function in \cite[Assumption 13]{ding2025efficient}.

\begin{condition} \label{condfgevrey}
For $\beta > 0$, we consider a filter function $\widehat{f}$ such that $\nu\mapsto \widehat{f}(\nu/\beta)\in   \mathcal{G}^s_{\smash{\xi_q,\xi_u}}$ for some constants $\xi_q,\xi_u \ge 1$ and $s \ge 1$. In addition, we assume
$
\frac{d}{d\nu}\widehat{f}(\nu) \in L^1(\mathbb{R}),
$
and denote
\[
C_{1,f}
:=
\left\|
\frac{d}{d\nu}\widehat{f}(\nu/\beta)
\right\|_{L^1(\mathbb{R})}.
\]
\end{condition} 
Together with the KMS symmetry in
\eqref{eq:KMSFilterFunction}, this is the filter part of
\cite[Assumption~13]{ding2025efficient}: set
$u(x)=e^{x/4}\widehat f(x/\beta)$, so that
$u(-x)=\overline{u(x)}$, and use the window from
Condition~\ref{cond:Window}. The corresponding weighting function is
$u(\beta\nu)\kappa(\nu)$, whose product with
$e^{-\beta\nu/4}$ is precisely $\widehat f_\kappa(\nu)$.
Under these assumptions $\widehat f_\kappa$ is smooth with
compact support, so $f_\kappa$ is Schwartz even when $f$ is not.

 \noindent The polynomial filters satisfy the required Gevrey
condition directly.
\begin{lemma}\label{lem:PolynomialFilterGevrey}
For every $\eta>0$ and $\beta>0$, the filter $\widehat f_\eta$
in \eqref{eq:PolynomialFilter} satisfies
Condition~\ref{condfgevrey}.
More precisely, with $h_\eta(x):=\widehat f_\eta(x/\beta)$ and
$C_\eta=(4/3)^{\eta/2}$,
\[
\|h_\eta^{(n)}\|_\infty\le C_\eta\,4^n n!,
\qquad n\ge0,
\qquad \|h_\eta'\|_1\le3.
\]
Thus one can take $\xi_q=\max\{1,C_\eta\}$, $\xi_u=4$ and
$s=1$ (and hence any $s>1$).
\end{lemma}
\begin{proof}
The function
\[
h_\eta(z)=
\frac{e^{-(\sqrt{1+z^2}+z)/4}}{(1+z^2)^{\eta/2}}
\]
is holomorphic on $|\Im z|<1$, using the principal branches.
For $z=x+iy$ with $|y|\le1/2$, one has
$|1+z^2|\ge\Re(1+z^2)\ge3/4$ and
$\Re\sqrt{1+z^2}\ge|x|$. Hence $|h_\eta(z)|\le C_\eta$.
Cauchy's estimate on the circle $|z-x|=1/4$ gives the displayed
derivative bound; use $n!\le n^n$ for $n\ge1$.
For the $L^1$ bound, write $h_\eta=uv$ on the real axis, where
$u(x)=e^{-(\sqrt{1+x^2}+x)/4}$ and
$v(x)=(1+x^2)^{-\eta/2}$. Then $0<u,v\le1$,
$u$ decreases from $1$ to $0$, and $v$ increases from $0$ to $1$
on $(-\infty,0)$ and decreases back to $0$ on $(0,\infty)$.
Consequently,
\[
\|h_\eta'\|_1\le\|u'\|_1+\|v'\|_1=1+2=3.
\]
These bounds are independent of $\beta$.
\end{proof}

\begingroup
We next reduce the quadrature generator \eqref{eq:ImplementedGenerator} to the
setting of \cite[Theorem~18]{ding2025efficient}. The reduction keeps track of
the jump normalization, the Gauss--Hermite weights, the additional Hamiltonian
conjugations, and the required oracle access, leading to the following circuit
complexity bound.
\begin{corollary}[Circuit simulation of the quadrature generator]
\label[corollary]{thm:Lin}
Let $\beta,\sigma_E>0$, $M,\mathfrak n\in\mathbb N$, $t\ge0$, and
$0<\eps<1$. Consider the invariant finite-dimensional system register
$\operatorname{im}(P_M)$ from Section~\ref{sec:FinitDimGenerator} and assume
the adjoint-pair convention, the norm bound \eqref{eq:NormTruncatedJump}
with $q'(M)\ge1$, the KMS symmetry \eqref{eq:KMSFilterFunction}, and
Conditions~\ref{condfgevrey} and~\ref{cond:Window} with a common order
$s>1$ and
\[
S\ge\max\{1,4\|H_{\le M}\|\}.
\]
Define
\[
\kappa_B(\nu):=w(\nu/(4S)),
\qquad
\widehat t_{\kappa_B}(\nu)
:=-\frac{i}{2}\tanh(-\beta\nu/4)\kappa_B(\nu).
\]

Assume coherent, index-controlled access to block encodings of
$(A^\alpha)^{\le M}/q'(M)$ and their adjoints, controlled Hamiltonian
simulation with respect to $H_{\le M}$, and the state-preparation oracles
for $f_\kappa$ and $t_{\kappa_B}$ required in
\cite[Eqs.~(3.27)-(3.31)]{ding2025efficient}. Assume further coherent
reversible access to the Gauss-Hermite nodes and weights to the precision
required by the simulation.

Then the restriction of
$e^{t\mathcal L^{\le M,\mathfrak n}_{\sigma_E,\widehat f,H_{\le M}}}$
to the system register can be simulated within $\eps$ in diamond norm
with total Hamiltonian simulation time
\begin{equation}
\label{eq:QuadratureCircuitCost}
\widetilde{\mathcal O}\left(
[1+t(q'(M))^2]\mathfrak n^2|\cA|^2
\left(\beta+1+\frac{\sqrt{\mathfrak n}}{\sigma_E}\right)
[1+\log^{2+s}S]\log^{1+s}(2/\eps)
\right).
\end{equation}
In addition to the system register and the workspace required by the
supplied oracles, the number of ancillary qubits is at most
\begin{align*}
\widetilde{\mathcal O}\Bigg(
&\log(2+S^2+S\|H_{\le M}\|)
+\log^2\left(
2+\frac{[1+t(q'(M))^2]\mathfrak n|\cA|}{\eps}
\right) \\
&+\log(2+\beta)
+\log\left(
2+\frac{\|H_{\le M}\|\sqrt{\mathfrak n}}{\sigma_E}
\right)
\Bigg).
\end{align*}
The hidden constants depend only on
$s,\xi_q,\xi_u,\xi_w,C_{1,f}$, and
$\widetilde{\mathcal O}$ additionally suppresses subdominant
polylogarithmic factors in the displayed parameters.
\end{corollary}

\begin{proof}
Write $H_M:=H_{\le M}$. We first reduce the quadrature generator to the
normalised setting of \cite[Theorem~18]{ding2025efficient}. For each
non-self-adjoint pair $A,A^\dagger$, replace the two jumps by
$(A+A^\dagger)/\sqrt2$ and $(A-A^\dagger)/(i\sqrt2)$, leaving
self-adjoint jumps unchanged. This unitary change of jump labels leaves
the completely positive part, $\sum_\alpha(L^\alpha)^\dagger L^\alpha$,
and hence also the coherent correction unchanged. Denoting the resulting
self-adjoint jumps by $A_j^{\mathrm{sa}}$, we have
$\|A_j^{\mathrm{sa}}\|\le\sqrt2\,q'(M)$.

Let $\{(x_k,w_k)\}_{k=1}^{\mathfrak n}$ be the Gauss-Hermite nodes and
weights from Proposition~\ref{prop:DiscreteGauss}, and set
\[
p_k:=\frac{w_k}{\sqrt\pi},
\qquad
s_k:=\frac{x_k}{\sqrt2\sigma_E},
\qquad
R:=\sqrt2\,q'(M).
\]
Since Gauss-Hermite quadrature is exact on constants, $p_k>0$ and
$\sum_kp_k=1$. We define the self-adjoint contractions
\[
\overline A^{j,k}
:=
\frac{\sqrt{p_k}}{R}
e^{is_kH_M}A_j^{\mathrm{sa}}e^{-is_kH_M}.
\]
There are at most $N=\mathfrak n|\cA|$ such jumps. Since filtering
commutes with Hamiltonian conjugation, the corresponding filtered jumps
are
\[
\overline L_{j,k}
=
\frac{\sqrt{p_k}}{R}
e^{is_kH_M}L_je^{-is_kH_M}.
\]
Moreover, the choices of $\kappa$ and $\kappa_B$ leave all frequencies
entering the finite-dimensional jump and coherent terms unchanged.
Therefore, the generator $\overline{\mathcal L}$ associated with the
jumps $\overline A^{j,k}$ in the framework of
\cite[Theorem~18]{ding2025efficient} satisfies
\begin{equation}
\label{eq:NormalizedQuadratureGenerator}
\overline{\mathcal L}
=
R^{-2}
\mathcal L^{\le M,\mathfrak n}_{\sigma_E,\widehat f,H_M}.
\end{equation}
It is consequently enough to simulate $\overline{\mathcal L}$ for time
$t'=R^2t=2(q'(M))^2t$.

The KMS symmetry together with Conditions~\ref{condfgevrey} and
\ref{cond:Window} gives precisely the filter assumptions of
\cite[Theorem~18]{ding2025efficient}. In particular, with
$u(x):=e^{x/4}\widehat f(x/\beta)$, the KMS symmetry gives
$u(-x)=\overline{u(x)}$, while the coherent filter
$t_{\kappa_B}$ is the enlarged-window filter prescribed in
\cite[Eq.~(3.17)]{ding2025efficient}.

It remains to account for the additional Hamiltonian conjugations.
A block encoding of $\overline A^{j,k}$ is obtained from that of
$A_j^{\mathrm{sa}}/R$ by a controlled rotation of amplitude
$\sqrt{p_k}$ and the conjugations $e^{\pm is_kH_M}$. Since
$\max_k|x_k|\le2\sqrt{\mathfrak n}$,
\[
\max_k|s_k|
=
\mathcal O\left(\frac{\sqrt{\mathfrak n}}{\sigma_E}\right).
\]
Hence each such query requires at most
$\mathcal O(\sqrt{\mathfrak n}/\sigma_E)$ additional Hamiltonian
simulation time. Applying \cite[Theorem~18]{ding2025efficient} with
$N=\mathfrak n|\cA|$ and $t'=2(q'(M))^2t$ now gives
\eqref{eq:QuadratureCircuitCost}.

The ancillary-qubit estimate follows from the corresponding estimate in
\cite[Theorem~18]{ding2025efficient}, together with the registers used
for the quadrature label and the reversible evaluation of the nodes and
weights. Choosing the precision of these auxiliary operations below the
corresponding fraction of the total error budget and applying the
standard hybrid argument yields total diamond-norm error at most $\eps$.
\end{proof}

\par\endgroup

\begin{remark}[Circuit implementation of $a_i^{\le M}$ and $(a_i^{\le M})^\dagger$]
\label{rem:CircuitTruncJumps}
For an $m$-mode continuous variable system with Hilbert space $(L^2(\R))^\otimes \cong L^2(\R^m),$ we usually consider bare jumps being the annihilation and creation operators on each site, i.e.~$\{A^\alpha\}_{\alpha\in\cA} \equiv \{a_i,a^\dagger_i\}_{i=1}^{m}.$ Furthermore, we then consider their Fock basis truncations $a_i^{\le M}$ and $\left(a_i^{\le M}\right)^\dagger$ (c.f.~\eqref{eq:DefTruncAnnCrea}), whose operator norms satisfy $\|a^{\le M}_i\|\,,\, \|(a^{\le M}_i)^\dagger\|\le \sqrt{M}.$
For fixed mode $i$ these truncated operators can be seen as acting on a $(M+1)$-dimensional space and we, hence, argue in the following how to implement them on a $q=\lceil \log_2(M+1)\rceil$ qubit device. In binary encoding, the  subnormalised operators
$a_i^{\le M}/\sqrt{M}$ and $(a_i^{\le M})^\dagger/\sqrt{M}$ are weighted shift operators. A naive
implementation via a flat lookup of the $M+1$ coefficients has depth
$\mathcal O(M)$, up to precision overhead. 

However, this can be drastically improved
by noting that the shift part $\ket{n}_i \mapsto \ket{n\pm1}_i$ is simply reversible
increment/decrement on the $q=\lceil \log_2(M+1)\rceil$ qubits, while the corresponding
matrix elements are given by the efficiently computable functions
$n\mapsto \sqrt{ n/M}$ and $n\mapsto \sqrt{(n+1)/M}$. Using reversible arithmetic to compute
these coefficients to accuracy $\varepsilon$ and loading them into an ancilla rotation
yields a block-encoding of the subnormalized operators
$a_i^{\le M}/\sqrt M$ and $(a_i^{\le M})^\dagger/\sqrt M$ with operator-norm error at
most $\varepsilon$, and with \[\mathcal{O}\left(\operatorname{poly}(\log(M),\log(1/\eps))\right)\] circuit depth.
\end{remark}

\begin{remark}[Circuit implementation of $e^{itH_{\le M}}$] 
For an $m$-mode continuous variable system with Hilbert space $(L^2(\R))^{\otimes m} \cong L^2(\R^m),$ we usually consider the projection onto the system register given by $P_M = \pi^{\otimes m}_M$ with local projections $\pi_M = \sum_{n=0}^{M} \kb{n}$. Due to the tensor product structure of $P_M$, we see for $H$ being an $\mathcal{O}(1)$ degree polynomial in the creation and annihilation operators $a_i$ and $a^\dagger_i,$ that also the truncated Hamiltonian $H_{\le M} = P_MHP_M$ is a sum of local terms. Encoding each mode separately uses $q=m\ceil{\log_2(M+1)}$ qubits, and each term acts on $\mathcal O(\log(M+1))$ qubits. Given a term decomposition or block encoding, QSVT techniques provide a circuit implementation of $e^{itH_{\le M}}$ with resource bounds depending on $m$, $M$, the Hamiltonian normalization, and the cost of the supplied coherent access, in addition to $t$ and the target accuracy
\label{rem:CircuitTruncHam}
 \cite{Low_2017,low2017hamiltoniansimulationuniformspectral,Low2019hamiltonian,PhysRevX.6.041067,10.1145/3313276.3316366}.
\end{remark}

\subsubsection[Circuit implementation for integrable filters]
{Circuit implementation for
{ integrable} filter functions}
 Combining Proposition~\ref{prop:DiscreteGauss} and   Corollary~\ref{thm:Lin}  with Theorem~\ref{thm:IntegrableGibbsDynamics} yields the following result on efficient implementation of the Lindblad dynamics for 
integrable filters satisfying the weighted-moment and Gevrey assumptions . Recall that the
image of the finite rank projection $P_M$ involved in the truncation of the Hamiltonian in Section~\ref{sec:FinitDimGenerator}, i.e. $H_{\le M} = P_MHP_M,$ is referred to as the system register, which governs the size of the quantum computer on which can provide an circuit implementation of the Lindblad dynamics. To make the following results more explicit we consider the scaling   $\log(\dim(\operatorname{im}(P_M))) = \mathcal{O}(|\cA|\log(M))$  which is the typical scaling required for quantum many-body or multi-mode continuous variable systems.We note, however, that   Theorem~\ref{thm:CircuitDynamicsIntegrable}
and Corollary~\ref{cor:GibbsCircuitIntegrable}  remain valid for
more general projections $P_M$, provided the remaining
assumptions are satisfied. The corresponding number of   system qubits, in addition to the ancillas in   Corollary~\ref{thm:Lin} , is then 
$
\left\lceil\log_2\left( \dim\bigl(\operatorname{im}(P_M)\bigr)\right)\right\rceil.
$

\begin{theorem}[Circuit implementation of $e^{t\cL_{\sigma_E,\widehat f,H}}$ for { integrable filters}]
\label{thm:CircuitDynamicsIntegrable}
Let $\beta>0,$ $t\ge 0$ and 
$f\in L^1(\mathbb R)$ .
 Under the assumptions of Theorem~\ref{thm:IntegrableGibbsDynamics},~   Corollary~\ref{thm:Lin}  and Proposition~\ref{prop:DiscreteGauss}, with the coherent oracle access specified in Corollary~\ref{thm:Lin}, where 
  \begin{align}
\label{eq:MDimension}
 M =\Theta\left(\operatorname{poly}\left(\log\left(2+\frac{t\, \mathfrak{c}\,E_{\operatorname{Gibbs}}|\cA|}{  \eps}\right)\right)\right)
    \end{align} 
    and assume oracle access to a state preparation circuit for 
$\rho_M$ satisfying 
\begin{equation}\label{eq:registerInputSurgical}
\rho_M:=\frac{P_M \rho P_M}{\Tr(P_M\rho)},\qquad \|P_M\rho P_M-\rho\|_1\le\varepsilon/8,\quad \text{for some state $\rho$ satisfying }\quad
        \rho \le \mathfrak{c}\, \sigma_\beta.
    \end{equation}
Then for $\sigma_E\in(0,\infty)$ the state $e^{t\cL_{\sigma_E,\widehat f, H}}(\rho)$ can be prepared within $\eps$-trace distance on a quantum computer with   $\mathcal{O}\left(|\cA|\log(M+1)\right)$ system qubits, together with the ancillas specified in   Corollary~\ref{thm:Lin} , with  total Hamiltonian simulation time 
      \begin{align*}
\widetilde{\mathcal{O}}\left( \,(1+t)\,\operatorname{poly}\left(|\mathcal{A}|\,,\, \log\left(2+\frac{\mathfrak{c}\,E_{\operatorname{Gibbs}}}{\epsilon}\right)\right)\right).
    \end{align*}  
Here, $\mathcal{O}$, $\widetilde{\mathcal{O}}$ and $\Theta$ treat all non-displayed parameters as constant and $\widetilde{\mathcal{O}}$ additionally supresses subdominant polylogarithmic factors in $t.$
\end{theorem}
\begin{remark}
To illustrate the above theorem, let us consider an $m$-mode continuous variable system on the Hilbert space $(L^2(\R))^{\otimes m} \cong L^2(\R^m),$ choice of bare jumps being $\{A^\alpha\}_{\alpha\in\cA} \equiv \{a_i,a^\dagger_i\}_{i=1}^m$ and with local truncations in the Fock basis $a^{\le M}_i$ and   $(a^{\le M}_i)^\dagger$  being defined in \eqref{eq:DefTruncAnnCrea}. In this case all assumptions on the jump operators in the above theorem are naturally satisfied for the choice $\NA = N_{\operatorname{tot}} = \sum_{i=1}^m a^\dagger_ia_i$ at the relevant places in Section~\ref{sec:TruncateJumps}.
Furthermore, as noted in Remark~\ref{rem:CircuitTruncJumps}, the truncated jumps $a^{\le M}_i$ and   $(a^{\le M}_i)^\dagger$  can be implemented efficiently by a quantum circuit.

Moreover, for   bounded-degree  $H$   satisfying \eqref{eq:ExpBound}, the truncation assumptions are verified  in Section~\ref{sec:FinitDimGenerator}. As seen in Remark~\ref{rem:CircuitTruncHam}, Hamiltonian simulation with respect to  $H_{\le M}$ can also be efficiently implemented on a quantum circuit.
\end{remark}

\begin{proof}[Proof of Theorem~\ref{thm:CircuitDynamicsIntegrable}]
By Theorem~\ref{thm:IntegrableGibbsDynamics} we know for $\sigma_E\in(0,\infty)$ fixed that
  \begin{align*}
\left\|\left(e^{t\mathcal{L}_{\sigma_E,\widehat{f},H}}-e^{t\mathcal{L}^{\le M}_{\sigma_E,\widehat{f},H_{\le M}}}\right)(\rho)\right\|_1  \le \eps/4 
\end{align*} 
for some $M$ satisfying \eqref{eq:MDimension} and the input-tail
condition \eqref{eq:registerInputSurgical}. The latter implies
$\Tr(P_M\rho)>0$ (we may take $0<\varepsilon<1$) and
\[
\|\rho_M-\rho\|_1\le
\|P_M\rho P_M-\rho\|_1+1-\Tr(P_M\rho)\le\varepsilon/4.
\]
Thus input normalization consumes at most one quarter of the error budget.

Furthermore, by Proposition~\ref{prop:DiscreteGauss} we have 
  \begin{align*}
\Big\|e^{t\mathcal{L}^{\le M}_{\sigma_E,\widehat{f},H_{\le M}}}-e^{t\mathcal{L}^{\le M,\mathfrak{n}}_{\sigma_E,\widehat{f},H_{\le M}}}\Big\|_{1\to 1}\le \epsilon/4
\end{align*} 
for a sufficient quadrature count
\begin{align*}
\mathfrak n
&=\widetilde{\mathcal O}\left(\operatorname{poly}(|\cA|,M)+\log\left(2+\frac{|\cA|t}{\eps}\right)\right)\\
&=\widetilde{\mathcal O}\left(\operatorname{poly}\left(|\cA|,\log\left(2+\frac{t\,\mathfrak c\,E_{\operatorname{Gibbs}}|\cA|}{\eps}\right)\right)\right).
\end{align*}
Corollary~\ref{thm:Lin}, used with simulation error $\varepsilon/4$, the normalized-input estimate above, and the triangle inequality complete the four $\varepsilon/4$ contributions to the error budget.

\end{proof}
As a direct consequence of Theorem~\ref{thm:CircuitDynamicsIntegrable} and Corollary~\ref{cor:IntegrableGibbs}, we find the following result on efficient Gibbs state preparation under the assumptions of a simple kernel and a positive spectral gap.
\begin{corollary}[Gibbs state preparation for { integrable} filter functions]
\label{cor:GibbsCircuitIntegrable}

Under the same assumptions as in Theorem~\ref{thm:CircuitDynamicsIntegrable} and assuming additionally that  $L_{\sigma_E,\widehat{f},H},$ the self-adjoint generator on $\mathscr{T}_2(\cH)$ associated to the Lindbladian $\cL_{\sigma_E,\widehat f,H},$ has a positive spectral gap 
$\lambda_2 \equiv \operatorname{gap}\left(L_{\sigma_E,\widehat f,H}\right) >0$ and $\ker L_{\sigma_E,\widehat f,H}=\mathbb C\sqrt{\sigma_\beta}$,
the Gibbs state of the Hamiltonian can be prepared within $\eps$-trace distance on a quantum computer with   $\mathcal{O}(|\cA|\log(M+1))$ system qubits, together with the ancillas specified in   Corollary~\ref{thm:Lin} , for some 
\begin{align*}    
M =\Theta\left(\operatorname{poly}\left(\log\left(2+\frac{\mathfrak{c}\,E_{\operatorname{Gibbs}}|\cA|}{\lambda_2 \eps}\right)\right)\right), 
\end{align*} 
with total Hamiltonian simulation time 
  \begin{align*}
\widetilde{\mathcal{O}}\left(\left(1+\frac{1}{\lambda_2}\right)\operatorname{poly}\left(|\mathcal{A}|\,,\,\log\left(2+\frac{\mathfrak{c}E_{\operatorname{Gibbs}}}{\epsilon}\right)\right)\right).
    \end{align*}  
Here, $\mathcal{O}$, $\widetilde{\mathcal{O}}$ and $\Theta$ treat all non-displayed parameters as constant and $\widetilde{\mathcal{O}}$ additionally supresses subdominant polylogarithmic factors in $1/\lambda_2.$
\end{corollary}

For each fixed $\eta>0$, Lemmas~\ref{lem:PolynomialFilterTime}
and~\ref{lem:PolynomialFilterGevrey} verify all filter assumptions
in these implementation results for $\widehat f_\eta$ defined
in~\eqref{eq:PolynomialFilter}. In the one-mode setting  this together with Corollary~\ref{cor:SharpPolynomialFilterGaussian} gives for the choice  $H=h(N)$ and bare
jumps
$a,a^\dagger$ the stated Gibbs-state preparation
bounds whenever $M_\eta(h)<\infty$, with $M_\eta(h)$ defined in \eqref{eq:PolynomialSpacingCriterion}, and the Hamiltonian, energy,
truncation and oracle assumptions above hold.

\appendix

\section{\Cref{eq:condAalphas} for Schr\"odinger operators}
\label{appendixC}
In this section we prove Theorem \ref{thm:Schroedinger_operators}. To verify Condition \theconditionA \ for Schr\"odinger operators, we start by defining the generalized Sobolev spaces
\begin{definition}
Let $s,\sigma \ge 0$. We define
\[
\mathcal{H}^{s,\sigma}(\mathbb{R}^d)
:=
\Bigl\{
f\in L^2(\mathbb{R}^d)
:\;
\langle x\rangle^\sigma f \in L^2(\mathbb{R}^d),
\quad
\langle \xi\rangle^s \widehat{f}(\xi)\in L^2(\mathbb{R}^d)
\Bigr\},
\]
where $\langle x\rangle := (1+|x|^2)^{1/2}$ and $\widehat{f}$ denotes the Fourier transform of $f$.
\end{definition}

  \begingroup
Only the oscillator square-root domain is needed here. For
$H_0=-\Delta+|x|^2$, its closed nonnegative quadratic form is
\[
q_0[u]=\|\nabla u\|_2^2+\||x|u\|_2^2,
\qquad
D(q_0)=H^1(\mathbb R^d)\cap\{u:|x|u\in L^2\}
=\mathcal H^{1,1}(\mathbb R^d).
\] 
  The representation theorem for closed forms therefore gives
$D(H_0^{1/2})=\mathcal H^{1,1}$ with equivalent norms
\mbox{%
\cite[Chapter~10]{Schmdgen2012}}\hskip0pt%
.
For later use, the standard fractional-power interpolation identity
\mbox{%
\cite[Chapter~4]{Lunardi2018} }\hskip0pt%
gives
\[
[L^2,D(I+A)]_\theta=D((I+A)^\theta),\qquad 0<\theta<1,
\] 
with   graph norms whenever $A\ge0$ is self-adjoint. Thus no general
interpolation identity for intersections is needed. We next recall the
operator domain of $H_0$.
\par\endgroup
 \begin{lemma}[Domain of the quantum harmonic oscillator]
Let
\[
H_0:=-\Delta+|x|^2
\]
initially on Schwartz space $\mathcal S(\mathbb R^d)\subset L^2(\mathbb R^d)$. Then $H_0$
is essentially self-adjoint, and its self-adjoint realization satisfies
\[
D(H_0)
=
\{u\in H^2(\mathbb R^d): |x|^2u\in L^2(\mathbb R^d)\}.
\]
Equivalently,
\[
D(H_0)
=
\{u\in L^2(\mathbb R^d): -\Delta u\in L^2(\mathbb R^d),\ |x|^2u\in L^2(\mathbb R^d)\}.
\]
Moreover, the graph norm of $H_0$ is equivalent to
\[
u\mapsto \|u\|+\|\Delta u\|+\||x|^2u\|.
\]
\end{lemma}

\begin{proof}
For $u\in \mathcal S(\mathbb R^d)$, the harmonic oscillator is symmetric and   essentially self-adjoint by the Hermite basis . We compute
\[
\|H_0u\|^2
=
\|-\Delta u+|x|^2u\|^2
=
\|\Delta u\|^2+\||x|^2u\|^2-2\Re\langle \Delta u, |x|^2u\rangle.
\]
By integration by parts,
\[
-\langle \Delta u, |x|^2u\rangle
=
\sum_{j=1}^d \langle \partial_j u,\partial_j(|x|^2u)\rangle
=
\sum_{j=1}^d \langle \partial_j u,2x_ju+|x|^2\partial_j u\rangle.
\]
Taking real parts gives
\[
-\Re\langle \Delta u, |x|^2u\rangle
=
\||x|\nabla u\|^2
+
2\Re\sum_{j=1}^d \langle \partial_j u,x_ju\rangle.
\]
Now
\[
2\Re\langle \partial_j u,x_ju\rangle
=
\int_{\mathbb R^d} x_j\,\partial_j(|u|^2)\,dx
=
-\int_{\mathbb R^d}|u|^2\,dx
=
-\|u\|^2,
\]
hence
\[
2\Re\sum_{j=1}^d \langle \partial_j u,x_ju\rangle
=
-d\|u\|^2.
\]
Therefore
\[
-\Re\langle \Delta u, |x|^2u\rangle
=
\||x|\nabla u\|^2-d\|u\|^2,
\]
and so
\[
\|H_0u\|^2
=
\|\Delta u\|^2+\||x|^2u\|^2+2\||x|\nabla u\|^2-2d\|u\|^2.
\]
In particular,
\[
\|\Delta u\|^2+\||x|^2u\|^2
\le
\|H_0u\|^2+2d\|u\|^2.
\]
Since also
\[
\|H_0u\|\le \|\Delta u\|+\||x|^2u\|,
\]
we obtain 
\[
\|u\|+\|H_0u\|
\text{ is equivalent to }
\|u\|+\|\Delta u\|+\||x|^2u\|
\text{ for }u\in\mathcal S(\mathbb R^d).
\]

Now let $H_0$ denote the closure of the operator on $\mathcal S(\mathbb R^d)$.
By the graph norm equivalence, $u\in D(H_0)$ if and only if there exists a
sequence $u_n\in\mathcal S(\mathbb R^d)$ such that
\[
u_n\to u,\qquad \Delta u_n\to v,\qquad |x|^2u_n\to w
\quad\text{in }L^2(\mathbb R^d)
\]
for some $v,w\in L^2(\mathbb R^d)$. Passing to distributions shows that
$v=\Delta u$ and $w=|x|^2u$. Thus
\[
D(H_0)\subset
\{u\in L^2(\mathbb R^d): \Delta u\in L^2(\mathbb R^d),\ |x|^2u\in L^2(\mathbb R^d)\}.
\]

Conversely, if $u\in L^2(\mathbb R^d)$ satisfies
\[
\Delta u\in L^2(\mathbb R^d),\qquad |x|^2u\in L^2(\mathbb R^d),
\]
then $u\in H^2(\mathbb R^d)$.  Cutoff and mollification give a Schwartz graph-norm approximation.  Choosing $u_n\in \mathcal S(\mathbb R^d)$ with
\[
u_n\to u,\qquad \Delta u_n\to \Delta u,\qquad |x|^2u_n\to |x|^2u
\quad\text{in }L^2(\mathbb R^d),
\]
we obtain
\[
H_0u_n=-\Delta u_n+|x|^2u_n\to -\Delta u+|x|^2u
\quad\text{in }L^2(\mathbb R^d).
\]
Hence $u\in D(H_0)$. Therefore
\[
D(H_0)=
\{u\in L^2(\mathbb R^d): \Delta u\in L^2(\mathbb R^d),\ |x|^2u\in L^2(\mathbb R^d)\}.
\]
The minimal operator $H_0$ is equal to the maximal operator, thus $H_0$ is the self-adjoint realization. 
Finally, since $u,\Delta u\in L^2$ implies $u\in H^2(\mathbb R^d)$, this is
equivalent to
\[
D(H_0)=
\{u\in H^2(\mathbb R^d): |x|^2u\in L^2(\mathbb R^d)\}.
\]

\end{proof}

It is well-known that 
  \[ V \in L^{p_d}(\mathbb R^d)+ \langle x \rangle^{\alpha} L^{\infty}(\mathbb R^d)\text{ with } \alpha<2\] 
is relatively zero-bounded by \cite[Theorem X.15]{reed1975ii} and \cite[Theorem 4.28]{Zworski2012} with respect to the harmonic oscillator $H_0$.
 Here $p_d=\max\{2,d/2\}$ for $d\ne4$ and $p_4>2$ for $d=4$; pair Coulomb terms are covered by Hardy's inequality.  Thus, Kato-Rellich's theorem \cite[Theorem X.12]{reed1975ii} shows the self-adjointness of $H_0+V$ on $\mathcal H^{2,2}(\mathbb R^d).$
 \begingroup
After a positive shift, $\widetilde H=H_0+V+\lambda\ge1$ and $I+H_0$
have the same operator domain with equivalent graph norms. Interpolation
at $\theta=1/2$ consequently gives
\[
D(\widetilde H^{1/2})=D((I+H_0)^{1/2})=\mathcal H^{1,1}.
\]
The same argument applies to $-\Delta+\nu|x|^2+V$ for every $\nu>0$.
The creation and annihilation operators map
$\mathcal H^{1,1}$ continuously to $L^2$ and $\mathcal H^{2,2}$
continuously to $\mathcal H^{1,1}$; the mixed terms are controlled by
the oscillator graph-norm estimate above. This verifies the second
case of Theorem~\mbox{%
\ref{thm:Schroedinger_operators}}\hskip0pt%
, with
$\gamma=1/2$ and $\mu=1$.
\par\endgroup

\begin{proposition}
Let
\[
H=-\Delta+V(x)
\quad\text{on }L^2(\mathbb R^d),
\]
where $V\in C^\infty(\mathbb R^d)$ is real-valued and satisfies
\[
V(x)\ge c\langle x\rangle^r-C_0,
\qquad c>0,\quad r\ge 1,
\]
as well as
\[
|\partial_x^\alpha V(x)|\le C_\alpha \langle x\rangle^{r-|\alpha|},
\qquad \alpha\in\mathbb N^d.
\]
For $h_0 \ge C_0$, we have 
\[
\widetilde H:=H+(h_0+1)I.
\]
Let $a_j,a_j^{\dagger}$ be the annihilation and creation operators.
Then for every   $n\in\mathbb N_0$ ,
\[
\widetilde H^n a_j \widetilde H^{-n-1},\widetilde H^n a_j^{\dagger} \widetilde H^{-n-1}\in \mathcal B(L^2(\mathbb R^d)).
\]
\end{proposition}

\begin{proof}
We work in the scattering calculus $\Psi^{m,l}_{\mathrm{sc}}(\mathbb R^d)$. Recall \cite[(1.3)]{HassellJiaSussman2025} that a smooth function $a(x,\xi)\in S^{m,l}_{\mathrm{sc}}$ if
\[
|D_x^\alpha D_\xi^\beta a(x,\xi)|
\le C_{\alpha\beta}\langle x\rangle^{l-|\alpha|}\langle \xi\rangle^{m-|\beta|}
\]
for all multi-indices $\alpha,\beta.$
Its quantization \cite[(1.1)]{HassellJiaSussman2025} defines $\Psi^{m,l}_{\mathrm{sc}}$.

We first show that
\[
\widetilde H\in \Psi^{2,r}_{\mathrm{sc}}.
\]
Indeed, its symbol is
  \[
\widetilde{p}(x,\xi)=|\xi|^2+V(x)+h_0+1.
\] 
The term $|\xi|^2$ satisfies
\[
|D_x^\alpha D_\xi^\beta |\xi|^2|
\le C_{\alpha\beta}\langle \xi\rangle^{2-|\beta|}
\le C_{\alpha\beta}\langle x\rangle^{r-|\alpha|}\langle \xi\rangle^{2-|\beta|},
\]
  For $\alpha\ne0$ the differentiated $|\xi|^2$ term vanishes, and for $V(x)+h_0+1$  we use the assumption on derivatives of $V$. 
By the composition property of the scattering calculus \cite[(2.1)]{HassellJiaSussman2025},
\[
\Psi^{m_1,l_1}_{\mathrm{sc}}\circ \Psi^{m_2,l_2}_{\mathrm{sc}}
\subset \Psi^{m_1+m_2,\;l_1+l_2}_{\mathrm{sc}}.
\]
Since $\widetilde H\in \Psi^{2,r}_{\mathrm{sc}}$, an induction gives
\begin{equation}\label{eq:Hn}
\widetilde H^n\in \Psi^{2n,rn}_{\mathrm{sc}}
\qquad\text{for all }n\in\mathbb N.
\end{equation}

  For the inverse estimate use the elliptic sum weight $m(x,\xi)=\langle\xi\rangle^2+\langle x\rangle^r$. Standard elliptic symbolic calculus gives
\begin{equation}\label{eq:Hneg}
\widetilde H^{-n-1}\in\operatorname{Op}S(m^{-n-1}),
\end{equation} 
 and $\widetilde H^n\in\operatorname{Op}S(m^n)$.

The operator
\[
a_j=\frac{1}{\sqrt2}(x_j+\partial_{x_j}) \text{ and }a_j^{\dagger}=\frac{1}{\sqrt2}(x_j-\partial_{x_j})
\]
have symbols
\[
b_j(x,\xi)=\frac{1}{\sqrt2}(x_j+i\xi_j) \text{ and }\overline{b_j}(x,\xi)=\frac{1}{\sqrt2}(x_j-i\xi_j).
\]
A direct inspection   gives $a_j,a_j^\dagger\in\operatorname{Op}S(\ell)$ with $\ell=\langle x\rangle+\langle\xi\rangle$. Hence
\[
\widetilde H^n a_j\widetilde H^{-n-1},\ \widetilde H^n a_j^\dagger\widetilde H^{-n-1}
\in\operatorname{Op}S(\ell/m)\subset\Psi^{0,0}_{\mathrm{sc}}.
\] 
  Since $\ell/m$ is bounded for $r\ge1$, the preceding inclusion in $\Psi^{0,0}_{\mathrm{sc}}$ follows directly from the symbol derivative bounds. The  standard boundedness
result for the scattering calculus \cite[Prop. 3.6]{HassellJiaSussman2025} implies that every operator in
$\Psi^{m,l}_{\mathrm{sc}}$ with $m\le 0$ and $l\le 0$ is bounded on $L^2(\mathbb R^d)$.
Therefore
\[
\widetilde H^n a_j \widetilde H^{-n-1},\widetilde H^n a_j^{\dagger} \widetilde H^{-n-1}\in \mathcal B(L^2(\mathbb R^d)).
\]

\end{proof}

\end{document}